\documentclass[sigconf,nonacm]{acmart}

\usepackage{natbib}
\usepackage[utf8]{inputenc} %
\usepackage[T1]{fontenc}    %
\usepackage{hyperref}       %
\usepackage{url}            %
\usepackage{booktabs}       %
\usepackage{amsfonts}       %
\usepackage{nicefrac}       %
\usepackage{microtype}      %
\usepackage{amsmath,csquotes,xcolor,array,multicol,amsthm,thmtools,amsbsy}
\usepackage{enumitem,xparse}
\usepackage{footmisc}
\usepackage{multirow}
\usepackage{thm-restate}
\usepackage{balance}

\usepackage[ruled,vlined, algo2e]{algorithm2e} %

\usepackage{subfigure}
\usepackage{graphicx}
\graphicspath{{plots/}}
\newlength{\smallfigwidth}
\newlength{\smallfigheight}
\newlength{\smallfigsep}
\newlength{\legendheight}
\usepackage{wrapfig}

\usepackage{enumitem}
\setlist{nolistsep}
\setlist[itemize]{leftmargin=*}
\setlist[enumerate]{leftmargin=*}

\usepackage{marginnote}

\usepackage{microtype}
\usepackage{verbatim} %
\usepackage{hyperref}

\usepackage{xspace}

\usepackage{algorithm}
\usepackage{algpseudocode}
\algtext*{EndWhile}%
\algtext*{EndFor}%
\algtext*{EndProcedure}%
\algtext*{EndIf}%
\algnewcommand\algorithmicinput{\textbf{Input:}}
\algnewcommand\Input{\item[\algorithmicinput]}
\algnewcommand\algorithmicoutput{\textbf{Output:}}
\algnewcommand\Output{\item[\algorithmicoutput]}

\newtheorem{definition}{Definition}
\newtheorem{problem}[definition]{Problem}
\newtheorem{lemma}[definition]{Lemma}

\newtheorem{proposition}[definition]{Proposition}
\newtheorem{corollary}[definition]{Corollary}
\newtheorem{theorem}[definition]{Theorem}

\newtheorem{claim}[definition]{Claim}

\newcommand{\para}[1]{\noindent{\textbf{#1}}}
\newcommand{\spara}[1]{\smallskip\noindent{\textbf{#1}}}
\newcommand{\epara}[1]{\smallskip\noindent{\emph{#1}}}

\newcommand{\abs}[1]{\left|#1\right|}
\newcommand{\sgn}{\operatorname{sgn}}
\renewcommand{\lg}{\log}
\renewcommand{\varepsilon}{\epsilon}
\renewcommand{\vartheta}{\theta}

\newcommand{\poly}{\operatorname{poly}}
\newcommand{\polylg}{\operatorname{poly\,log}}

\newcommand{\bigO}{\ensuremath{\mathcal{O}}\xspace}
\newcommand{\tO}{\ensuremath{\widetilde{\bigO}}\xspace}

\newcommand{\reals}{\ensuremath{\mathbb{R}}\xspace}

\DeclareMathOperator*{\argmin}{arg\,min}
\newcommand{\prox}{\operatorname{prox}}

\newcommand{\ProjQ}[1]{\operatorname{Proj}_Q\!\left(#1\right)}

\newcommand{\AXC}{\m+A_{\m+X}}
\newcommand{\LXC}{\m+L_{\m+X}}
\newcommand{\LXCp}{\m+L_{\m+X'}}
\newcommand{\LXCO}{\m+L_{\m+X_1}}
\newcommand{\LXCT}{\m+L_{\m+X_2}}
\newcommand{\zXC}{\v+z_{\m+X}}
\newcommand{\tzXC}{\widetilde{\v+z}_{\m+X}}
\newcommand{\tzXCit}[1]{\widetilde{\v+z}_{\m+X^{(#1)}}}
\newcommand{\tzXCT}{\tzXCit{T}}
\newcommand{\zXCO}{\v+z_{\m+X_1}}
\newcommand{\zXCT}{\v+z_{\m+X_2}}
\newcommand{\zXCTT}{\v+z_{\m+X^{(T)}}}

\newcommand{\Solve}{\operatorname{Solve}}

\newcommand{\tfidf}{{\texttt{tf}-\texttt{idf}}}

\newcommand{\ouralgo}{{GDPM}\xspace}
\newcommand{\cjl}{{Convex.jl}\xspace}
\newcommand{\blone}{{BL-1}\xspace}
\newcommand{\bltwo}{{BL-2}\xspace}

\DeclareRobustCommand{\ALG}{%
	\ifmmode
		\operatorname{ALG}
	\else
		\text{ALG}\xspace
	\fi
}
\DeclareRobustCommand{\OFF}{%
	\ifmmode
		\operatorname{OFF}
	\else
		\text{OFF}\xspace
	\fi
}
\DeclareRobustCommand{\APPROXALGO}{%
	\ifmmode
		\operatorname{APPROX}
	\else
		\text{APPROX}\xspace
	\fi
}

\newcommand{\twitter}{\ensuremath{\mathbb{X}}\xspace}
\newcommand{\TwitterFive}{\textsf{\small \twitter{-Small}}\xspace}
\newcommand{\TwitterFifty}{\textsf{\small \twitter{-Large}}\xspace}

\newcommand{\Erdos }{\textsf{\small Erdos992}\xspace}
\newcommand{\Advogato }{\textsf{\small Advogato}\xspace}
\newcommand{\PagesGovernment }{\textsf{\small PagesGovernment}\xspace}
\newcommand{\WikiElec }{\textsf{\small WikiElec}\xspace}
\newcommand{\HepPh }{\textsf{\small HepPh}\xspace}
\newcommand{\Anybeat }{\textsf{\small Anybeat}\xspace}
\newcommand{\PagesCompany }{\textsf{\small PagesCompany}\xspace}
\newcommand{\AstroPh }{\textsf{\small AstroPh}\xspace}
\newcommand{\CondMat }{\textsf{\small CondMat}\xspace}
\newcommand{\Gplus }{\textsf{\small Gplus}\xspace}
\newcommand{\Brightkite }{\textsf{\small Brightkite}\xspace}
\newcommand{\Themarker }{\textsf{\small Themarker}\xspace}
\newcommand{\Slashdot }{\textsf{\small Slashdot}\xspace}
\newcommand{\BlogCatalog }{\textsf{\small BlogCatalog}\xspace}
\newcommand{\WikiTalk }{\textsf{\small WikiTalk}\xspace}
\newcommand{\Gowalla }{\textsf{\small Gowalla}\xspace}
\newcommand{\Academia }{\textsf{\small Academia}\xspace}
\newcommand{\GooglePlus }{\textsf{\small GooglePlus}\xspace}
\newcommand{\Citeseer }{\textsf{\small Citeseer}\xspace}
\newcommand{\MathSciNet }{\textsf{\small MathSciNet}\xspace}
\newcommand{\TwitterFollows }{\textsf{\small \twitter-{Follows}}\xspace}
\newcommand{\Delicious }{\textsf{\small Delicious}\xspace}
\newcommand{\YoutubeSnap }{\textsf{\small YoutubeSnap}\xspace}
\newcommand{\Flickr }{\textsf{\small Flickr-und}\xspace}
\newcommand{\Flixster }{\textsf{\small Flixster}\xspace}

\newcommand{\Exact}{\textbf{\small Exact}\xspace}
\newcommand{\Approx}{\textbf{\small Approx}\xspace}
\newcommand{\Error}{\textbf{\small Error}\xspace}
\newcommand{\Uniform}{\textbf{\small Uniform}\xspace}
\newcommand{\Powerlaw}{\textbf{\small Power-law}\xspace}
\newcommand{\Exponential}{\textbf{\small Exponential}\xspace}
\newcommand{\Polarized}{\textbf{\small Polarized}\xspace}

\providecommand{\abs}[1]{\ensuremath{\left\lvert#1\right\rvert}}
\providecommand{\norm}[1]{\ensuremath{\left\lVert#1\right\rVert}}
\newcommand{\diag}{\operatorname{diag}}

\def\v+#1{\ensuremath{\mathbf{#1}}\xspace}
\def\m+#1{\ensuremath{\mathbf{#1}}\xspace}

\newcommand{\LB}{\mathsf{LB}}

\newcommand{\sninline}[1]{}

\usepackage{caption}

\allowdisplaybreaks

\begin{document}

\title[Modeling the Impact of Timeline Algorithms on Opinion Dynamics Using	Low-rank Updates]%
	{Modeling the Impact of Timeline Algorithms on \\ Opinion Dynamics Using	Low-rank Updates}

\author{Tianyi Zhou}
\orcid{0000-0001-9566-8035}
\affiliation{
  \institution{KTH Royal Institute of Technology}
  \city{Stockholm}
  \country{Sweden}}
\email{tzho@kth.se}

\author{Stefan Neumann}\authornote{This work was done while the author was at KTH Royal Institute of Technology.}
\affiliation{
  \institution{TU Wien}
  \city{Vienna}
  \country{Austria}}
\email{stefan.neumann@tuwien.ac.at}

\author{Kiran Garimella}
\affiliation{
  \institution{Rutgers University}
  \city{New Brunswick}
  \country{USA}}
\email{kg766@comminfo.rutgers.edu}

\author{Aristides Gionis}
\affiliation{
  \institution{KTH Royal Institute of Technology}
  \city{Stockholm}
  \country{Sweden}}
\email{argioni@kth.se}

\begin{abstract}
	Timeline algorithms are key parts of online social networks, but during
	recent years they have been blamed for increasing polarization and
	disagreement in our society.  
  Opinion-dynamics models have been used to study a variety of phenomena in online social networks, 
  but an open question remains on how these	models can be augmented 
  to take into account the fine-grained impact of user-level timeline algorithms.
	We make progress on this question by providing a way to model
	the impact of timeline algorithms on opinion dynamics.  Specifically,
	we show how the popular Friedkin--Johnsen opinion-formation model can be
	augmented based on \emph{aggregate information}, extracted from timeline data. 
	We use our model to study the problem of minimizing the polarization and
	disagreement; we assume that we are allowed to make small changes to the users'
	timeline compositions by strengthening some topics of discussion and
	penalizing some others.  We present a gradient descent-based algorithm for
	this problem, and show that under realistic parameter settings, our
	algorithm computes a $(1+\varepsilon)$-approximate solution in
	time~$\tO(m\sqrt{n} \lg(1/\varepsilon))$, where $m$~is the number of edges
	in the graph and $n$~is the number of vertices.
	We also present an algorithm that provably
	computes an $\varepsilon$-approximation of our model in near-linear time.
    We evaluate our method on
	real-world data and show that it effectively reduces the polarization and
	disagreement in the network. 
	Finally, we release an anonymized graph dataset with ground-truth opinions
	and more than 27\,000~nodes (the previously largest publicly available dataset
	contains less than 550~nodes).
\end{abstract}



\maketitle

\section{Introduction}
Online social networks are used by millions of people on a daily basis
and they are integral parts of modern societies. However, during the last decade
there has been growing criticism that timeline algorithms, employed in online social networks,
create filter bubbles and increase the polarization and
disagreement in societies.

Despite significant research effort, our understanding of these phenomena is still
limited.  One of the main challenges is that polarization and
disagreement appear at a \emph{global network-level}, whereas timeline
algorithms operate on a \emph{local user-level}. 
So, on the one hand, opinion dynamics are commonly studied in the context of 
the graph structure of the social network. 
On the other hand, timeline algorithms provide a personalized ranking of content
(such as posts on Facebook or \twitter) and only consider users' local neighborhoods in the graph (e.g., $k$-hop neighborhoods), without considering the global polarization and disagreement. Providing models that bridge the gap between these two levels of abstraction is a major challenge to facilitate our understanding of the underlying phenomena.

A popular way to study the network-level polarization and disagreement is using
opinion-formation models, and one of the most popular abstractions is
the Friedkin--Johnsen~(FJ) model~\cite{friedkin1990social}.
The vanilla version of the FJ model, however, is not sufficient to model real-world online social networks, 
since it assumes that the underlying graph is static, based only on friendship relations, 
and not taking into account \emph{additional} relations and interactions
based on recommendations from timeline algorithms.

To address these limitations a lot of attention has been devoted to augmenting
the FJ~model to understand phenomena that are more closely aligned with the real
world~\cite{zhu2021minimizing,musco2018minimizing,chitra2020analyzing,cinus2023rebalancing,tu2022viral,racz2023towards}.  However, existing augmentations are rather simplistic:
they either study a small number of edge additions or
deletions~\cite{zhu2021minimizing,racz2023towards} or they
directly perform global changes to the graph structure to minimize the
polarization and
disagreement~\cite{musco2018minimizing,chitra2020analyzing,cinus2023rebalancing}
(see Section~\ref{sec:related} for a more detailed description of existing
 approaches).
Most importantly, these papers assume that the graph structure is manipulated
directly, which does not align with how timeline algorithms interact with the
underlying graph structure.  Hence, the augmentations studied in existing papers
provide no way of incorporating the properties of timeline algorithms into
opinion-formation models. They also provide no means of updating a timeline
algorithm's recommendations to reduce polarization and disagreement.

\vspace{2mm}
\para{Our contributions.} In this paper, we make progress on these issues by
introducing an augmentation of the FJ~model that \emph{combines a fixed
underlying graph and a network that is based on aggregate information of a
timeline algorithm}.

In particular, we obtain our \emph{aggregate information} 
by aggregating along the topics that are discussed in the social network. First,
for each user we consider how many posts of their timeline are from each topic.
This provides us with the topic distribution on the timeline of each user.
Second, for each topic we consider how frequently posts by the users are displayed by the timeline algorithm.
This provides us with a distribution for each topic, indicating how influential each user is for this topic.
We argue that this is a realistic way to obtain aggregate
information for a large range of timeline algorithms in real-world platforms, 
e.g., on \twitter or Reddit.

Based on the aggregate information, we introduce a low-rank graph update, which
encodes the social-network connections created by the timeline algorithm's
recommendations. In other words, we use the aggregate information and the
low-rank graph to bridge between the network-level opinion dynamics and the
user-level recommendations of a timeline algorithm.  Our model is the first
that allows to quantify how timeline algorithms impact polarization and
disagreement; we also show that our model can be computed in nearly-linear time.
Details are presented in Section~\ref{sec:problem}.

Next, we use our model to study how a timeline algorithm's recommendations need
to be adapted to reduce polarization and disagreement, by allowing small
changes to the aggregate information.  More concretely, we allow small changes to
the timelines of the users, such as reducing a user's interest in a highly
polarizing topic and slightly strengthening a less controversial topic in the
user's timeline.
Incorporating these types of the changes into real-world timeline algorithms is quite practical.

For the problem of reducing polarization and disagreement, 
we provide a gradient descent-based algorithm, called \ouralgo, 
and show that under realistic parameter settings it computes a
$(1+\varepsilon)$-approximate solution in time
$\tO(m\sqrt{n}\log(1/\varepsilon))$, where $n$~is the number of vertices and $m$
is the number of edges in the original graph.
The details are presented in Section~\ref{sec:gradient-descent}.

To obtain our efficient optimization algorithm, we have to overcome significant
computational challenges.  In particular, since it is possible that the number
of edges introduced by the low-rank graph is much larger than in
the original graph, even writing down the edges introduced by the
recommender system may be infeasible in practice. 
Therefore, in Section~\ref{sec:fast-opinions} 
we show that we can efficiently approximate the opinions, the polarization, and the disagreement 
in time that is \emph{near-linear} in the size of the \emph{original}~graph.

Furthermore, we experimentally evaluate our algorithm on 27~real-world datasets.
Our results show that \ouralgo can efficiently reduce the
\emph{disagreement--polarization index} proposed by~\citet{musco2018minimizing}. 
We also qualitatively evaluate which topics are favored and which topics are
penalized when reducing the polarization and the disagreement.  Additionally,
our experiments show that our algorithms are orders of magnitude faster than
baseline algorithms and that they scale to graphs with millions of nodes and
edges.

Finally, we make our code and two anonymized \twitter datasets available for
research purposes~\cite{Code_data}.
Our anonymized graph datasets contain ground-truth opinions
and the graph structure for more than 27\,000~nodes. The previoulsy largest
publicly available dataset contains less than 550~nodes~\cite{de2019learning}.

We include all omitted proofs in the appendix.
Our code and our data is available online~\cite{Code_data}.

\section{Related work}
\label{sec:related}
Over the past few years, researchers have studied 
the phenomena of political polarization on 
social media~\cite{iyengar2015fear,pariser2011filter}.
The work includes understanding the impact of polarized
discussions~\cite{barber2015causes,levin2021dynamics} as well as developing 
mitigation strategies~\cite{balietti2021reducing,mccarty2015reducing}.

From a practical point of view, there have been various attempts to develop
algorithmic solutions to reduce polarization.  
Several works propose approaches that expose users to opposing viewpoints on 
online social networks~\cite{garimella2017mary,garimella2017reducing,graells2016data}.
\citet{munson2010presenting} design a browser extension that visualizes
the bias of a user's content consumption.

To study polarization in online social networks theoretically, researchers
resorted to opinion formation models and in recent years the most popular model
in this context is the the Friedkin--Johnsen (FJ)
model~\cite{friedkin1990social}.  It has been popular to augment the FJ~model
with abstractions of algorithmic
interventions~\cite{bhalla2021local,chitra2020analyzing,xu2021fast,zhu2021minimizing}.
Other works in this research area study the impact of
adversaries~\cite{chen2021adversarial,gaitonde2020adversarial,tu2023adversaries}
and viral content~\cite{tu2022viral}, as well basic properties of the FJ
model~\cite{bindel2015how}.

Several works in this area dealt with the question of minimizing the
polarization and disagreement using small updates to the underlying
graph~\cite{matakos2017measuring,musco2018minimizing,zhu2021minimizing}.  Zhu et
al.~\cite{zhu2021minimizing} and R{\'{a}}cz and Rigobon~\cite{racz2023towards}
allow $k$~edge updates to the underlying graph.  Musco et
al.~\cite{musco2018minimizing} allow to redistribute all edge weights
arbitrarily, whereas Cinus et al.~\cite{cinus2023rebalancing} allow edge updates
under the constraints that the vertex degrees must stay the same and that no new
edges are added to the graph.  The main limitation of these works is that the
graph updates performed in their algorithms have no clear correspondence with
operations of timeline algorithms; for instance, it is unclear how the
graph updates proposed in~\cite{musco2018minimizing,cinus2023rebalancing} should
be incorporated into a timeline algorithm. 
In contrast, incorporating the changes to the aggregate information that we study in this
paper is feasible in practice. 
We believe that this is a significant contribution to this line of work.

\section{Preliminaries}
\label{sec:preliminaries}

\para{Linear algebra.}
Let $G=(V,E,w)$ be an undirected, connected, weighted graph with $n=\abs{V}$
vertices and $m=\abs{E}$ edges.  We set $\m+L = \m+D-\m+A$ to the
Laplacian of~$G$, where $\m+D$ is the diagonal matrix with
$\m+D_{ii}=\sum_{j\colon(i,j)\in E} w_{ij}$ and~$\m+A$ 
is the weighted adjacency matrix
with $\m+A_{ij} = w_{ij}$.

For $\m+X\in\mathbb{R}^{n\times k}$, we denote the Frobenius norm by
$\norm{\m+X}_F = ({\sum_{i,j} \m+X_{ij}^2 })^{1/2}$.
The spectral norm of~$\m+X$ is 
$\norm{\m+X}_2 = \sigma_{\max}(A)$, where $\sigma_{\max}(\m+A)$ is the largest
singular value of $\m+A$.
We also use the $1$-norm of a matrix, which is given by
$\norm{\m+X}_{1,1}=\sum_{ij}\abs{\m+X_{ij}}$. 
We write $\m+X_i$ to denote the $i$-th row of $\m+X$.

We write $\m+I$ to denote the identity matrix and $\v+1$ to denote the vector
with all entries equal to 1; the dimension will typically be clear from the
context. Given a vector $\v+v$, we write $\diag(\v+v)$ to denote the diagonal
matrix with $\diag(\v+v)_{ii} = \v+v_i$.
For vectors~$\v+u,\v+v\in\mathbb{R}^n$, we 
write $\v+u \odot \v+v \in \mathbb{R}^{n}$ to denote their
Hadamard product, i.e., $(\m+u\odot\m+v)_{i} = \m+u_{i} \m+v_{i}$.
We define $\norm{\v+v}_2=({\sum_{i} \v+v_i^2})^{1/2}$ to be the Euclidean norm of $\v+v$.
For a vector $\v+v$ and a convex set~$Q$, $\ProjQ{\v+v}$ denotes the
orthogonal projection of $\v+v$ onto~$Q$.

We write $\sgn(x) \colon \mathbb{R} \to \{-,0,+\}$ to denote the sign of~$x$.
We use the notation $\tO(T)$ to denote running times of the
form $T  \lg^{\bigO(1)}(n)$. We write $\poly(n)$ to denote numbers bounded
by $n^{\bigO(1)}$.

\spara{Friedkin--Johnsen (FJ) model.}
Let $G$ and $\m+L$ be as defined above.
In the FJ model~\cite{friedkin1990social}, 
each node~$i$ has a fixed \emph{innate opinion} $\v+s_i$ and 
an \emph{expressed opinion}~$\v+z^{(t)}_i$ at time~$t$. 
Initially, $\v+z^{(0)}_i = \v+s_i$, and 
at time $t+1$ every node~$i$ updates their expressed
opinion as the weighted average of its own innate opinion and the expressed opinions of its neighbors:
\begin{align}
\label{eq:update-rule}
	\v+z^{(t+1)}_i
	= \frac{\v+s_i + \sum_{j\colon(i,j)\in E} w_{ij} \v+z^{(t)}_i}{1 + \sum_{j\colon(i,j)\in E} w_{ij}}.
\end{align}
We write $\v+s\in\reals^n$ and $\v+z^{(t)}\in\reals^n$ to denote the vectors of
innate and expressed opinions, respectively.  It is known that in the limit, the
expressed opinions converge to
$\v+z = \lim_{t\to\infty} \v+z^{(t)} = (\m+I + \m+L)^{-1} \v+s$.

We assume that the innate opinions are mean-centered
and in the interval $[-1,1]$, i.e., $\sum_{i\in V} \v+s_i = 0$ and
$\v+s_i\in[-1,1]$ for all $i\in V$.  The latter implies that
$\v+z^{(t)}_i\in[-1,1]$.  We note that these assumptions are made without loss
of generality as they can always be achieved by re\-scaling the opinions~$\v+s$.  

For mean-centered opinions, the \emph{polarization index}~$P(G)$ measures the
variance of the opinions and is given by
\( P(G) = \sum_{i\in V} \v+z_i^2.\) The \emph{disagreement index}~$D(G)$
describes the tension along edges in the network and is given by
\(   D(G) = \sum_{(i,j)\in E} w_{ij} (\v+z_i - \v+z_j)^2\).
Finally, the \emph{disagreement--polarization index}~$I(G)$, on which we will
focus for the rest of the paper, is given by
\begin{align}
\label{eq:dis-pol}
	I(G) = P(G) + D(G) = \v+s^\top (\m+I + \m+L)^{-1} \v+s,
\end{align}
where the last equality was shown by Musco et al.~\cite{musco2018minimizing}.
They also observe that the function
$f(\v+L) = \v+s^\top (\m+I + \m+L)^{-1} \v+s$ is
convex if $\m+L\in\mathcal{L}$ is from a convex set of
Laplacians~$\mathcal{L}$~\cite{nordstrom2011convexity}.

\section{Problem formulation}
\label{sec:problem}

In this section, we formally introduce our augmented version of the FJ~model and
we state the optimization problem we study for minimizing the
disagreement--polarization index.  In our model, we show how aggregate
information from a timeline algorithm can be used to obtain a low-rank graph
update for the FJ~model.  At a high level, we start with the initial adjacency
matrix $\m+A$, which only contains interaction-information (such as who follows
whom), and add an adjacency matrix $\AXC$ based on the aggregate information. 

\begin{table}[t]
  \small
  \centering
  \caption{Notation}
  \label{tab:notation}
  \begin{tabular}{ll}
    \toprule
    Variable & Meaning \\
    \midrule
        $G = (V,E)$ & Original graph, vertex set, edge set \\  
	$n$ & Number of vertices in the original graph \\
	$m$ & Number of edges in the original graph\\
	$k$ & Number of topics \\
	$\m+X$ & User--topic matrix (variable of our algorithm)\\
	$\m+Y$ & Influence--topic matrix (fixed)\\
	$\m+A$ & Adjacency matrix of the original graph \\
	$\AXC$ & Low-rank adjacency matrix based on aggregate information\\
	$\m+L$ & Laplacian of the original graph\\
	$\LXC$ & Laplacian of the low-rank graph\\
	$\v+s$ & Innate opinions\\
	$\v+z$ & Expressed opinions for the original graph\\
	$\zXC$ & Expressed opinions after adding the low-rank \\
		& \quad update to the original graph\\
	$\tzXC$ & Approximation of $\zXC$ \\
	$f(\m+X)$ & Objective function value for $\m+X$ \\
	$\m+X^{(L)}$ & Entry-wise lower bound for $\m+X'$ in Problem~\ref{problem:min-dpi} \\
	$\m+X^{(U)}$ & Entry-wise upper bound for $\m+X'$ in Problem~\ref{problem:min-dpi} \\
  $\vartheta$ & Parameter used to define $\m+X^{(L)}\!$ and $\m+X^{(U)}$  \\
	$Q$ & Feasible set of matrices~$\m+X$ with $\m+X^{(L)}\!\leq \m+X \leq \m+X^{(U)}$ \\
  $C$ & Percentage of extra edge weight added by low-rank update \\
    \bottomrule
  \end{tabular}
\end{table}

The \emph{aggregate information} that we consider is as follows. We consider
$k$~different topics and two row-stochastic matrices $\m+X\in[0,1]^{n\times k}$
and $\m+Y\in[0,1]^{k\times n}$, i.e.,
$\sum_{j=1}^k \m+X_{ij}=1$, for all $i=1,\dots,n$, and 
$\sum_{r=1}^n \m+Y_{jr}=1$, for all $j=1,\dots,k$, and we assume $k\leq n$.
Here, $\m+X$ models how user timelines are formed based on various topics; more
concretely, we assume that $\m+X_{ij}$ is the fraction of
posts in user~$i$'s timeline from topic~$j$.
The matrix~$\m+Y$ models which users are recommended by the timeline algorithm
for each topic; that is, when the algorithm recommends contents for topic~$j$,
then a fraction of $\m+Y_{jr}$ of the contents was composed by user~$r$.  We
believe that it is possible to obtain this type of aggregate data for providers
of online social networks in the real-world, for instance, by monitoring users'
timelines (to get $\m+X$) and the recommendations of timeline algorithms (to get
$\m+Y$).

Observe that if we consider the product $\m+X \m+Y$, 
a $(\m+X \m+Y)_{ij}$-fraction of the recommended contents in the timeline of
user~$i$ is composed by user~$j$. 
This can also be viewed as the impact that a user~$j$ 
has on another user~$i$.
Since in general $\m+X \m+Y$ is
a non-symmetric matrix, we also add the transposed term $\m+Y^\top \m+X^\top$,
which ensures symmetry of the adjacency matrix.
This can be interpreted as the impact of users' audience to them, 
for instance, users want to create content that is liked by their audience.

Thus, we consider a scaled version
of $\m+X\m+Y + \m+Y^T \m+X^T$.  In the following lemma, we
show that this matrix adds (weighted) edges of total weight~$2n$.

\begin{lemma}
\label{lem:weight-edges}
	It holds that $\norm{\m+X\m+Y + \m+Y^T \m+X^T}_{1,1} = 2n$.
\end{lemma}

To obtain a more fine-grained control over how many edges we add to the original
graph, we consider a scaled version of $\m+X\m+Y + \m+Y^T \m+X^T$.  More
concretely, based on the result from Lemma~\ref{lem:weight-edges}, we add the
low-rank adjacency matrix given by
\[ \AXC=\frac{C W}{2n} \left( \m+X \m+Y + \m+Y^T \m+X^T \right), \]
where $C>0$ is a parameter that is fixed throughout the paper and
$W=\sum_{(i,j)\in E}w_{ij}$ is the total weight
of edges in the original graph~$G$. 
Observe that Lemma~\ref{lem:weight-edges} implies that $\norm{\AXC}_{1,1}=CW$
and thus if we add the edges in $\AXC$ to the graph, the total weight of edges
increases by a $C$-fraction.\footnote{
	We note that while here we only guarantee that the \emph{global} increase of
	edges is a $C$-fraction, in Figure~\ref{fig:node-degree-increase} we show
	that also on a \emph{local} user-level, most individual node-degrees are
	increased by approximately a $C$-fraction.
}  In practice, it may be realistic to
think of $C=10\%$ or $C=50\%$.

After adding the edges $\AXC$, which are based on the aggregate information, the
new adjacency matrix becomes
$$\m+A + \AXC = \m+A + \frac{C W}{2n} \left( \m+X \m+Y + \m+Y^T \m+X^T \right),$$
where $\m+A$ is the adjacency matrix of the original graph and $\AXC$ is the
adjacency matrix of the edges that are introduced by the low-rank update.
Next, we write
\begin{align*}
	\LXC = \diag(\AXC\v+1) - \AXC
\end{align*}
to denote the Laplacian associated with the adjacency matrix $\AXC$.
Note that the Laplacian of the combined graph is given by $\m+L + \LXC$, where
$\m+L$ is the Laplacian of the original graph that only contains the
follow-information.

Now, after adding the edges from the low-rank update, the expressed
equilibrium opinions that are produced by the FJ opinion dynamics are given by
$\zXC = (\m+I + \m+L + \LXC)^{-1} \v+s$.

Next, we formally introduce the optimization problem that we study.
Intuitively, the problem states that we wish to minimize the
disagreement--polarization index (Eq.~\eqref{eq:dis-pol}), while allowing
small changes to the aggregate information. In particular, we allow to make
changes to how the users' timelines are composed of different topics.  The
formal definition is as follows.

\begin{problem}
\label{problem:min-dpi}
Given a graph $G=(V,E)$ with adjacency matrix~$\m+A$ and Laplacian~$\m+L$, 
user--topic matrix $\m+X\in[0,1]^{k\times n}$, 
influence--topic matrix $\m+Y\in[0,1]^{k\times n}$, 
and lower and upper bound matrices $\m+X^{(L)}\!$ and $\m+X^{(U)}\!$, respectively,
find a matrix $\m+X' \in [0,1]^{n\times k}$ to satisfy
\begin{equation}
\label{eq:problem}
\begin{aligned}
	\min_{\m+X'} \quad & f(\m+X') = \v+s^T\left(\m+I + \m+L + \LXCp\right)^{-1}\v+s , \\
	\text{such that} \quad & \norm{\m+X_i'}_{1} =1, \quad \text{for all } i=1,\dots,n, \text{ and}\\
	 &\m+X^{(L)}\! \leq \m+X' \leq \m+X^{(U)}.
\end{aligned}
\end{equation}
\end{problem}

In Problem~\ref{problem:min-dpi}, 
we write $\m+X_i'$ to denote the $i$-th row of the matrix-valued
variable $\m+X'$. The first constraint ensures that $\m+X'$ is a row-stochastic
matrix. Furthermore, the matrices $\m+X^{(L)}\!\in[0,1]^{n\times k}$ and
$\m+{X}^{(U)}\in[0,1]^{n\times k}$ are part of the input and they give entry-wise
lower and upper bounds for the entries in $\m+X'$, 
i.e., we require  
$0\leq \m+{X}^{(L)}_{ij}\!\leq \m+{X}_{ij}' \leq \m+{X}^{(U)}_{ij}\! \leq 1$ 
for all $i,j$.  
This constraint can be interpreted as a quantification of how much we can
increase/decrease the attention of user~$i$ to topic~$j$ without 
the risk of making non-relevant recommendations and without violating ethical considerations.
We further assume that $\m+{X}^{(L)}\!\leq \m+X \leq \m+{X}^{(U)}\!$, which corresponds
to the assumption that the initial matrix $\m+X$ is a feasible solution to our
optimization problem.

In the following, we let $Q$ denote the set of all matrices $\m+X'$ that satisfy
the constraints of Problem~\ref{problem:min-dpi}. Observe that $Q$ is a convex
set, since it is the intersection of a box and a hyperplane (the first
constraint is equivalent to the hyperplane constraint $\langle \m+X_i', \v+1
\rangle = 1$, since all entries in $\m+X'$ are in the interval $[0,1]$; the
second constraint is a box constraint). Furthermore, observe that the
constraints are independent across different rows $\m+X_i'$, which we will
exploit later.

Since the objective function and $Q$ are convex, Problem~\ref{problem:min-dpi}
can be solved optimally in polynomial time.  However, if we use a off-the-shelf
solver for this purpose, its running time will be prohibitively high in
practice (see Section~\ref{section:experimental-evaluation}). Even more, already
\emph{a single} computation of the gradient is impractical when done na\"ively
(see Section~\ref{section:experimental-evaluation}).  We address these
challenges in the following section.

\section{Optimization algorithm}
\label{sec:alg}

In this section, we present a gradient-descent algorithm, which converges to an
optimal solution for Problem~\ref{problem:min-dpi}. We present bounds for its
running time and its approximation error after a given number of iterations.  We
also show that we can approximate the expressed opinions $\zXC$ highly
efficiently. We conclude the section by presenting two greedy baseline algorithms.

\subsection{Efficient estimation of expressed opinions}
\label{sec:fast-opinions}

To understand the impact of the low-rank update on the user opinions, it
is highly interesting to inspect the expressed opinions~$\zXC$: 
comparing them with the original expressed opinions~$\v+z$ will offer us
insights into the impact of the timeline algorithm.  However, even though in
Lemma~\ref{lem:weight-edges} we bound the total \emph{weight} of edges that are
added, their \emph{number} could still be $\Omega(n^2)$, 
since the matrix $\AXC$ might be dense. 
Thus, even writing down $\AXC$
would result in running times of $\Omega(n^2)$ and would be prohibitively expensive. 
Therefore, one challenge is to show how to compute $\zXC$ efficiently.

In the following proposition, we show that since $\AXC$ has small rank, we can
exploit the Woodbury identity to obtain an approximation~$\tzXC$ via
Algorithm~\ref{alg:opinions} (see Appendix~\ref{alg:opinions} for the
pseudocode). 
By using such an approximation we can achieve much faster running times, 
while still obtaining provably small errors. 
In the following proposition we use $\m+U =
\begin{pmatrix} \m+X ~\:~ \m+Y^\top \end{pmatrix}$ and $\m+V = \begin{pmatrix} \m+Y
\\ \m+X^\top \end{pmatrix}$.

\begin{proposition}
\label{prop:fast-opinions}
	Let $\varepsilon > 0$.
	Suppose $\left(- \frac{2n}{CW} \m+I + \m+V \m+M^{-1} \m+U\right)^{-1}$
	exists and $\norm{\m+V \m+M^{-1} \m+U}_2 \leq 0.99 \frac{2n}{CW}$.
	Algorithm~\ref{alg:opinions} computes $\tzXC$ with
	$\norm{\tzXC - \zXC}_2 \leq \varepsilon$ in
	expected time $\tO((mk + nk^2 + k^3)\lg(W/\varepsilon))$.
\end{proposition}
\begin{proof}[Proof sketch]
	The algorithm is based on the observation that
	using the Woodbury matrix identity with
	$\m+M = \m+I + \m+L + \diag(\AXC\v+1)$,
	and $\m+U$ and $\m+V$ as before, we get that
	\begin{align*}
		\zXC &= \m+M^{-1} \v+s
			 	+ \frac{CW}{2n}
					\m+M^{-1}
					\m+U
			 		\left(
						\m+I
						- \frac{CW}{2n} \m+V \m+M^{-1} \m+U
					\right)^{-1}
					\m+V
					\m+M^{-1} \v+s.
	\end{align*}
	Now Algorithm~\ref{alg:opinions} (pseudocode in the appendix) basically computes this quantity from
	right to left.
    Our main insight here is that we can compute the quantities $\m+M^{-1}\v+s$ and $\m+M^{-1}\m+U$ using the Laplacian solver from Lemma~\ref{lem:laplacian-solver}. Here, we approximate $\m+M^{-1}\m+U$ column-by-column using the call $\Solve(\m+M,\v+w_j,\varepsilon_{\m+R})$, where $\v+w_j$ is the $j$'th column of $\m+U$ and $\varepsilon_{\m+R}$ is a suitable error parameter. The remaining matrix multiplications are efficient since $\m+U$ has only $2k$~columns and since matrix $\m+V$ has only $2k$~rows.
 
    To obtain our guarantees for the approximation error, we have to perform an intricate error analysis to ensure that errors do not compound too much. This is a challenge since we solve $\m+I - \frac{CW}{2n} \m+V\m+M^{-1}\m+U$ only approximately but then we have to compute an inverse of this approximate quantity. 
In the proposition we used the assumptions that $\m+V \m+M^{-1} \m+U$ exists and
that $\norm{\m+V \m+M^{-1} \m+U}_2 \leq 0.99 \cdot \frac{2n}{CW}$, to ensure that this can be done without obtaining too much error.
\sninline{While in general
this may not be true, we show experimentally that this condition is satisfied in
practice.} \sninline{We have to run these experiments.} %
In the proof we will also show that these assumptions imply that the inverse
$\m+S^{-1}$ used in the algorithm exists.
 See
	Appendix~\ref{sec:proof-prop-fast-opinions} for details.
\end{proof}

\sninline{Check whether we need the assumption that
	$\left(- \frac{2n}{CW} \m+I + \m+V \m+M^{-1} \m+U\right)^{-1}$ exists.
	It probably follows from the assumption that 
	$\norm{\m+V \m+M^{-1} \m+U}_2 \leq 0.99 \frac{2n}{CW}$.}

The input of Algorithm~\ref{alg:opinions} are the innate opinions $\v+s$,
the user--topic matrix $\m+X$, the influence--topic matrix $\m+Y$, 
the fraction of weight parameter $C$, and the approximation error parameter $\varepsilon$. 
The algorithm returns the approximated expressed opinions $\tzXC$.
Note that if we consider the practical scenario of $k=\polylg(n)$ and
$W\leq\poly(n)$, the running time of Algorithm~\ref{alg:opinions} is
$\tO(m\lg(1/\varepsilon))$.

Proposition~\ref{prop:fast-opinions} also allows us to efficiently evaluate the
disagree\-ment--polarization index after adding the edges in $\AXC$. More
concretely, in the following corollary we show that we can efficiently evaluate
our objective function $f(\m+X) = \v+s^T (\m+I + \m+L + \LXC)^{-1} \v+s$ with
small error.
\begin{corollary}
\label{cor:fast-objective-function}
	Let $\varepsilon > 0$.
	Suppose $\left(- \frac{2n}{CW} \m+I + \m+V \m+M^{-1} \m+U\right)^{-1}$
	exists and $\norm{\m+V \m+M^{-1} \m+U}_2 \leq 0.99 \frac{2n}{CW}$.
	We can compute a value $\widetilde{f}$ such that
	$\abs{\widetilde{f} - f(\m+X)} \leq \varepsilon$
	in expected time $\tO((mk + nk^2 + k^3)\lg(W/\varepsilon))$.
\end{corollary}

\subsection{Gradient descent-based polarization minimization}
\label{sec:gradient-descent}

Next, we present our gradient descent-based polarization minimization (\ouralgo) algorithm. 
We start by presenting basic
facts about the gradient of our problem in the following proposition.
\sninline{The way we set the parentheses in the gradient is a bit odd at the
	moment. We should fix this at some point.}
 \begin{proposition}
\label{prop:gradient}
	The following three facts hold for the gradient of $f(\m+X)$ with respect to $\m+X$:
	\begin{enumerate}
		\item The gradient $\nabla_{\m+X} f(\m+X)$ is given by
		\begin{equation}
		\label{eq:gradient}
		\begin{split}
			&\nabla_{\m+X} f(\m+X) = \\
			& \frac{CW}{2n} (2\cdot \zXC \cdot \zXC^\top\cdot \m+Y^\top
			- (\zXC\odot \zXC) \cdot \v+1_k^\top
			- \v+1_n \cdot (\zXC^\top \odot \zXC^\top)\cdot \m+Y^\top ).
		\end{split}
		\end{equation}
		\item The function $f(\m+X)$ is $L$-smooth with
		$L=\frac{8CW}{\sqrt{n}} \cdot \norm{\v+s}_2 \cdot \norm{\m+Y}_2^2$,
		i.e., for all $\m+X_1,\m+X_2 \in Q$ it holds that
		\begin{align*}
			\norm{\nabla_{\m+X} f(\m+X_1) - \nabla_{\m+X} f(\m+X_2)}_F
				\leq \frac{8CW}{\sqrt{n}} \cdot \norm{\v+s}_2 \cdot \norm{\m+Y}_2^2
						\cdot \norm{\m+X_1 - \m+X_2}_F.
		\end{align*}
		\item Let $\varepsilon>0$. Suppose the conditions of
		Proposition~\ref{prop:fast-opinions} hold, then we can compute an
		approximate gradient
			$\widetilde{\nabla}_{\m+X} f(\m+X)$ such that
			$\norm{\widetilde{\nabla}_{\m+X} f(\m+X) - \nabla_{\m+X} f(\m+X)}_F \leq
			\varepsilon$
			in expected time $\tO((mk + nk^2 + k^3)\lg(W/\varepsilon))$.
	\end{enumerate}
\end{proposition}

The gradient of our problem is given in Eq.~\eqref{eq:gradient} and in the
second point we show that it is Lipschitz continuous. Computing the
gradient exactly involves computing $\zXC$ exactly; however, this requires to
compute the matrix inverse $(\m+I + \m+L)^{-1}$, which is expensive for large graphs. 
Hence, in the third point we show that an approximate
gradient can be computed highly efficiently and with error guarantees.

Since we only have an approximate gradient, \ouralgo is an implementation of the
gradient descent method by d'Aspremont~\cite{d2008smooth}, who analyzed a method
of Nesterov~\cite{nesterov1983method} with approximate gradient.
\sninline{Do we have to explain gradient descent here?}
We use Kiwiel's algorithm~\cite{kiwiel2008breakpoint} to
compute the orthogonal projections $\ProjQ{\cdot}$ on our set of feasible
solutions~$Q$ in linear time, where we exploit that our constraints are
independent across different rows of $\m+X$.
The pseudo\-code of \ouralgo is given in Algorithm~\ref{alg:optimization} in
the appendix.

Algorithm~\ref{alg:optimization} takes as input the innate opinions $\v+s$,
the user--topic matrix $\m+X$, the influence--topic matrix $\m+Y$, 
the budget $\vartheta$, and the extra \emph{weight} parameter $C$. 
It returns $\m+X^{(T)}$ after a number of iterations~$T$.

In the following theorem we present error and running-time guarantees for
{\ouralgo}, which show that it converges to the optimal solution given enough iterations.
\begin{theorem}
\label{thm:approximation-guarantee}
	Let $\varepsilon>0$.
	Suppose at each iteration of {\ouralgo} the
	conditions of Proposition~\ref{prop:fast-opinions} are satisfied.
	Then {\ouralgo} computes a solution $\m+X^{(T)}$ such
	that $f(\m+X) - f(\m+X^*) \leq \varepsilon$ in expected time
	\[
	\tO\left(\sqrt{\varepsilon^{-1} \cdot CWkn} \cdot (mk + nk^2 + k^3)\lg(W/\varepsilon)\right),
	\]
	where $\m+X^*$ is the optimal solution for Problem~\ref{problem:min-dpi}.
\end{theorem}

We note that in parameter settings that are realistic in practice, 
\ouralgo computes a solution with
multiplicative error at most $(1+\varepsilon')$ in time
$\tO(m\sqrt{n}\lg(1/\varepsilon'))$.
More concretely, this is the case when the number of topics $k=\polylg(n)$ is small, 
the fraction of additional edges $C=\bigO(1)$ is small,  and 
the network is sparse with $W=\tO(n)$. 
Additionally, it is realistic to assume that the
optimal solution still has a large amount of polarization and disagreement since
at least a constant fraction of the users will differ from the average opinion by at
least $0.01$; 
this argument implies that the polarization is at least $\LB = \Omega(n)$, 
which in turn implies that $f(\m+X^*) \geq \LB = \Omega(n)$. 
Hence, if in the theorem we set $\varepsilon = \varepsilon'\, \LB$,
we get the bound above.

\subsection{Baselines}
\label{sec:baseline}

Next, we introduce two greedy baseline algorithms. The baselines proceed in
iterations and, intuitively, in each iteration they update the user
timelines such that some topics are penalized and others are favored; 
the choice of these topics depends on the baseline.

More concretely, the baselines obtain as input the original graph and the
matrices $\m+X$, $\m+X^{(L)}\!$, $\m+X^{(U)}$, $\m+Y$ and a number~$T_{\max}$ of
iterations to perform. First, we set $\m+X^{(0)}\gets \m+X$. Now the algorithm
performs $T_{\max}$~iterations. In each iteration~$T$, we initialize
$\m+X^{(T)}\gets\m+X^{(T-1)}$. Then we manipulate the timeline of each user~$i$
by redistributing the weights in row~$i$ of $\m+X^{(T)}$.  We pick two
topics~$j$ and $j'$ and transfer as much weight as possible from topic~$j'$ to
topic~$j$.  Intuitively, one can think of $j$ as a topic that we want to
strengthen and $j'$ as a topic that we want to penalize; how these topics are
picked depends on the implementation of the baseline (see below). To denote how
much weight we can transfer, we set
$\delta \gets \min\{\m+X^{(U)}_{ij} - \m+X^{(T)}_{ij},
					\m+X^{(T)}_{ij'} - \m+X^{(L)}_{ij'}\}$,
i.e., $\delta$ corresponds to the weight that we can transfer from topic $j'$
to~$j$ without violating the constraints of Problem~\ref{problem:min-dpi}. Then we
set $\m+X^{(T)}_{ij} \gets \m+X^{(T)}_{ij} + \delta$ and
$\m+X^{(T)}_{ij'} \gets \m+X^{(T)}_{ij'} - \delta$. As stated before, we do this
for each user~$i$. Then the next iteration~$T+1$ starts.

\epara{Baseline 1: Strengthening non-controversial topics} (\blone). 
We introduce our first baseline (\blone), 
which aims to penalize controversial topics and to
strengthen non-controversial topics. We build upon the meta-algorithm above and
state how to pick the topics~$j$ and $j'$ for the current user~$i$.  First, we
compute $\tzXCit{T}$ using Algorithm~\ref{alg:opinions} and set
$\bar{z}=\frac{1}{n}\sum_{u\in V} \tzXCit{T}(u)$ to the average user opinion.
Also, for each topic~$j$ we set $\tau_j = \sum_{u\in V} \m+Y_{ju} \tzXCit{T}(u)$ to
the weighted average of the opinions of influential users for topic~$j$. Since this does
not depend on the user~$i$, this can be done at the beginning of each
iteration~$T$.  In \blone, we set $j$ to a controversial topic that is
``far away'' from the average opinion and $j'$ to a non-controversial topics
which is ``close'' to the average opinion.  More concretely, we let $j$ be the
topic with $\m+X^{(T)}_{ij} < \m+X^{(U)}_{ij}$ that minimizes
$\abs{\tau_j-\bar{z}}$; and we let $j'$ be the topic with $\m+X^{(T)}_{ij'} >
\m+X^{(L)}_{ij'}$ that maximizes $\abs{\tau_{j'}-\bar{z}}$.

\epara{Baseline 2: Strengthening opposing topics} (\bltwo).
Our second baseline (\bltwo) can be viewed as a reverse of the above strategy and is
inspired by the experimental outputs that we observed from \ouralgo:
it penalizes non-controversial topics and strengthens topics that
are opposing to user~$i$'s opinion. More concretely, we compute $\tzXCit{T}$ and
$\bar{z}$ as before. However, then we strengthen the topic~$j$ with
$\m+X'_{ij} < \m+X^{(U)}_{ij}$ that maximizes $-\tzXCit{T}(i) \tau_j$.
For instance, if $\tzXCit{T}(i) > 0$ then the algorithm will pick the topic $\tau_j<0$
of largest absolute value; note that since $\tzXCit{T}$ and $\tau_j$ must have
different signs, this corresponds to connecting user~$i$ to a topic that opposes
its own opinion.  Also, we let $j'$ be the topic with $\m+X'_{ij'} >
\m+X^{(L)}_{ij'}$ and $\tzXCit{T}(i)\tau_j>0$ that minimizes
$\abs{\tau_{j'}-\bar{z}}$; this corresponds to our choice of non-controversial
topics in \blone assuming that~$\tau_j$ has the same sign as $\tzXCit{T}(i)$.

\smallskip
 The pseudo\-code for \blone and \bltwo is presented in 
Algorithm~\ref{alg:baseline} in the appendix.

\section{Experimental evaluation}
\label{section:experimental-evaluation}

We evaluate our algorithms on 27~real-world datasets.  To conduct realistic
experiments, we collect two novel real-world datasets from \twitter, which we
denote \TwitterFive ($n=1\,011$, $m=1\,960$) and \TwitterFifty 
($n=27\,058$, $m=268\,860$). These two datasets contain ground-truth opinions
and we use retweet-information to obtain the aggregate information for the
interest--topic and in\-flu\-ence--topic matrices $\m+X$ and $\m+Y$.
We provide details on how this data was obtained in Section~\ref{sec:exp-setting}.
We make our novel datasets available online~\cite{Code_data} and we will
release them for research purposes upon acceptance of the paper; we note that
\TwitterFifty contains more than 27\,000~nodes and is thus almost 50~times
larger than the previoulsy largest publicly available dataset with ground-truth
opinions (which contains less than 550~nodes)~\cite{de2019learning}.
Appendix~\ref{sec:exp-setting} also provides details for the remaining
25~datasets. 

We experimentally compare \ouralgo  
against the greedy baselines \blone and \bltwo. We also compare our
gradient-descent algorithm against the off-the-shelf solver \cjl which uses the
SCS solver internally.
In our experiments, given $\m+X$ and a parameter $\vartheta \in [0,1]$, we set
$\m+X_{ij}^{(U)}\! = \min\{1, \m+X_{ij} + \vartheta\}$ and
$\m+X_{ij}^{(L)}\! = \max\{0, \m+X_{ij}-\vartheta\}$, when not mentioned
otherwise. We run \ouralgo with learning rate $L=10$ (see
Section~\ref{app:additional-exp} for a justification).

We conduct our experiments on a Linux workstation with a 2.90\,GHz Intel
Core i7-10700 CPU and 32\,GB of RAM.
Our code is written in Julia~v1.7.2 and available online~\cite{Code_data}.

\begin{figure}[t]
    \centering
    \begin{tabular}{cc}
        \includegraphics[width=0.45\columnwidth]{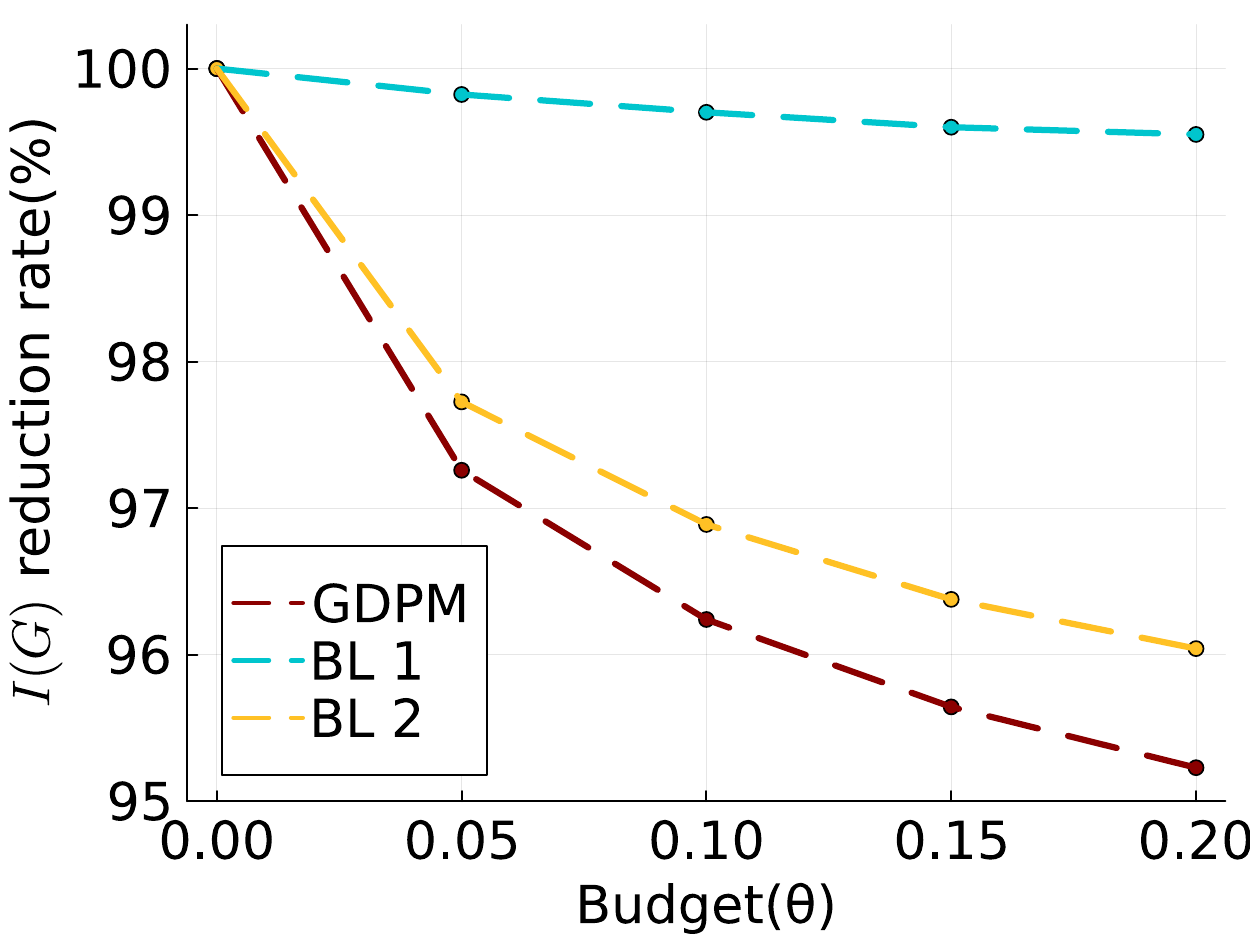} &
	\includegraphics[width=0.45\columnwidth]{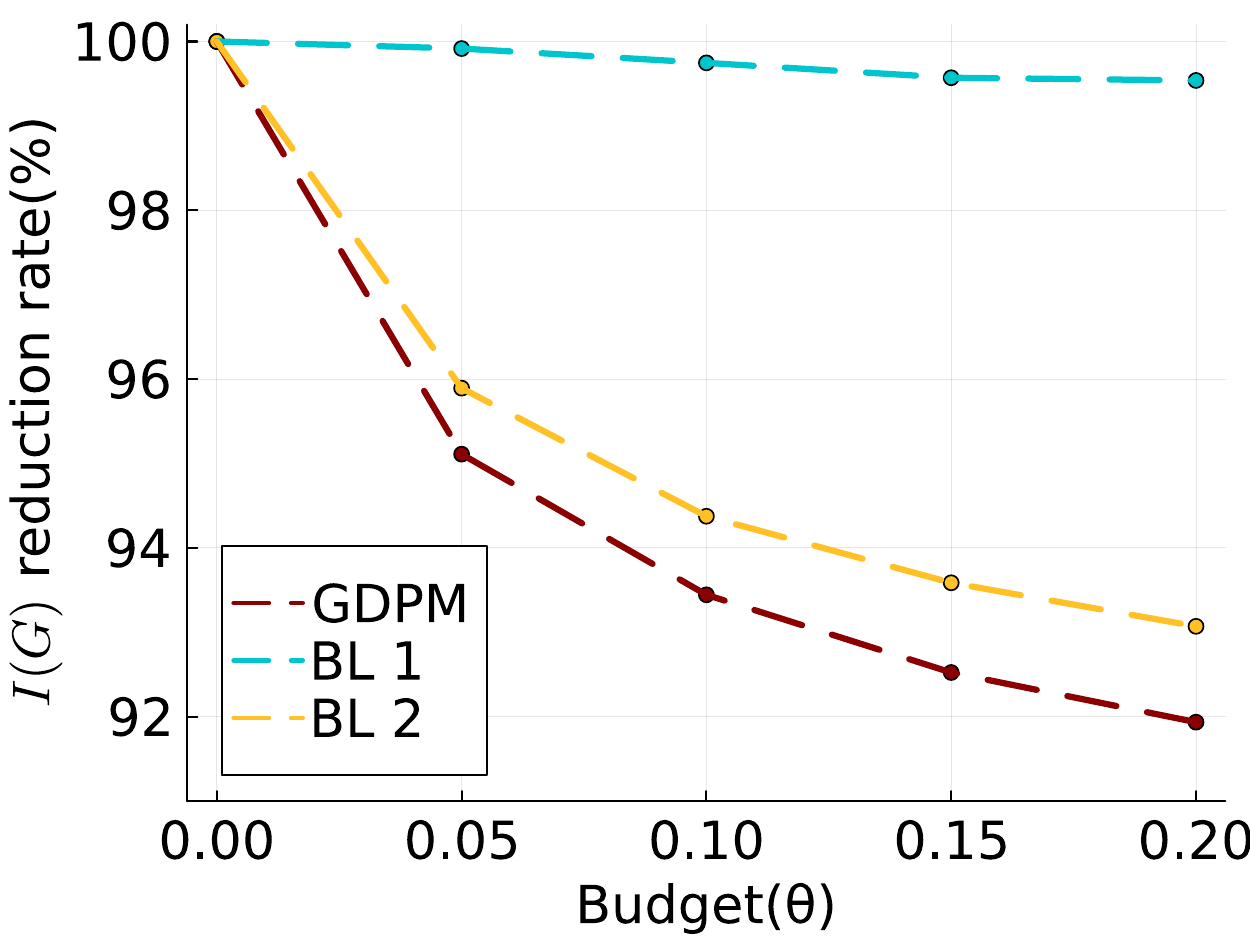} \\
        (a) \TwitterFive & (b) \TwitterFifty \\
        \includegraphics[width=0.45\columnwidth]{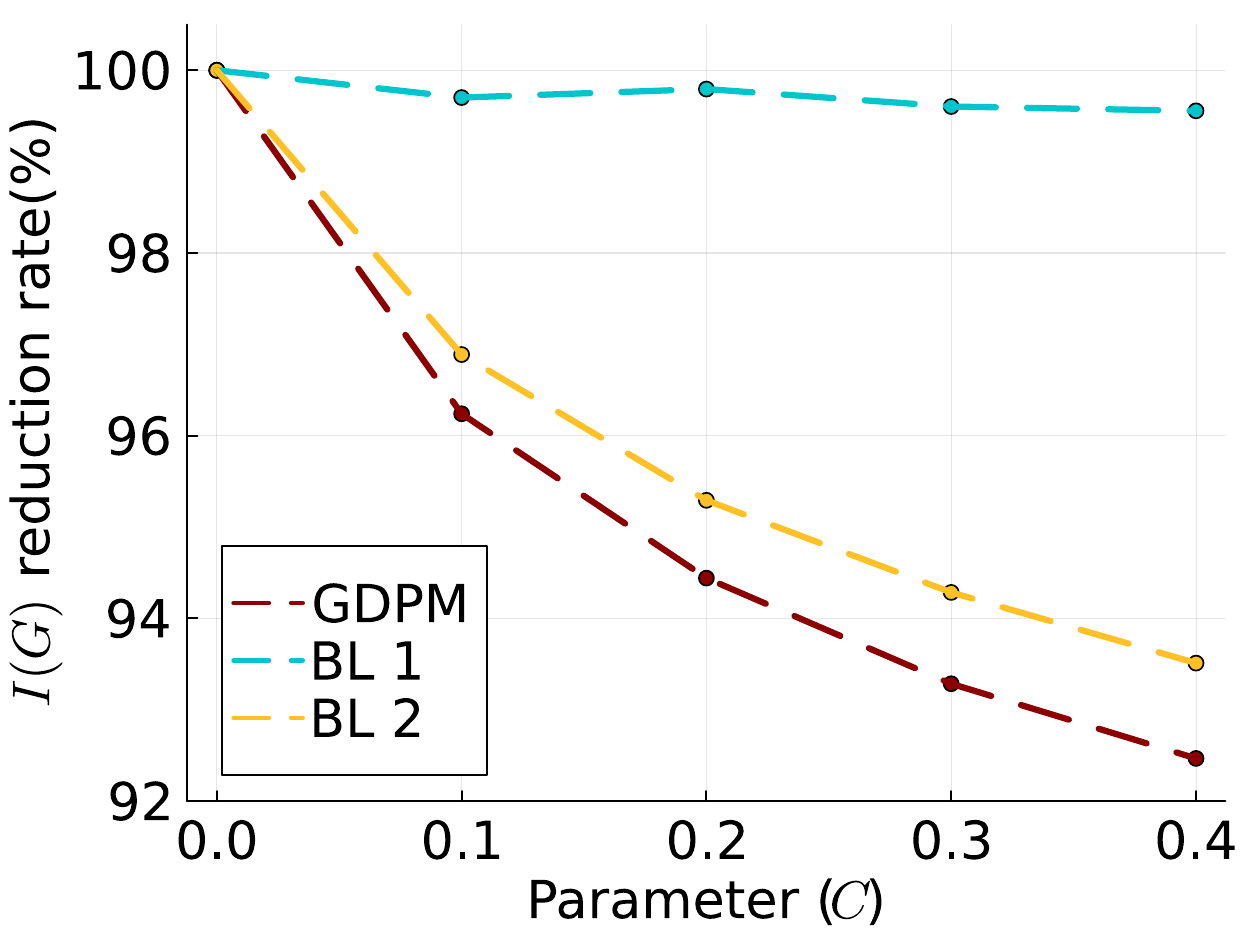} &
	\includegraphics[width=0.45\columnwidth]{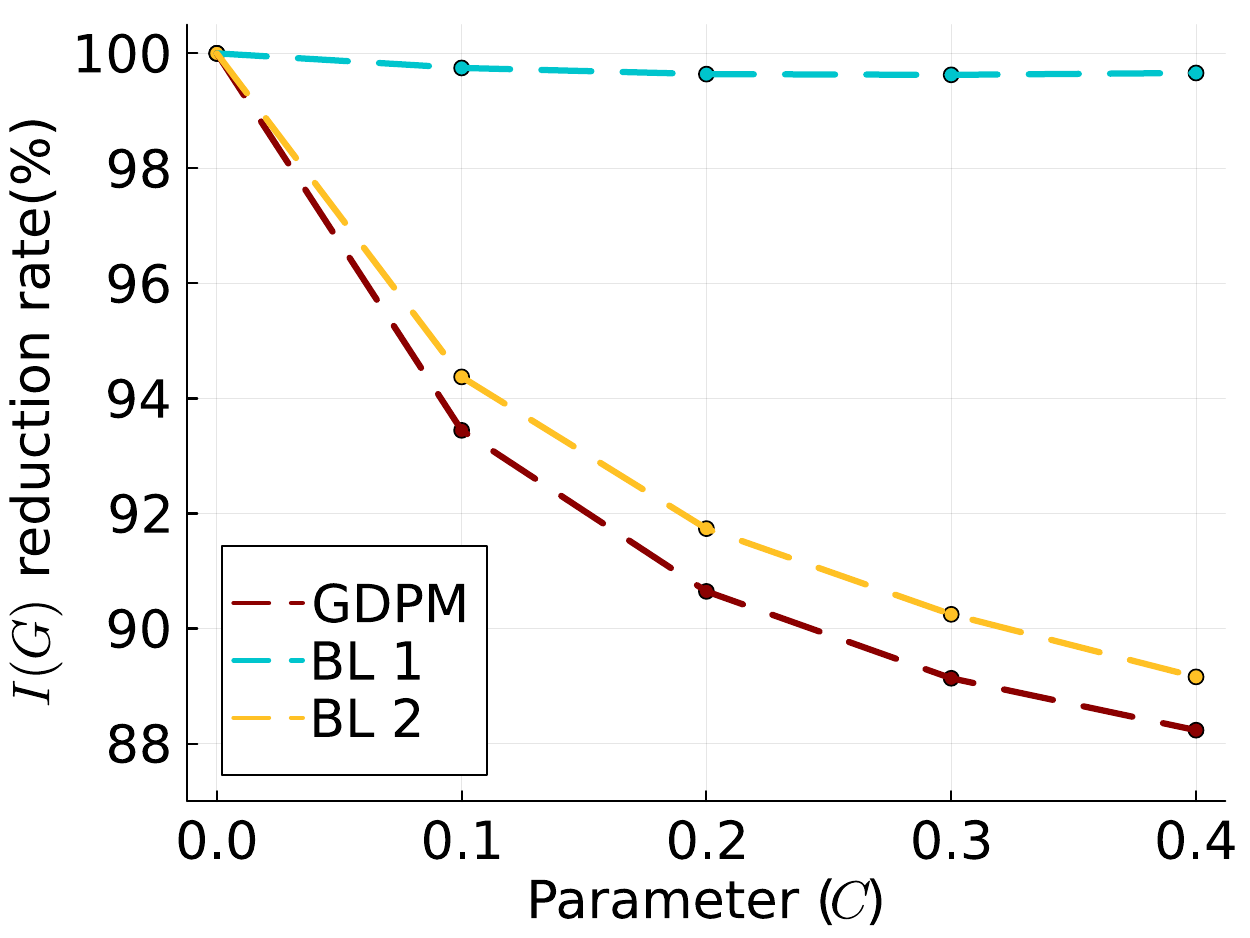} \\
        (c) \TwitterFive & (d) \TwitterFifty \\
    \end{tabular}
    \caption{\label{fig:alg-comparison}
		Reduction of the disagreement--polarization index on two datasets
		for all of our algorithms ($L=10$). The $y$-axis shows the reduction ratio ${f(\m+X_{\ALG})}/{f(\m+X)}$. In (a)-(b) we set $C=0.1$ and vary $\vartheta\in\{0.05,0.1,0.15,0.2\}$. In (c)-(d) we set $\vartheta=0.1$ and vary $C \in\{0.1,0.2,0.3,0.4\}$. }
    \label{fig:algo-comparison}
\end{figure}

\begin{table}[t]
\centering
\caption{Comparison of the reduction ratio $\frac{f(\m+X^{(T)})}{f(\m+X)}$(\%) of
	\ouralgo, \blone, and \bltwo on real-world graphs.
	In the experiments, we set $\vartheta =0.1$ and $C=0.1$. We used $L=10$ and $T=100$
	for \ouralgo and we set $T=10$ for the greedy baselines. We use synthetic
	opinion vectors, user-interest, and influence-topic matrices for all
	datasets except \TwitterFive and \TwitterFifty. }
\label{tab:convex-optimization}
\scalebox{0.95}{
\begin{tabular}{lccc}
\toprule

\textbf{Graph} & \textbf{\ouralgo} & \textbf{\blone} &  \textbf{\bltwo} \\
\midrule
\TwitterFive     & \textbf{96.24} & 99.70 & 96.88 \\
\TwitterFifty & \textbf{93.44} & 99.74 & 94.37\\
\midrule
\Erdos           & \textbf{94.34} & 100 & 94.97\\
\Advogato        & \textbf{91.36} & 100 & 92.59\\
\PagesGovernment & \textbf{87.82} & 100 & 88.62\\
\WikiElec        & \textbf{87.30} & 100 & 89.72\\
\HepPh           & \textbf{86.13} & 100 & 88.69\\
\Anybeat         & \textbf{92.17} & 100 & 93.21\\
\PagesCompany    & \textbf{92.41} & 100 & 93.28\\
\AstroPh         & \textbf{88.20} & 100 & 89.74\\
\CondMat         & \textbf{91.75} & 100 & 94.36\\
\Gplus           & \textbf{93.91} & 100 & 94.48\\
\Brightkite      & \textbf{93.02} & 100 & 93.99\\
\Themarker      & \textbf{85.62} & 100 & 89.30\\
\Slashdot        & \textbf{92.23} & 100 & 93.43\\
\WikiTalk        & \textbf{92.82} & 100 & 93.47\\
\Gowalla         & \textbf{91.79} & 100 & 92.63\\
\Academia        & \textbf{92.04} & 100 & 93.37 \\
\GooglePlus      & \textbf{86.43} & 100 & 88.28\\
\Citeseer        & \textbf{90.53} & 100 & 91.48\\
\MathSciNet      & \textbf{93.55} & 100 & 93.93\\
\TwitterFollows  & \textbf{94.22} & 100 & 95.59\\
\YoutubeSnap     & \textbf{93.58} & 100 & 94.76\\
\bottomrule
\end{tabular}
}
\end{table}

\spara{Comparison with greedy baselines and varying parameters.}
We compare \ouralgo against the baselines \blone and \bltwo on \TwitterFive and
\TwitterFifty, and vary the
parameters~$\vartheta$ and~$C$.  Note that since \ouralgo is guaranteed to
converge to an optimal solution, we expect it to outperform both baselines.

We report the results of all algorithms with varying 
$\vartheta \in \{0.05,0.1,\allowbreak 0.15,0.2\}$ in 
Figures~\ref{fig:alg-comparison}(a)--(b)
for \TwitterFive and \TwitterFifty.
As expected, \ouralgo obtains the largest reduction of the objective.
Furthermore, \bltwo outperforms \blone by a large margin.
This is not surprising, since we designed \bltwo based on insights
from analyzing the behavior of \ouralgo (see below);
the observed behavior thus suggests that our intuition about \ouralgo is correct.
In addition, we observe that the reduction in disagreement and polarization 
increases with~$\vartheta$.
This behavior aligns with our expectation, 
as larger values of $\vartheta$ enlarge the feasible space
and allow for more flexibility in recommending 
interesting topics.

In Figures~\ref{fig:alg-comparison}(c)--(d) we report the results of all algorithms with varying 
$C\in\{0.1,0.2,\allowbreak 0.3,0.4\}$ for \TwitterFive and \TwitterFifty.
The behavior of all algorithms remains consistent:
\ouralgo achieves the largest reduction, 
while \bltwo outperforms \blone.
As expected, the reduction in disagreement and polarization increases with $C$, 
since larger values of $C$ allow for more impact of the timeline algorithm.

Finally, we note that \ouralgo achieves a larger reduction on \TwitterFifty than
on \TwitterFive throughout all experiments. This is perhaps a bit surprising
since on both datasets we increase the total edge weight by a
$C$-fraction. However, the average node degree of \TwitterFifty is
larger than for \TwitterFive.  Furthermore, the user--topic matrix $\m+X$ and
influence--topic matrix $\m+Y$ have different structure for $\TwitterFive$
and $\TwitterFifty$, which results in the low-rank adjacency matrix~$\AXC$
containing 25\% and 33\% of non-zero entries, respectively. We believe that both
of these characteristics of the datasets lead to higher connectivity in
$\TwitterFifty$, which results in better averaging of the opinions and thus
ultimately in less polarization and disagreement.

\spara{Performance of the optimization algorithms.}
We report the optimization results of \ouralgo, \blone, and \bltwo in
Table~\ref{tab:convex-optimization}. We used different real-world graphs with
synthetically generated polarized opinions and synethically generated matrices~$\m+X$ and
$\m+Y$. We run the greedy baselines 10~iterations due to the high computation
cost and choose the best $\m+X^{(T)}$ as output in our experiments. 

We observe that \ouralgo outperforms the two baselines on all graphs. This is
the expected behavior, since \ouralgo guarantees decreasing the objective
function constantly and converges to the optimal solution; the two baselines
have no such property. Across all datasets, \ouralgo decreases the polarization
and disagreement by at least 5.6\% and by up to 14.4\%. Furthermore, \bltwo
outperforms \blone by a large margin and its results are often not much worse
than those of \ouralgo. Interestingly, \blone cannot reduce the polarization and
disagreement on all graphs. 

\spara{Comparison with off-the-shelf convex solver.}
Since Problem~\ref{eq:problem} is convex, we compare our gradient-descent based
algorithm \ouralgo against \cjl. \cjl is a popular off-the-shelf convex optimization
tool written in Julia and we use it with the SCS solver. Our experiments show
that \ouralgo is orders of magnitude more efficient than \cjl. In particular,
 even though for this experiment we used
102~GB of RAM, running \cjl on graphs with more than 500~nodes exceeds the
memory constraint. In contrast, \ouralgo scales up to graphs with millions of
nodes and edges (see Tables~\ref{tab:convex-optimization}
and~\ref{tab:running-time-error}). For the details of these experiments, see
Section~\ref{app:additional-exp}. Here, one of the bottlenecks for \cjl is that
it cannot access our efficient opinion estimation routine from
Proposition~\ref{prop:fast-opinions}. In the appendix
(Table~\ref{tab:running-time-error}) we show that the routine from the
proposition is indeed orders of magnitude more efficient than estimating the
opinions using na\"ive matrix inversion.

\spara{Understanding the behavior of \ouralgo.}
Next, we perform experiments to obtain further insights into which topics 
are favored by \ouralgo and which ones are penalized.

\begin{figure*}[t]
  \centering
  \begin{tabular}{ccc}
  \includegraphics[width=0.66\columnwidth]{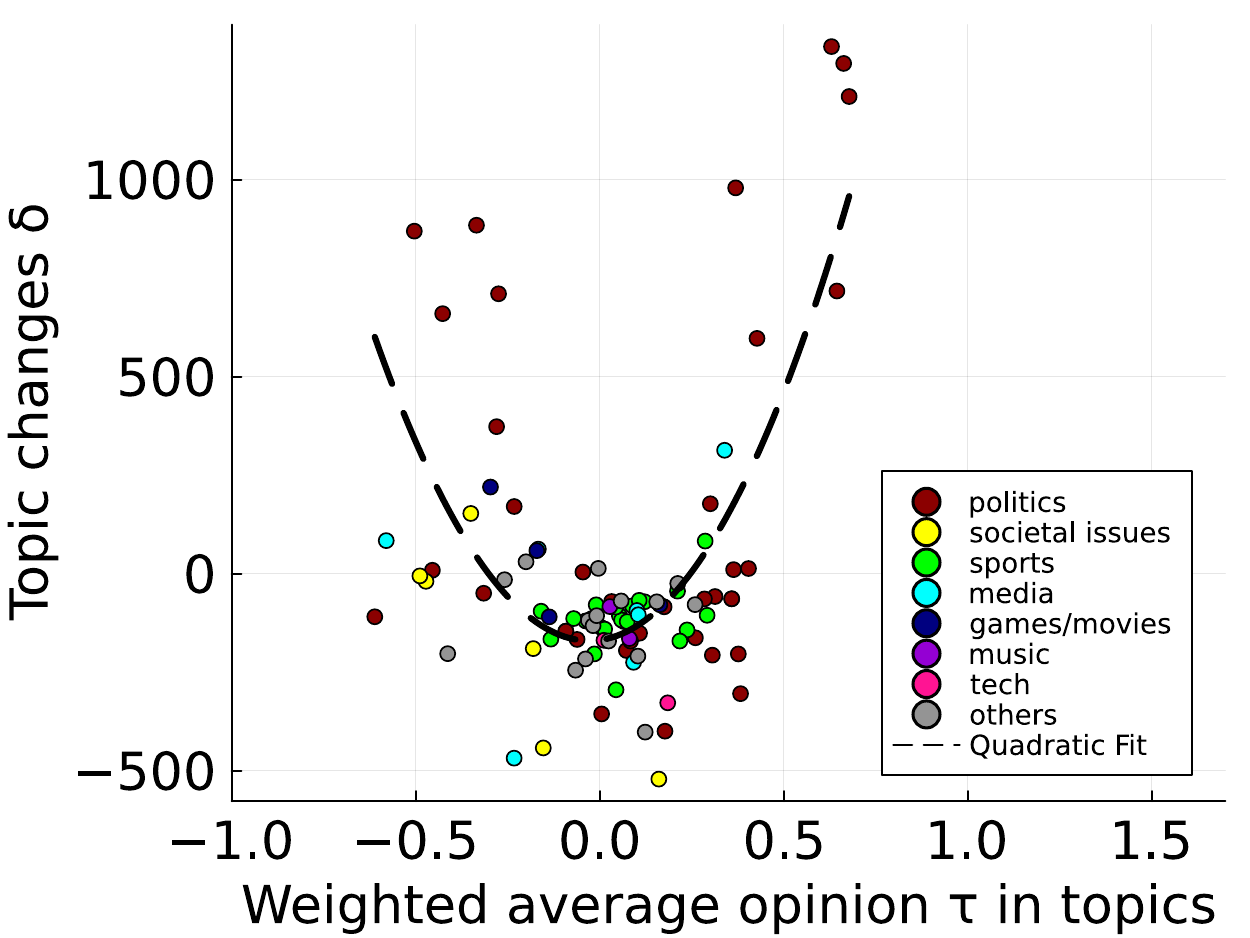} &
  \includegraphics[width=0.66\columnwidth]{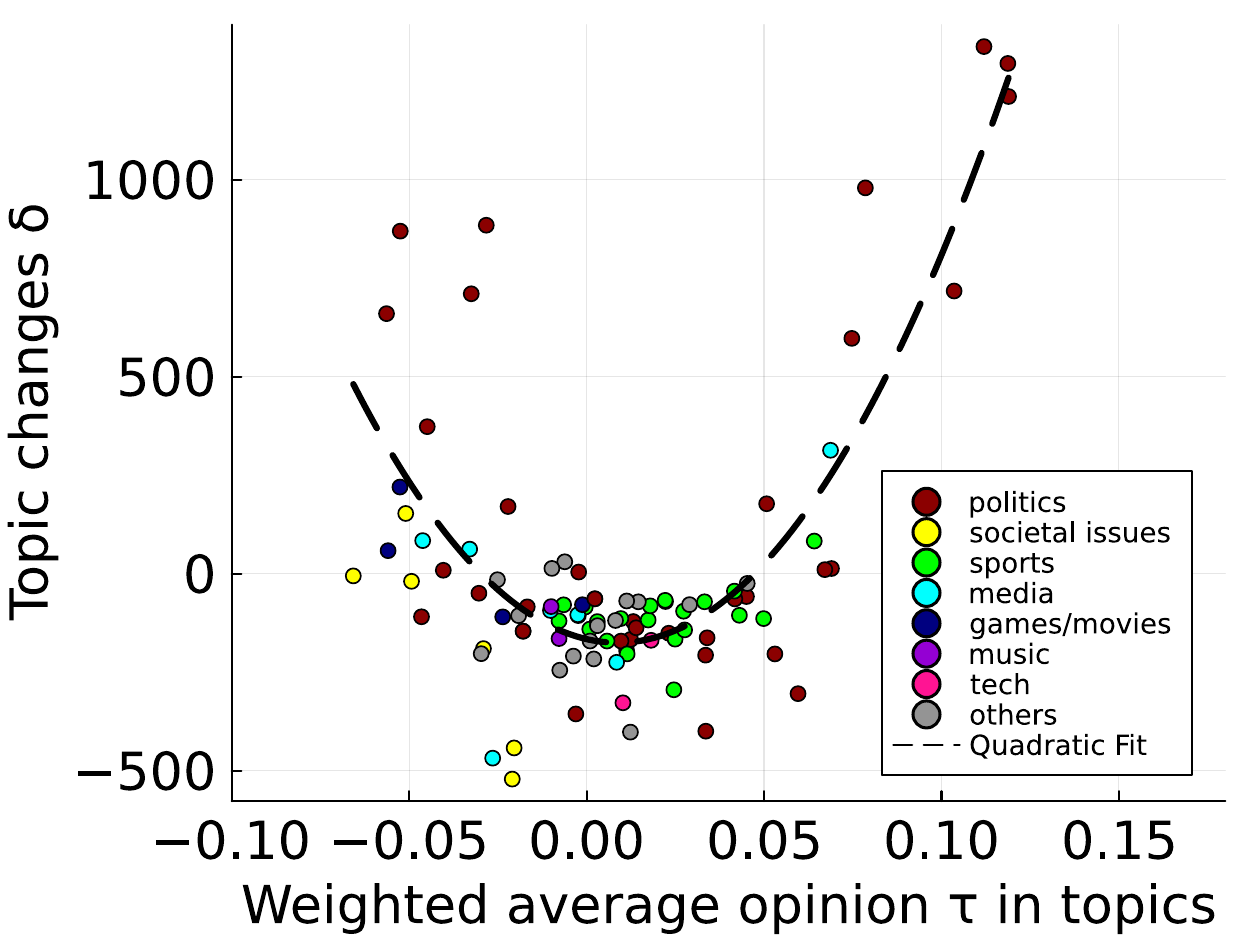} &
  \includegraphics[width=0.66\columnwidth]{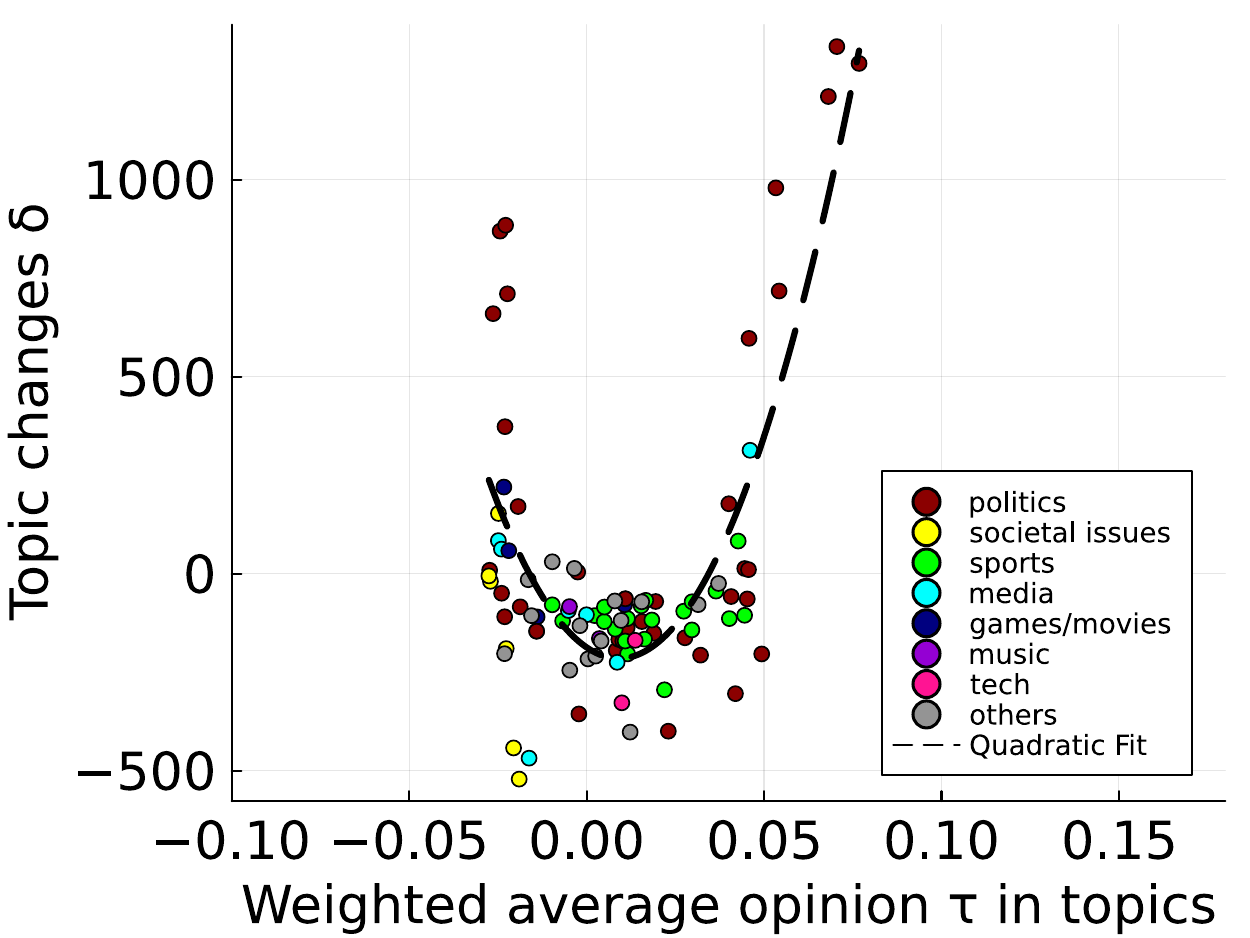} \\
  (a) & (b) & (c) \\
    \includegraphics[width=0.66\columnwidth]{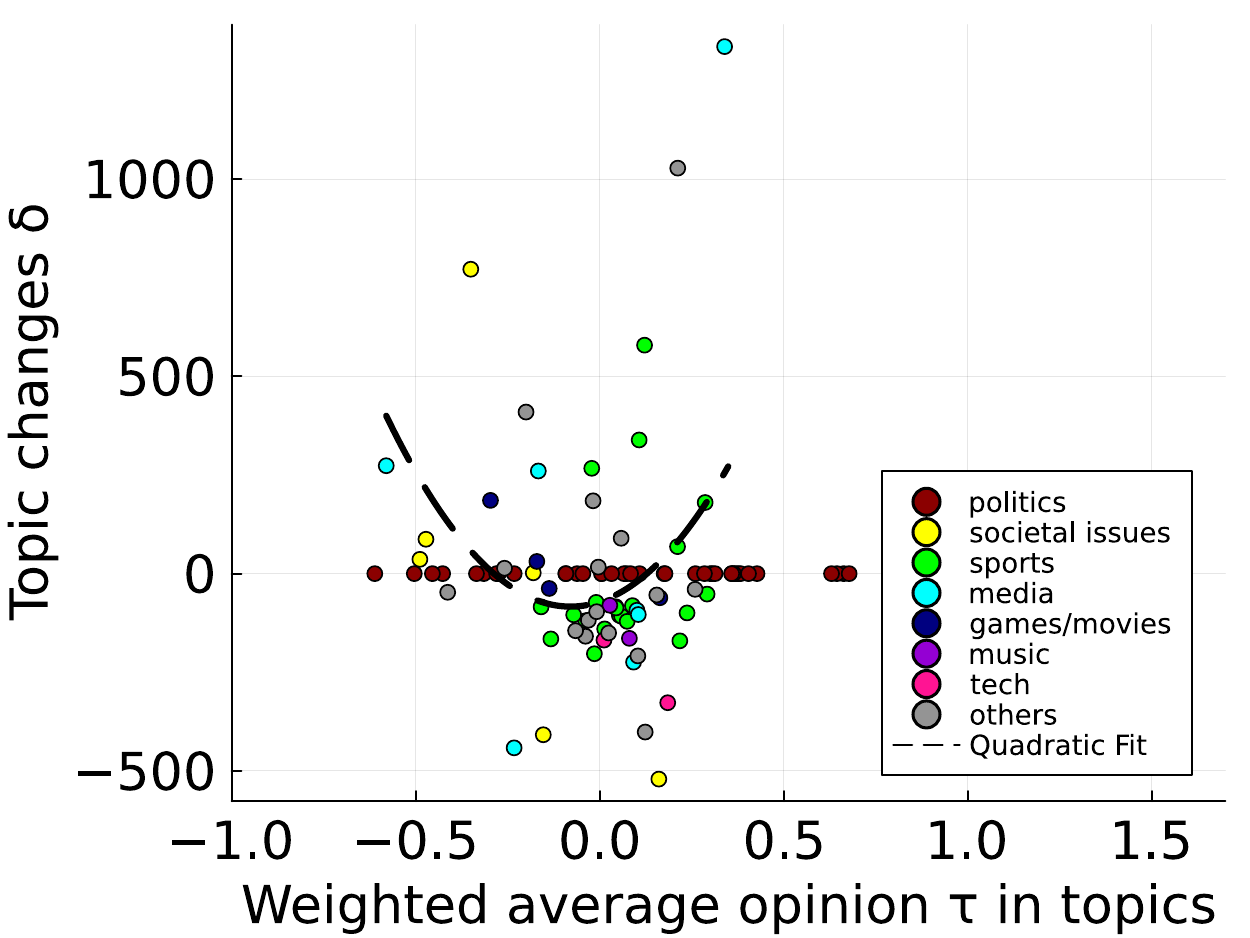} & 
  \includegraphics[width=0.66\columnwidth]{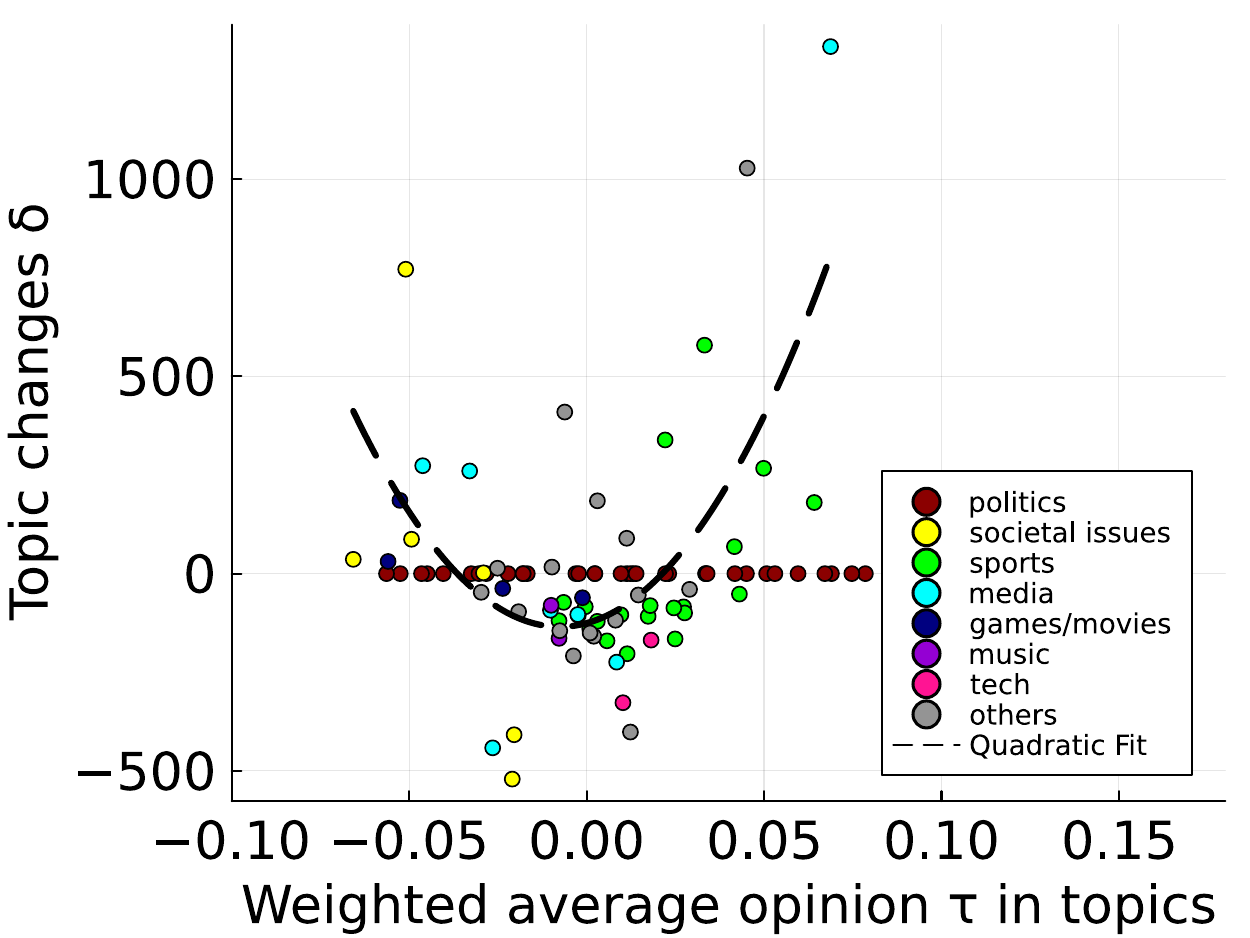} & 
  \includegraphics[width=0.66\columnwidth]{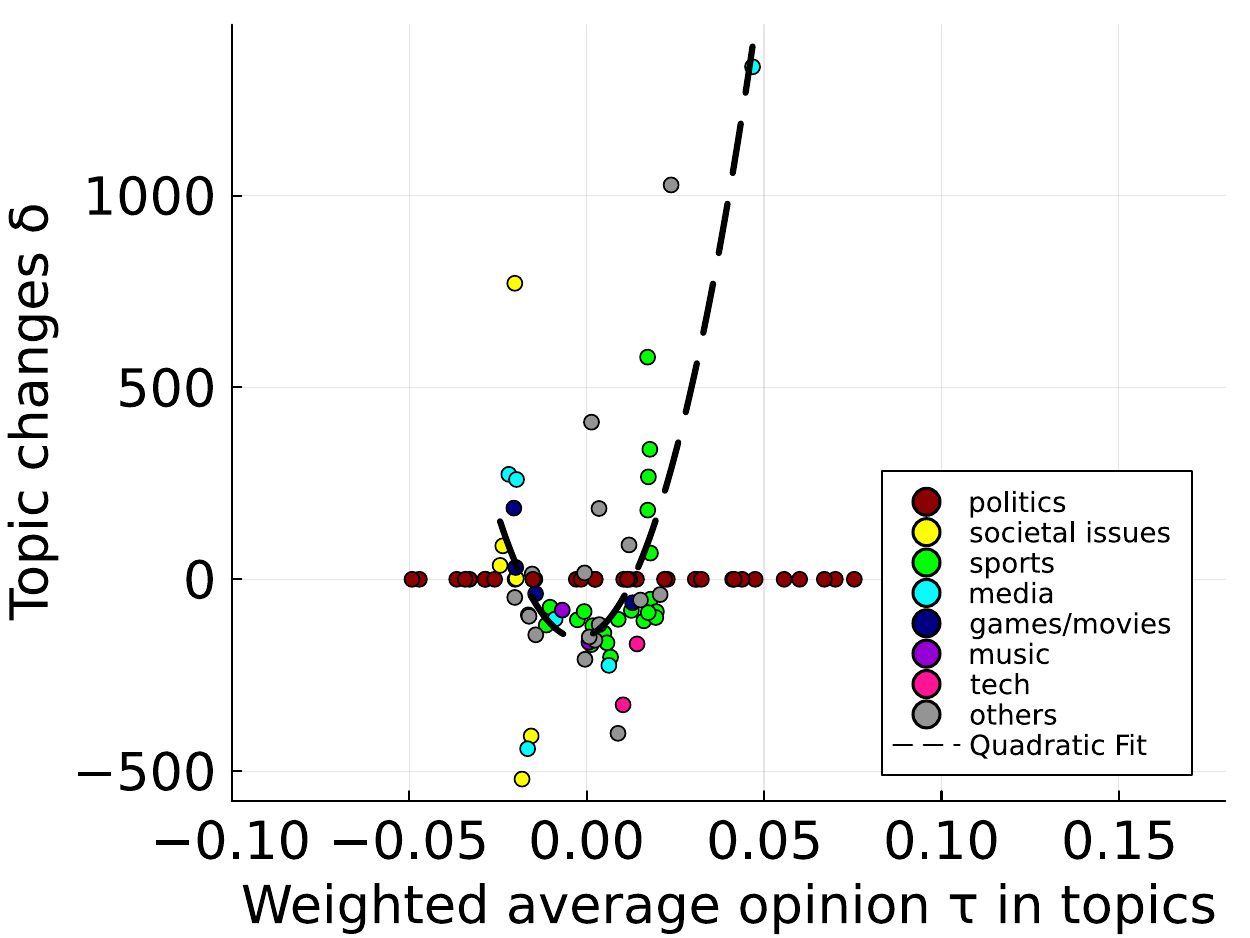} \\
  (d) & (e) & (f) \\
  \end{tabular}
  \caption{
  \label{fig:behavior-fifty}
	  Behavior of {\ouralgo} on \TwitterFifty
  	($\vartheta=0.1$, $C=0.1$, $L=10$).
	  We report the change of topic importance ($y$-axis) and the 
	  weighted average of the opinions of influential users for each topic ($x$-axis):
	  (a) weighted innate opinions~$\tau_{j,\v+s}$;
	  (b) weighted expressed opinions before optimization~$\tau_j$;
	  (c) weighted expressed opinions after optimization~$\tau_j$.
   (d)---(f) repeat the same plots when the algorithm must not change interest in political topics.
	  For reference, the results are fitted with a quadratic function.}
  \label{fig:behavior-innate-fifty}
\end{figure*}

To answer this question,
we consider the initial interest matrix~$\m+X$ and the 
matrix~$\m+X^{(T)}$ obtained after \ouralgo converged. 
To quantify the behavior of \ouralgo, 
we consider the column changes among $\m+X$ and $\m+X^{(T)}$. 
Specifically, for each topic~$j$, we measure the change of its weight 
given by $\delta_j = \sum_{i} \m+X^{(T)}_{ij} - \sum_{i} \m+X_{ij}$.
Note that $\delta_j>0$ indicates that topic~$j$ has more weight in
$\m+X^{(T)}_j$ than in~$\m+X$, i.e., \ouralgo ``favors'' it; 
similarly, $\delta_j<0$ indicates that topic~$j$ has less weight in
$\m+X^{(T)}_j$ than in~$\m+X$, i.e.,~\ouralgo ``penalizes'' it.

In Figure~\ref{fig:behavior-fifty}(a) we plot tuples
$(\tau_{j,\v+s}, \delta_j)$ for each topic~$j$, 
where $\delta_j$ is the change in importance for topic~$j$, 
as defined in the previous paragraph, 
and $\tau_{j,\v+s} = \sum_{u\in V} \m+Y_{ju} \v+s(u)$ is the weighted
average of the innate opinions of the influencers for topic~$j$. 
We also color-code the topics based on their content.  
We observe that \ouralgo clearly favors
topics with large absolute values~$\abs{\tau_{j,\v+s}}$ and it penalizes
non-controversial topics with $\abs{\tau_{j,\v+s}}$ close to~$0$. 
We explain this behavior as a consequence of the FJ model opinion dynamics: 
more controversial topics have a larger impact on the polarization, 
and to reduce the polarization one has to bring together people from opposing sides.

We note that in all plots, the most favored topics are political. 
This is surprising, as the algorithm is not aware of the topic labeling. 
However, we believe this is a consequence of the fact that 
political topics are among the most controversial (see also below).

In Figures~\ref{fig:behavior-fifty}(b)
and~\ref{fig:behavior-fifty}(c), 
we again show the $\delta_j$ values 
but this time plotted against $\tau_{j,\zXC}$ using the 
original expressed opinions~$\zXC$ (before optimization) and 
$\tau_{j,\zXCTT}$ using the final expressed opinions~$\zXCTT$ (after optimization). 
Qualitatively, we observe the same behavior as before, so that
more controversial topics are favored and non-controversial topics are penalized. 
Observe that now the $x$-axes have smaller scales, 
since the expressed opinions are contractions of the innate opinions.
Here, it is important to observe that before the optimization
(Figure~\ref{fig:behavior-fifty}(b)) the average topic opinions
were in $[-0.066,0.118]$ and after the optimization
(Figure~\ref{fig:behavior-fifty}(c)) they are  in $[-0.028,0.076]$.
Thus, the algorithm clearly brought all topics closer together.

Next, we study the behavior of \ouralgo when
we do not allow to make any changes on the accounts' interests in political topics, 
i.e., we set
 $\m+X^{(U)}_{ij}=\m+X_{ij}$ and $\m+X^{(L)}_{ij} = \m+X_{ij}$ for all political topics~$j$ and all accounts~$i$.
In Figures~\ref{fig:behavior-fifty}(d)--(f)
we show the same plots as in Figures~\ref{fig:behavior-fifty}(a)--(c),
when weight changes for political topics are not allowed.
We obtain the same qualitative outcome as before: 
controversial topics are favored and
non-controversial topics are~penalized.
We used these qualitative insights
to develop the second baseline algorithm \bltwo.

As expected, 
when weight changes for political topics are not allowed,
we obtain a restricted version of the problem
which limits the disagreement-polarization reduction.
For reference, in the setting of Figures~\ref{fig:behavior-fifty}(a)--(c), 
when weight changes for all topics are allowed, 
the disagreement-polarization index is reduced to $93.44\%$ of its original value. 
In contrast, in the setting of Figures~\ref{fig:behavior-fifty}(d)--(f),
with no changes on political topics,
the disagreement-polarization index is reduced to only $97.69\%$ of its original value. 

We stress that the above fine-grained analysis, that gives insights on which
should topics be penalized and favored to reduce the polarization and
disagreement in the FJ model, has only become possible due to the introduction
of our model from Section~\ref{sec:problem}. We believe that in the future
it is interesting to compare these insights with results from political science
and computational social sciences.

We include additional experiments, including a running-time analysis, in the
appendix.

\section{Conclusion}

We showed how to augment the popular FJ model to take into account aggregate
information of timeline algorithms. This allows us to bridge between
network-level opinion dynamics and user-level recommendations.
We then considered the problem of optimizing the timeline algorithm,
so as to minimize polarization and disagreement in the network,
and developed an efficient gradient-descent algorithm, \ouralgo, 
which computes an $(1+\varepsilon)$-approximate solution in
time~$\tO(m\sqrt{n}\log(1/\varepsilon))$ under realistic parameter settings.
We presented an algorithm that provably approximates our model, including the
measures of polarization and disagreement, in near-linear time.  
Our experiments confirm the efficiency and effectiveness of the proposed methods
and showed that our gradient-descent algorithm is orders of magnitude faster than
an off-the-shelf solver. %
We also release the largest graph datasets with ground-truth opinions.

We believe that our work provides several directions for future
research. 
First, extensions to directed graphs (and, hence, non-symmetric matrices) are
highly interesting. 
Second, inventing algorithms that are even more efficient than \ouralgo is
intriguing (for instance, in a setting with sublinear time/space).
Third, it will be valuable to 
consider other opinion-formation models, beyond the FJ model, and compare the results.
Fourth, it will be intriguing to design more complex models,
capturing real-world nuances,  
that allow us to bridge between opinion dynamics and properties
of present-day timeline algorithms.

\begin{acks}
This research has been funded by 
the ERC Advanced Grant REBOUND (834862), 
the EC H2020 RIA project SoBigData++ (871042), 
the Wallenberg AI, Autonomous Systems and Software Program (WASP) 
funded by the Knut and Alice Wallenberg Foundation, and
the Vienna Science and Technology Fund (WWTF) [Grant ID: 10.47379/VRG23013]. The computation was enabled by resources provided by the National Academic Infrastructure for Supercomputing in Sweden (NAISS) partially funded by the Swedish Research Council through grant agreement no. 2022-06725.
\end{acks}

\clearpage

\balance
\setcitestyle{numbers}
\bibliographystyle{plainnat}
\bibliography{main}

\appendix
\onecolumn

\section{Omitted pseudocode}
\label{sec:pseudocode}
In this section, we present the pseudocode of Algorithms~\ref{alg:opinions},~\ref{alg:optimization} and~\ref{alg:baseline}.

{\SetAlgoNoLine
 \LinesNumbered
 \DontPrintSemicolon
 \SetAlgoNoEnd
\begin{algorithm2e}[H]
\KwIn{Innate opinion $\v+s$, user--topic matrix $\m+X$, influence--topic matrix $\m+Y$, fraction of \emph{weight} $C$, error parameter $\varepsilon$}
\KwOut{Approximated expressed opinion $\tzXC$}
$\varepsilon_{\v+z_1} = \frac{\varepsilon}{4}
				\cdot \min\left\{
					1,
					\frac{2n}{200 \cdot CW \cdot \norm{\m+U}_2 \cdot \norm{\m+V}_2}
				\right\}$\;
$\varepsilon_{\m+R} =\frac{1}{2k}
			\min\left\{0.009 \frac{2n}{CW \cdot \norm{\m+V}_F},
			\frac{2n}{10^{5}\cdot CW \cdot \norm{\m+V}_2}
			\cdot \min\left\{ 100,
					\frac{2n}{CW} \cdot \frac{\varepsilon/4}
					{\norm{\m+U}_2 \cdot \norm{\m+V}_2
						\cdot \norm{\v+s}_2} \right\}
			\right\}$\;
\sninline{We should slightly change the constants such that they contain the
	Frobenius norm instead of the spectral norm (the Frobenius norm can be
	computed efficiently, while the spectral norm takes more time).}
$\varepsilon_{\v+z_2} = \frac{2n}{CW} \cdot \frac{\varepsilon}{4}$\;
$\m+M \gets \m+I + \m+L + \diag(\AXC\v+1)$\;
$\m+U \gets \begin{pmatrix} \m+X & \m+Y^\top \end{pmatrix}$,
	$\m+V \gets \begin{pmatrix} \m+Y \\ \m+X^\top \end{pmatrix}$\;
$\v+z_1 \gets \Solve(\m+M, \v+s, \varepsilon_{\v+z_1})$\;
$\v+y_1 \gets \m+V \v+z_1$\;
$\m+R \gets $ the $n\times (2k)$ matrix, where the $j$-th column is given by
	$\Solve(\m+M, \m+w_j,\varepsilon_{\m+R})$ with $\m+w_j$ denoting the $j$-th
	column of $\m+U$ for all $j$\;
$\m+S \gets \m+I - \frac{CW}{2n} \m+V \m+R$\;
$\m+T \gets \m+S^{-1}$\;
$\v+y_2 \gets \m+T \v+y_1$\;
$\v+y_3 \gets \m+U \v+y_2$\;
$\v+z_2 \gets \Solve(\m+M, \v+y_3, \varepsilon_{\v+z_2})$\;
\Return{$\tzXC \gets \v+z_1 + \frac{CW}{2n} \v+z_2$}\;
\caption{Compute an approximation $\tzXC$ of $\zXC$}
\label{alg:opinions}
\end{algorithm2e}
}

{\SetAlgoNoLine
 \LinesNumbered
 \DontPrintSemicolon
 \SetAlgoNoEnd
\begin{algorithm2e}[H]
\KwIn{Innate opinion $\v+s$, user--topic matrix $\m+X$, influence--topic matrix $\m+Y$, budget $\vartheta$, fraction of \emph{weight} $C$ }
\KwOut{User--topic matrix $\m+X^{(T)}$ after optimization }
$L \gets \frac{8CW}{\sqrt{n}} \cdot \norm{\v+s}_2 \cdot \norm{\m+Y}_2^2$\;
$\m+X^{(0)} \gets \m+X$\;
\For{$T = 1,\dots,O\left(\sqrt{\frac{CWkn}{\varepsilon}}\right)$}{
	Compute $\tzXCT$ using Algorithm~\ref{alg:opinions}\;
	$\widetilde{\nabla}_{\m+X} f(\m+X^{(T)}) \gets
		\frac{CW}{2n} (2\cdot \tzXCT \cdot \tzXCT^\top\cdot \m+Y^\top 
		- \tzXCT\odot \tzXCT\cdot \v+1_k^\top
		- \v+1_n \cdot (\tzXCT^\top \odot \tzXCT^\top) \cdot \m+Y^\top )$\;
	$\m+V^{(T)} \gets $ the matrix where the $i$-th row is given by
		$\ProjQ{\m+X^{(T)}_i-\frac{1}{L} (\widetilde{\nabla}_{\m+X} f(\m+X^{(T)}))_i}$\;
	$\alpha_{T} \gets \frac{T+1}{2}$\;
	$\m+W^{(T)} \gets $ the matrix where the $i$-th row is given by
		$\ProjQ{(\m+X^{(0)})_i
			- \frac{1}{2L}\sum_{t=1}^T \alpha_t (\widetilde{\nabla}_{\m+X} f(\m+X^{(t)}))_i}$\;
	$A_T \gets \sum_{i=0}^T \alpha_i$\;
	$\tau_{T} \gets \frac{\alpha_T}{A_T}$\;
	$\m+X^{(T+1)} \gets \tau_T \m+V^{(T)} + (1-\tau_T) \m+W^{(T)}$\;
}
\Return{$\m+X^{(T)}$}
\caption{\ouralgo}
\label{alg:optimization}
\end{algorithm2e}
}
\newpage

{\SetAlgoNoLine
 \LinesNumbered
 \DontPrintSemicolon
 \SetAlgoNoEnd
\begin{algorithm2e}[H]
\KwIn{Innate opinion $\v+s$, user--topic matrix $\m+X$, influence--topic matrix $\m+Y$, lower-bound matrix $\m+X^{(L)}$, upper-bound matrix $\m+X^{(U)}$, maximum iterations~$T_{max}$ }
\KwOut{User--topic matrix $\m+X^{(T)}$ after optimization }
$\m+X^{(0)} \gets \m+X$\;
\For{$T = 1,\dots,T_{\max}$}{
	$\m+X^{(T)} \gets \m+X^{(T-1)}$\;
	Compute $\tzXCit{T}$ using Algorithm~\ref{alg:opinions}\;
	$\bar{z}=\frac{1}{n}\sum_{u\in V} \tzXCit{T}(u)$\; %
	$\tau_j = \sum_{u\in V} \m+Y_{ju} \tzXCit{T}(u)$\;
	\For{ each row $i$}{
		\If{(\blone)}{
			$j \gets$ topic with $\m+X^{(T)}_{ij} < \m+X^{(U)}_{ij}$ 
				minimizing $\abs{\tau_j-\bar{z}}$\; 
			$j'\! \gets$\! topic with $\m+X^{(T)}_{ij'} > \m+X^{(L)}_{ij'}$
				minimizing $\abs{\tau_{j'}-\bar{z}}$\;
		}
		\If{(\bltwo)}{
			$j  \gets$ topic with $\m+X'_{ij} < \m+X^{(U)}_{ij}$ minimizing~$-\tzXCit{T}(i) \tau_j$\; 
			$j' \gets$ topic with $\m+X'_{ij'} > \m+X^{(L)}_{ij'}$ and
				$\tzXCit{T}(i)\tau_j>0$ minimizing $\abs{\tau_{j'}-\bar{z}}$
		}
		$\delta \gets \min\{\m+X^{(U)}_{ij} - \m+X^{(T)}_{ij}, \m+X^{(T)}_{ij'} - \m+X^{(L)}_{ij'}\}$\;
		$\m+X^{(T)}_{ij} \gets \m+X^{(T)}_{ij} + \delta$\;
	$\m+X^{(T)}_{ij'} \gets \m+X^{(T)}_{ij'} - \delta$\;
	}
}
\Return{$\m+X^{(T)}$}
\caption{Baselines \blone and \bltwo}
\label{alg:baseline}
\end{algorithm2e}
}

\section{Omitted experiments}
\subsection{Data Collection and Parameter Settings}
\label{sec:exp-setting}

\spara{Datasets.} We begin by describing our data-collection process.  Starting
from a list of \twitter accounts who actively engage in political discussions in
the US, which was compiled by Garimella and Weber~\cite{garimella2017long}, we
randomly sample two smaller subsets of $5\,000$ and $50\,000$ accounts,
respectively.  Since the dataset was more than 6 years old, only approximately
30-50\% of the accounts are still active or publicly accessible.  For these
accounts, we obtained the entire list of followers, except for users with more
than 100\,000 followers for whom we got only the 100\,000 most recent
followers (users with more than 100\,000 followers account for less than 2\% of
our dataset).  We also obtained the last 3\,200 tweets they posted on their own
timeline.  We use multiple \twitter-API keys and parallelize the data collection.
The data collection was started in March 2022 and took over a week to finish.

Based on this obtained information, we construct two graphs in which the nodes
correspond to \twitter accounts and the edges correspond to the accounts'
following relationships. Then we consider only the largest connected component
in each network and denote the resulting datasets \TwitterFive and \TwitterFifty
respectively.  In the end, \TwitterFive contains $1\,011$~nodes and
$1\,960$~edges.  \TwitterFifty contains $27\,058$~nodes and $268\,860$~edges.

To obtain the innate opinions of the nodes in the graphs, we proceed as follows.
First, we compute the political polarity score for each account using the method
proposed by Barber\'a~\citep{barbera2015birds}, which has been used widely in
the literature~\cite{brady2017emotion,boutyline2017social}.  The polarity scores
range from -2 to 2 and are computed based on following known political accounts.
To obtain the innate opinions~$\v+s$ of the retrieved accounts, 
we center the political scores to 0 and re\-scale them into the interval $[-1,1]$.  
We visualize the innate opinions of the accounts 
of \TwitterFive in Fig.~\ref{fig:opinions-twitter}(a) and 
of \TwitterFifty in Fig.~\ref{fig:opinions-twitter}(b). 
We observe that the distribution of the
opinions is relatively similar in both datasets, 
and that the opinion scores are significantly polarized.
A plausible explanation of this phenomenon is that 
our seed set consists of politically active accounts in the US, 
which are more likely to support one of the two extremes
of the political spectrum than having moderate~opinions.

\begin{figure}[t]
	\centering
	\begin{tabular}{cc}
	\includegraphics[width=0.4\columnwidth]{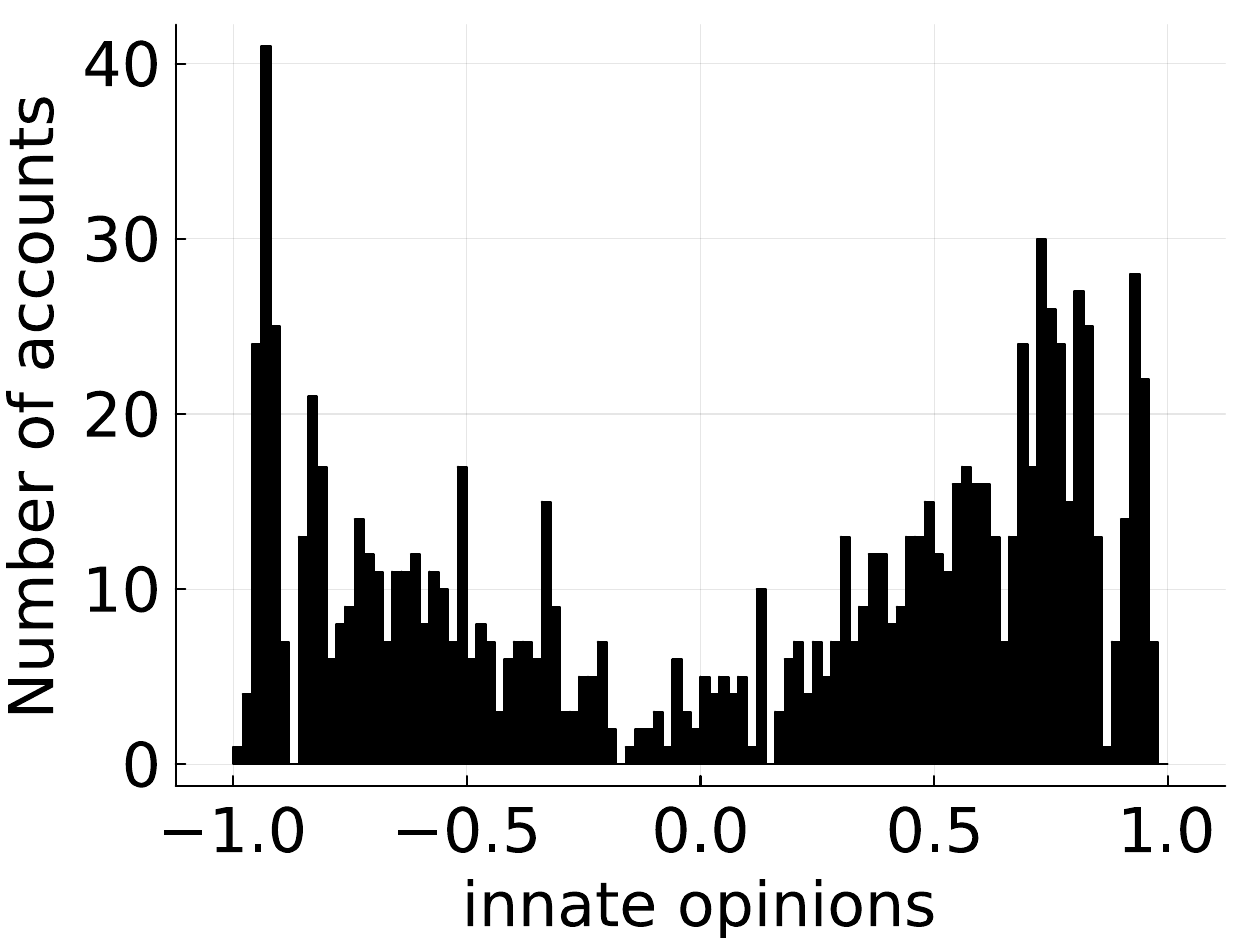} &
	\includegraphics[width=0.4\columnwidth]{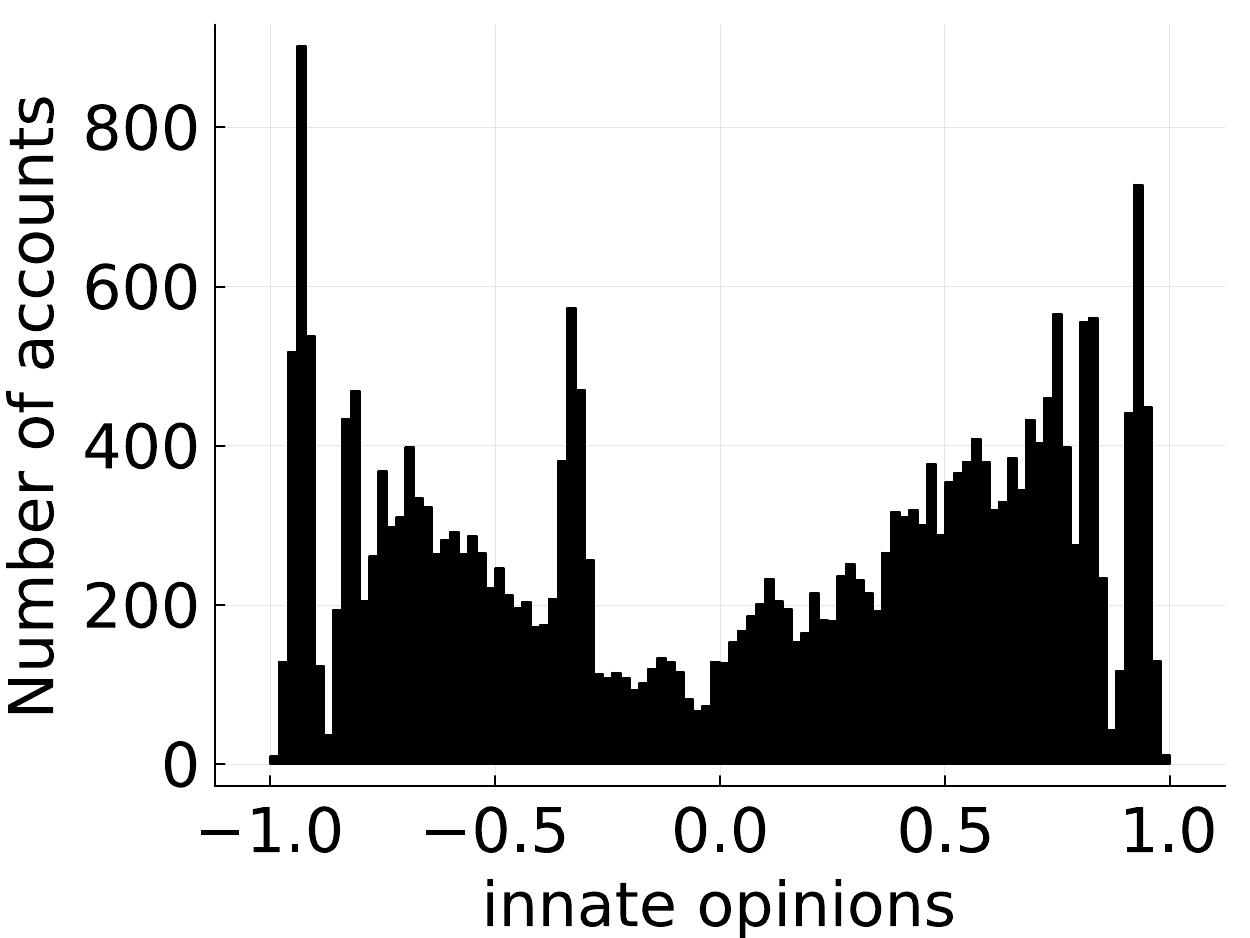} \\
	(a) \TwitterFive & (b) \TwitterFifty \\
	\end{tabular}
	\caption{
	Innate opinion distributions on our two real-world \twitter datasets.
	}
	\label{fig:opinions-twitter}
\end{figure}

\spara{User--topic and influence--topic matrices.} 
Next, we explain how we obtained the user--topic matrix~$\m+X$ and the
influence--topic matrix~$\m+Y$. We note that in an academic environment, it is
impossible to obtain these matrices exactly, since we cannot obtain data on how
the timelines of different users are composed and how the posts for each
topic are picked by the timeline algorithms that are deployed by online social
networks.  Therefore, we obtain~$\m+X$ and~$\m+Y$ by using retweet-data as a
surrogate, which indicates the users' interest and impact on different topics.
We now describe this process in detail.

We use textual information and hashtags in the tweets dataset to estimate the
interest of accounts and influential accounts in different topics.  More
concretely, we start by finding all hashtags that are used in the historical
tweets, and collect the hashtags used by each account.  We then apply {\tfidf}
on this data, where the documents correspond to accounts and the terms
correspond to hashtags.  The result gives a matrix~$\m+B$, in which each
entry~$\m+B_{uv}$ corresponds to the {\tfidf} score of account~$u$ for
hashtag~$v$.  Next, we apply non-negative matrix factorization~(NMF) on $\m+B$
to obtain topics from this matrix. NMF on a {\tfidf} matrix has been shown to
produce coherent topics in the past~\cite{kuang2015nonnegative}.
The NMF procedure produces two matrices: 
$\m+W \in \mathbb{R}^{n\times k}$ and $\m+H \in \mathbb{R}^{k\times h}$,
such that $\m+B \approx \m+W \m+H$. 
Here, $n$ is the number of accounts in the dataset,  
$h$ is the number of distinct hashtags, 
and the latent dimension~$k$ is the number of topics that we wish to find.
We systematically test different values of $k$ from 50--100 and find 
that for our data, $k=100$ produces the most reasonable topics.
Therefore, in our experiments we use $k=100$.

By the semantics of matrix factors in NMF, 
we interpret $\m+W_{ij}$ as an indicator of the interest of account~$i$ in
topic~$j$. Therefore, we set the $i$-th row of the interest matrix~$\m+X$ 
to $\m+X_i = {\m+W_i}/{\sum_{j=1}^k \m+W_{ij}}$, 
so as to satisfy the row-stochastic constraint, 
i.e., $\sum_j \m+X_{ij}=1$, for all~$i$.

Similarly, we interpret $\m+H_{jh}$ as the importance of hashtag~$h$ in topic~$j$. 
To avoid using hashtags that are too noisy, 
we consider only the most frequent hashtags
that make up for the 60\% of the volume of all hashtags.
We then set the influence--topic matrix~$\m+Y$ to the percentage of retweets that
an account receives for each topic. 
More concretely, we let $r_{ih}$ denote the number of retweets
for tweets posted by account~$i$ that contain hashtag~$h$. 
We set $\m+Y'$ to the matrix with 
$\m+Y_{ji}' = \sum_{h\in S_j} r_{ih}$, i.e., 
$\m+Y_{ji}'$ is the number of retweets for tweets posted by account~$i$ 
that contain hashtags assigned to topic~$j$. 
We then compute $\m+Y$ by normalizing the rows of $\m+Y'$, 
i.e., we set $\m+Y_{ji} = {\m+Y_{ji}'}/{\sum_{i=1}^n \m+Y_{ji}'}$ 
to ensure that $\m+Y$ is row stochastic.

\spara{Upper and lower bounds $\m+X^{(U)}$ and $\m+X^{(L)}\!$.} Given a matrix
$\m+X$ and a parameter $\vartheta \in [0,1]$, in our experiments (unless
mentioned otherwise) we construct the element-wise upper-bound matrix $\m+X^{(U)}$
and the lower-bound matrix $\m+X^{(L)}\!$ by setting 
$\m+X_{ij}^{(U)}\! = \min\{1, \m+X_{ij} + \vartheta\}$ and
$\m+X_{ij}^{(L)}\! = \max\{0, \m+X_{ij}-\vartheta\}$.
Intuitively, we can consider $\vartheta$ as a budget that the algorithm 
has to redistribute for each entry of $\m+X$. 
Note, however, that in the presence of topic label information, 
we can set topic-specific bounds.  
For example, if we set $\m+X_{ij}^{(U)}\! = \m+X_{ij}$, 
we then forbid the algorithm to increase account~$i$'s interest in topic~$j$.
We apply this idea in some of our experiments, 
by setting different bounds for political topics.
See Figures ~\ref{fig:behavior-fifty} (d)--(f) for details.

\spara{Additional datasets.}
To compare our algorithms across more datasets, we also consider several
real-world graphs for which we synthetically generate the innate opinions, the
user--topic matrix~$\m+X$, and the influenc--topic matrix~$\m+Y$.

The real-world graphs that we consider are publicly available from the Network
Repository~\cite{RossiA15}. Our experiments were conducted on the largest
connected component of each dataset.  Table~\ref{tab:running-time-error} lists
the networks that we consider in increasing order of the number of nodes. The
largest network has more than two million nodes, while the smallest one has
4\,991 nodes.

\begin{table*}[t]
\centering
\caption{Results of Algorithm~\ref{alg:opinions} for computing~$\tzXC$ on
	real-world graphs. We report graph statistics, comparison of running times
	(in seconds) and approximation errors for computing $\tzXC$ \emph{exactly}
	and computing~$\tzXC$ \emph{approximately} using
	Algorithm~\ref{alg:opinions}. We use four innate opinion distributions
	(uniform, power-law, exponential, and a custom ``polarized'' distribution).
	The synthetic user--topic and influence--topic matrices~$\m+X$ and $\m+Y$
	were drawn from the distributions described in the text. We set $C=0.1$.
	This experiment was conducted in a Linux server with E5-2630 V4 processor
	(2.2\,GHz) and 128\,GB memory. }
\label{tab:running-time-error}
\scalebox{0.7}{
\begin{tabular}{lccrrrrrrrrrrrr}
\toprule
\multirow{3}{*}{\textbf{Graph}} & \multirow{3}{*}{$n$} &
\multirow{3}{*}{$m$} & \multicolumn{12}{c}{\textbf{Running time\,(s)
	of evaluating $\zXC$ (\Exact) and $\tzXC$ (\Approx) with
		Algorithm~\ref{alg:opinions}, and approximation error\,(\Error, $\times
				10^{-8}$)}\smallskip} \\
 &  &  & \multicolumn{3}{c}{\Uniform} & \multicolumn{3}{c}{\Powerlaw} & \multicolumn{3}{c}{\Exponential} & \multicolumn{3}{c}{\Polarized} \\ \cmidrule(lr){4-6}  \cmidrule(lr){7-9} \cmidrule(lr){10-12} \cmidrule(lr){13-15}
 &  &  & \multicolumn{1}{c}{\Exact} & \multicolumn{1}{c}{\Approx} & \multicolumn{1}{c}{\Error} & \multicolumn{1}{c}{\Exact} & \multicolumn{1}{c}{\Approx} & \multicolumn{1}{c}{\Error} & \multicolumn{1}{c}{\Exact} & \multicolumn{1}{c}{\Approx} & \multicolumn{1}{c}{\Error} & \multicolumn{1}{c}{\Exact} & \multicolumn{1}{c}{\Approx} & \multicolumn{1}{c}{\Error} \\ \hline
\Erdos & 4,991 & 7,428 & 10.08 & 0.58 & 0.0108 & 9.57 & 0.49 & 0.0794 & 9.71 & 0.51 & 0.0794 & 9.71 & 0.51 & 0.0876 \\
\Advogato & 5,054 & 39,374 & 10.26 & 1.02 & 0.0180 & 10.26 & 1.12 & 0.1134 & 10.16 & 1.15 & 0.1134 & 10.16 & 1.15 & 0.2462 \\
\PagesGovernment & 7,057 & 89,429 & 26.17 & 1.76 & 0.0040 & 25.72 & 1.77 & 0.1553 & 25.94 & 1.71 & 0.1553 & 25.94 & 1.71 & 0.0597 \\
\WikiElec & 7,066 & 100,727 & 26.41 & 1.78 & 0.0057 & 26.12 & 1.53 & 0.0064 & 26.12 & 1.62 & 0.0064 & 26.12 & 1.62 & 1.0133 \\
\HepPh & 11,204 & 117,619 & 98.00 & 2.56 & 0.0095 & 97.89 & 2.73 & 0.1245 & 98.04 & 2.61 & 0.1245 & 98.04 & 2.61 & 0.0138 \\
\Anybeat & 12,645 & 49,132 & 140.91 & 1.82 & 0.0010 & 140.96 & 1.87 & 0.0132 & 141.75 & 2.07 & 0.0132 & 141.75 & 2.07 & 0.0073 \\
\PagesCompany & 14,113 & 52,126 & 195.51 & 2.45 & 0.0099 & 195.93 & 2.52 & 0.0154 & 194.27 & 2.46 & 0.0154 & 194.27 & 2.46 & 0.0473 \\
\AstroPh & 17,903 & 196,972 & 400.09 & 4.53 & 0.0282 & 398.54 & 4.86 & 0.0017 & 402.21 & 4.78 & 0.0017 & 402.21 & 4.78 & 0.0075 \\
\CondMat & 21,363 & 91,286 & 674.03 & 4.04 & 0.0011 & 669.62 & 4.05 & 0.1119 & 672.78 & 4.28 & 0.1119 & 672.78 & 4.28 & 0.0189 \\
\Gplus & 23,613 & 39,182 & 902.86 & 3.06 & 0.0002 & 919.68 & 2.62 & 0.0047 & 905.09 & 2.59 & 0.0047 & 905.09 & 2.59 & 0.0102 \\
\Brightkite & 56,739 & 212,945 & 13864.33 & 14.71 & 0.0007 & 13366.60 & 14.04 & 0.0119 & 14012.49 & 16.03 & 0.0119 & 14012.49 & 16.03 & 0.0720 \\
\Themarker & 69,317 & 1,644,794 & --- & 37.06 & --- & --- & 36.34 & --- & --- & 37.81 & --- & --- & 36.84 & --- \\
\Slashdot & 70,068 & 358,647 & --- & 16.01 & --- & --- & 14.76 & --- & --- & 14.45 & --- & --- & 14.47 & --- \\
\BlogCatalog & 88,784 & 2,093,195 & --- & 43.20 & --- & --- & 41.67 & --- & --- & 43.90 & --- & --- & 43.39 & --- \\
\WikiTalk & 92,117 & 360,767 & --- & 16.53 & --- & --- & 16.23 & --- & --- & 15.89 & --- & --- & 16.78 & --- \\
\Gowalla & 196,591 & 950,327 & --- & 56.87 & --- & --- & 51.58 & --- & --- & 53.04 & --- & --- & 53.13 & --- \\
\Academia & 200,167 & 1,022,440 & --- & 63.00 & --- & --- & 63.95 & --- & --- & 60.47 & --- & --- & 61.90 & --- \\
\GooglePlus & 201,949 & 1,133,956 & --- & 47.89 & --- & --- & 48.57 & --- & --- & 47.38 & --- & --- & 48.10 & --- \\
\Citeseer & 227,320 & 814,134 & --- & 46.57 & --- & --- & 47.23 & --- & --- & 46.45 & --- & --- & 46.76 & --- \\
\MathSciNet & 332,689 & 820,644 & --- & 68.91 & --- & --- & 62.26 & --- & --- & 67.23 & --- & --- & 61.37 & --- \\
\TwitterFollows & 404,719 & 713,319 & --- & 44.10 & --- & --- & 41.91 & --- & --- & 42.19 & --- & --- & 43.16 & --- \\
\Delicious & 536,108 & 1,365,961 & --- & 108.95 & --- & --- & 112.19 & --- & --- & 115.72 & --- & --- & 126.50 & --- \\
\YoutubeSnap & 1,134,890 & 2,987,624 & --- & 273.09 & --- & --- & 271.67 & --- & --- & 262.28 & --- & --- & 262.05 & --- \\
\Flickr & 1,624,992 & 15,476,835 & --- & 858.09 & --- & --- & 857.98 & --- & --- & 858.73 & --- & --- & 904.29 & --- \\
\Flixster & 2,523,386 & 7,918,801 & --- & 663.40 & --- & --- & 674.12 & --- & ---  & 644.32 & --- & --- & 653.05 & --- \\ 
\bottomrule
\end{tabular}}
\label{tab:optimization_simulation}
\end{table*}

Next, we consider four distributions to generate the innate opinions: uniform,
power-law, exponential, and a custom ``polarized'' distribution.  For the first
three distributions, we use the same parameter setting as
\citet{xu2021fast}. Note that they compute innate opinion $s\in[0,1]$ and here
we rescale the innate opinions to $[-1,1]$.  In the ``polarized'' distribution,
we mimic the opinion distribution from \TwitterFive and \TwitterFifty in
Figure~\ref{fig:opinions-twitter}, where the innate opinions tend to be
concentrated at the two opposite extremes, while sparsely distributed around
the middle. Thus, here we generate ``polarized'' opinions as follows.
For each node~$i$, we generate a value~$x_i$ based on the exponential opinion
distribution from above. 
Now for the first $n/2$ nodes we set their innate opinion to~$\v+s_i = x_i$ and for
the remaining $n/2$ nodes we set their opinion to $\v+s_i = 1-x_i$. Then we
rescale $\v+s$ such that all opinion are in $[-1,1]$.

We also compute synthetic user--topic matrices~$\m+X$ and influence--topic
matrices~$\m+Y$ by simulating properties of \TwitterFive and \TwitterFifty. 
More concretely, for \TwitterFifty we visualize the distribution of elements in
$\m+X$ and $\m+Y$ in Fig.~\ref{fig:X-Y-distribution}. It shows that the entries
in $\m+X$ and $\m+Y$ follow a power-law distribution.

\begin{figure}[t]
	\centering
	\begin{tabular}{cc}
	\includegraphics[width=0.4\columnwidth]{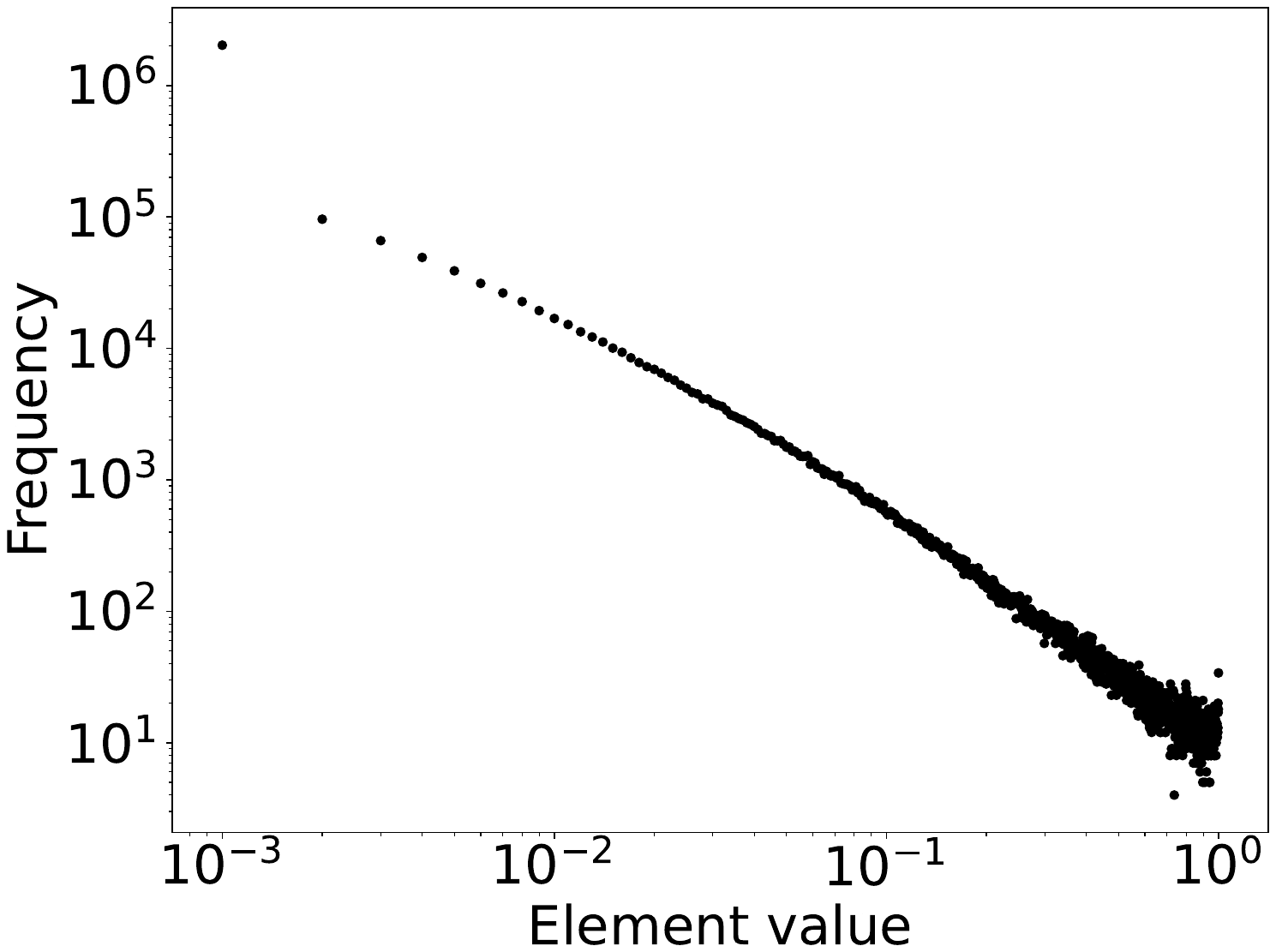} &
	\includegraphics[width=0.4\columnwidth]{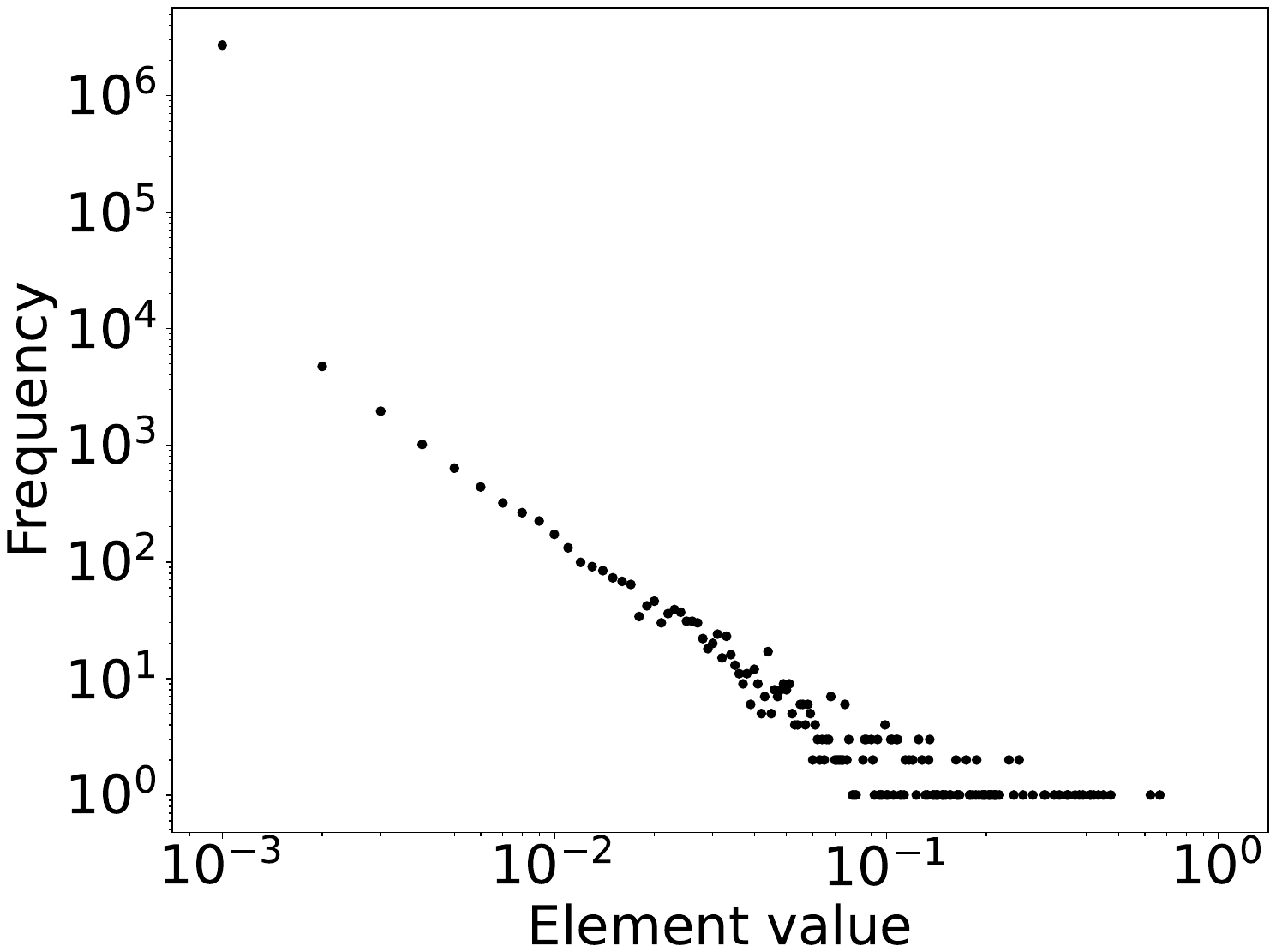} \\
	(a) User--topic matrix $\m+X$ & (b) Influence--topic matrix $\m+Y$ \\
	\end{tabular}
	\caption{
	Distribution of elements in $\m+X$ and $\m+Y$ in \TwitterFifty. 
	}
	\label{fig:X-Y-distribution}
\end{figure}

To generate $\m+X$ we proceed as follows. For each row~$\m+X_i$ that we generate
synthetically, we sample the entries~$\m+X_{ij}$ from a power-law distribution
with $\alpha=2.5$ (this value of $\alpha$ was also used in~\cite{xu2021fast}).
We control the sparsity of the matrix by removing elements with a value smaller
than 0.25 and rescaling the result such that~$\m+X_i$ is row-stochastic; we
chose the value~0.25 to match the sparsity of of our real-world matrices from
\TwitterFifty.

To generate $\m+Y$, we first observe from Fig.~\ref{fig:behavior-fifty}
that~$\tau_{j,\v+s}$, the weighted average of the innate opinions for
each topic~$j$, ranges from -0.65 to 0.65 and the majority of topics are located
around 0.  Inspired by this fact, we construct the influence--topic
matrix~$\m+Y$ such that the value $\tau_{j,\v+s}$ are spread across the opinion
spectrum (similar to the real-world behavior). More concretely, we equally
divide the opinion spectrum $[-1,1]$ into $d$~chunks and we assign a
weight~$w_i$ to each chunk~$i$.  Now, for each topic~$j$ we first sample its
\emph{bias}. That is, we sample a chunk~$i$ with probability proportional to the
weight~$w_i$ and then all users of topic~$j$ have their innate opinion
from chunk~$i$.  If $V_i$ denotes the set of users with innate opinion in
chunk~$i$ and $n$ is the number of all users, then we pick $0.02n$~users from
$V_i$ uniformly at random and for each $u\in V_i$, we set $Y_{ju}$ using a
power-law distribution with $\alpha=2.5$.  Finally, we rescale the result such
that $\m+Y_i$ is row-stochastic.  In our synthetic experiments we used $d=3$ and
$w=[0.3,0.4,0.3]$.

We note that this way of $\m+Y$ was crucial to obtain our experimental results
on synthetic data.  Initially, we simply picked $0.02n$~users for each
topic uniformly at random. However, this resulted in all $\tau_{j,\v+s}$ being
very close to~$0$, which is not the behavior that we saw in our real-world
datasets.  This also had the side-effect hat our optimization algorithms could
not reduce the polarization and disagreement significantly. 

\subsection{Additional experiment results}
\label{app:additional-exp}
Now we report additional experimental results, including a running time
analysis.

\spara{Impact of the learning rate.}
We first study the impact of the learning rate on the
convergence of \ouralgo.
Our theoretical analysis suggests using  learning rate
$L=({8CW}/{\sqrt{n}})  \norm{\v+s}_2  \norm{\m+Y}_2^2$, 
which is very large in practice and will result in slow convergence.
Thus, we study the convergence of \ouralgo for different learning rates, in particular, we test $L=10,10^2,10^3,10^4$.
The results for \TwitterFive and \TwitterFifty are shown 
in Figures~\ref{fig:learning_rates}(a) and \ref{fig:learning_rates}(b), respectively.
We observe that even with $L=10$, \ouralgo converges to the
same objective function value as for much larger values of~$L$, 
and it converges much faster.
Therefore, for the rest of our experiments we will use $L=10$.

\begin{figure}[t]
  \centering
    \begin{tabular}{cc}
         \includegraphics[width=0.45\columnwidth] {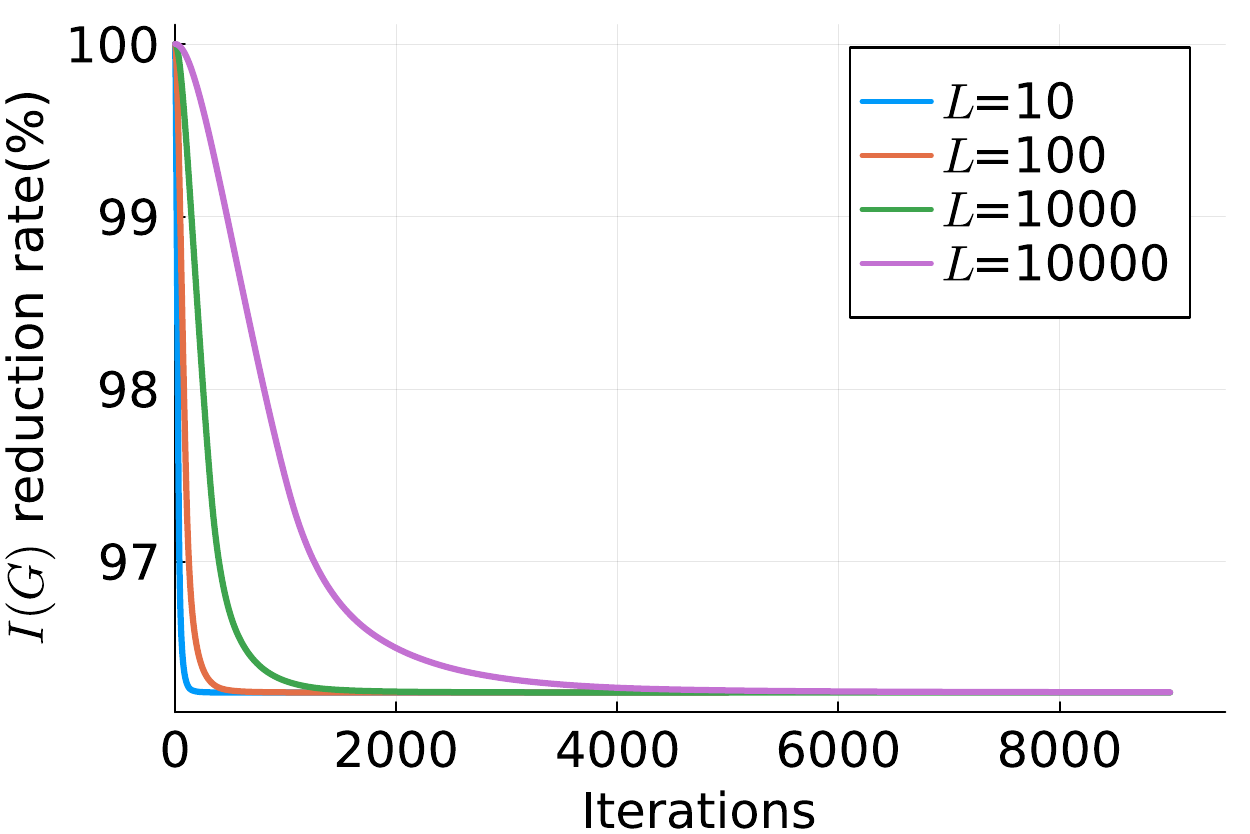} & \includegraphics[width=0.45\columnwidth]{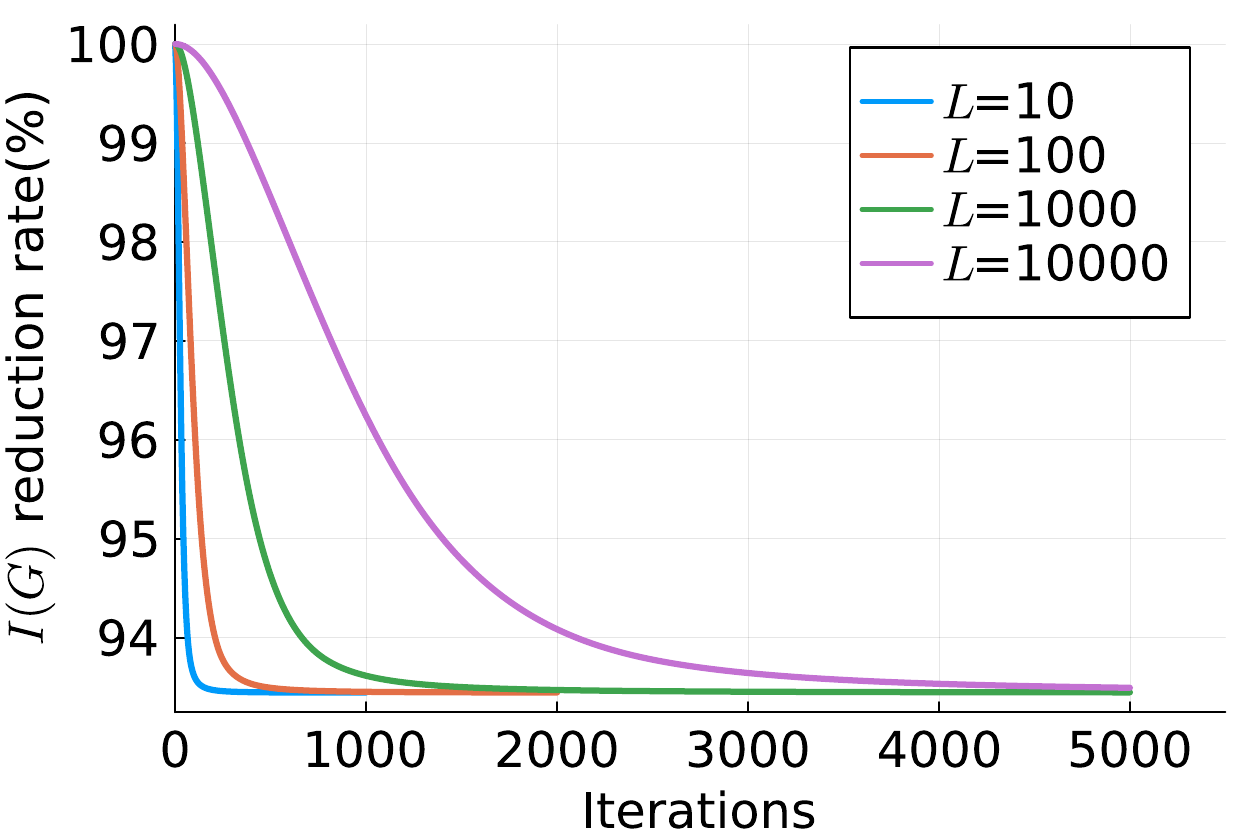} \\ 
         (a) \TwitterFive  & (b) \TwitterFifty \\
    \end{tabular}
    \caption{\label{fig:learning_rates}
    Convergence of \ouralgo for different learning rates on two \twitter datasets ($\vartheta=0.1$, $C=0.1$). The $y$-axis shows the reduction ratio ${f(\m+X_{\ALG})}/{f(\m+X)}$.}
    \label{fig:learning_rates}
\end{figure}

\spara{Approximating Expressed Opinions~$\tzXC$.}
Table~\ref{tab:running-time-error} reports running time and approximation error
of Algorithm~\ref{alg:opinions} for computing~$\tzXC$ on different real-world
graphs. We compare against the exact solution~$\zXC$ and note that we cannot
compute $\zXC$ for the 14~largest graphs due to the high running time of
computing the exact solution. We observe that Algorithm~\ref{alg:opinions} is
orders of magnitude faster than the na\"ive computation of $\zXC$ and its error
is negliglible in practice (note that errors are typically less than $10^{-8}$).

We also visualize the running times from Table~\ref{tab:running-time-error} for
uniformly distributed innate opinions in Fig.~\ref{fig:running-time}(a). We
observe that the running time grows linearly with the number of nodes.

\begin{figure}[t]
	\centering
	\begin{tabular}{ccc}
	\includegraphics[width=0.33\columnwidth]{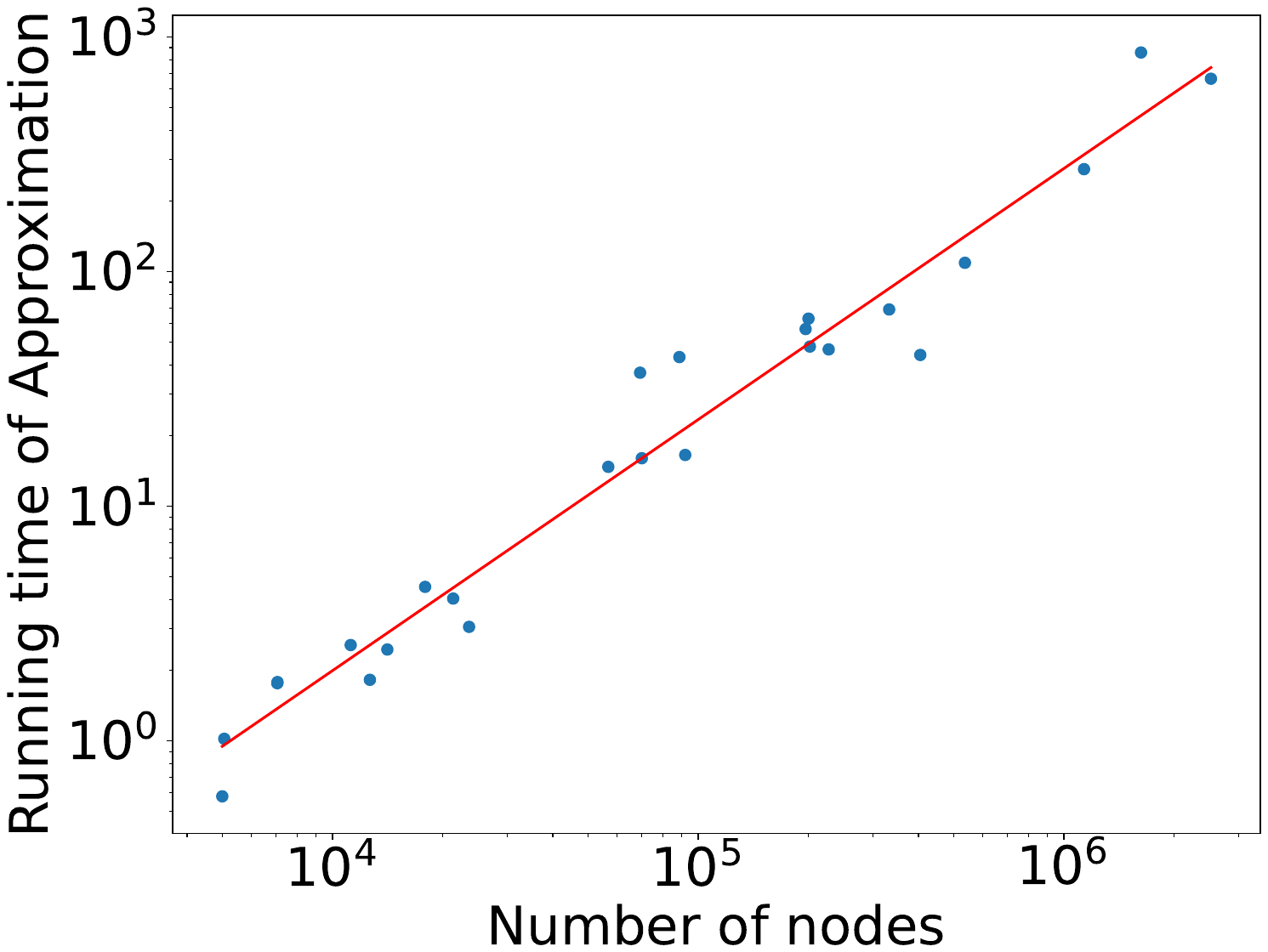} &
	\includegraphics[width=0.33\columnwidth]{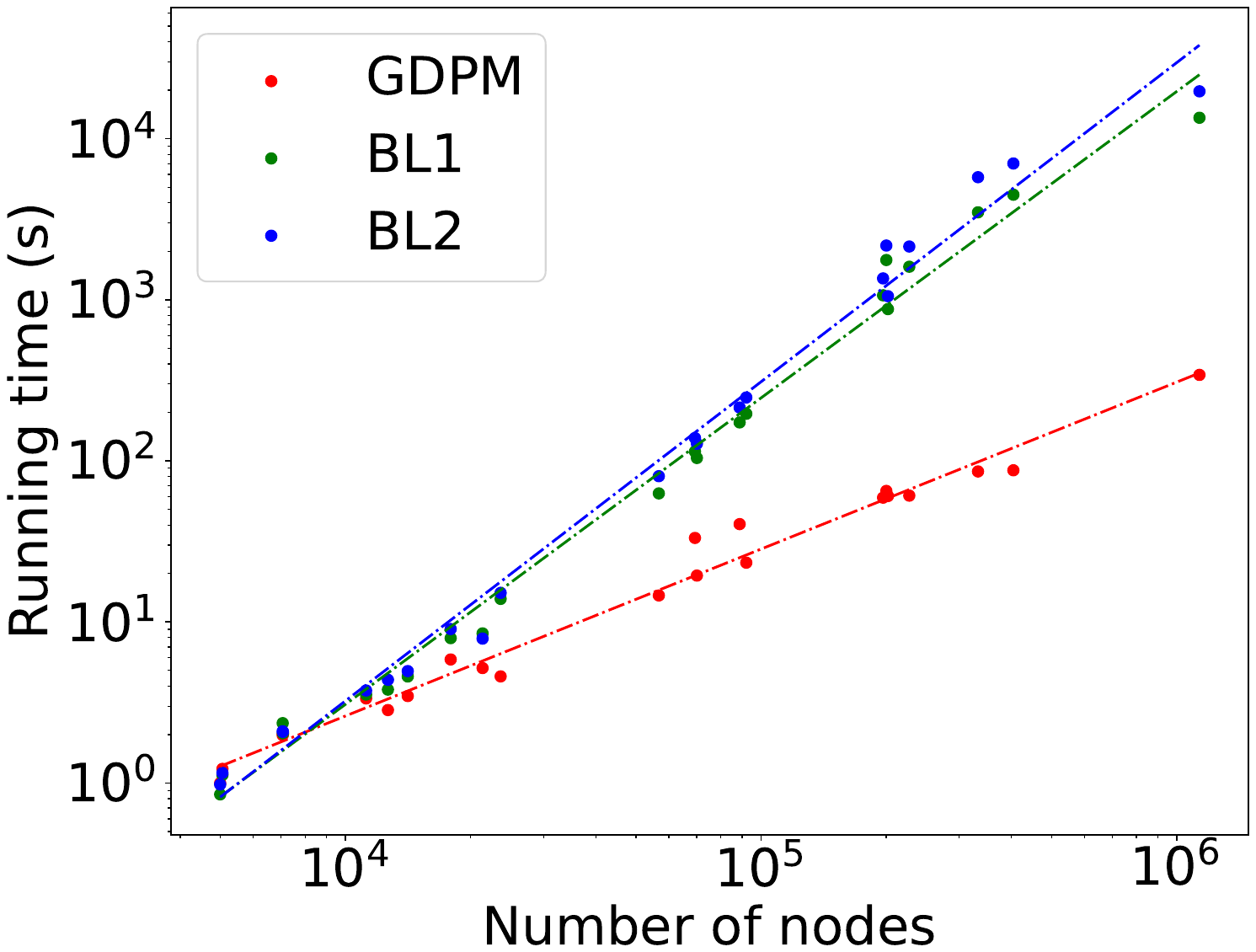} &
    \includegraphics[width=0.33\columnwidth]{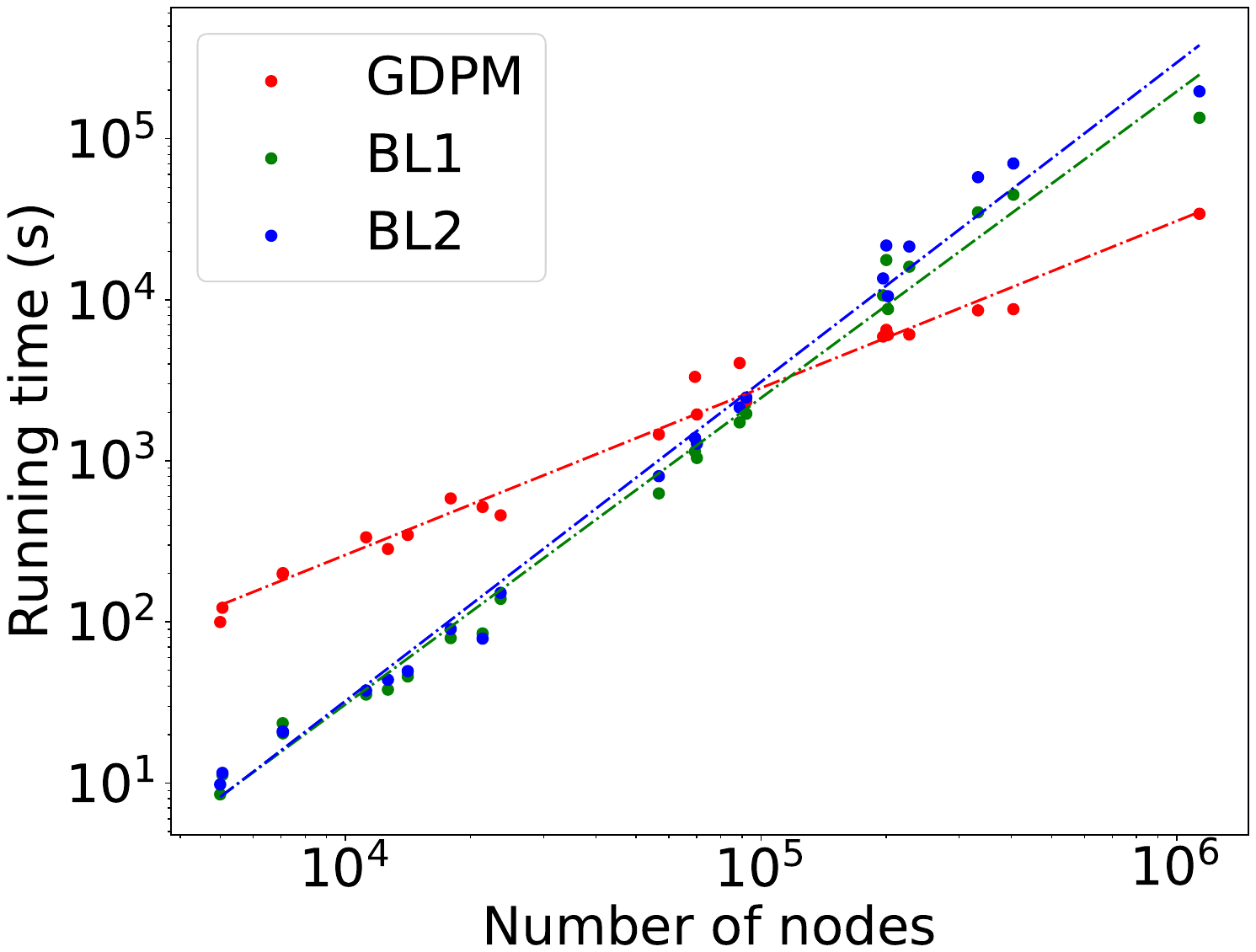} \\
	(a)  &
	(b)  & (c)   \\
	\end{tabular}
	\caption{
	Running time of algorithms in seconds. Each marker corresponds to one of the
	datasets from Table~\ref{tab:running-time-error}. We fit a linear
	regression to indicate trends. Plot~(a) shows the running time of
	Algorithm~\ref{alg:opinions} for computing~$\tzXC$.
	Plot~(b) shows the running time for a single iteration of \ouralgo, \blone
	and \bltwo. Plot~(c) shows the total running times for 100 iterations of
	\ouralgo and 10 iterations of \blone and \bltwo.
	}
	\label{fig:running-time}
\end{figure}

We also note that in our experiments the
error incurred on our objective function by the approximate opinions was very
small, with typically $\abs{\widetilde{f} - f(\m+X)}/f(\m+X)<10^{-8}$, where
$\widetilde{f}$ is as in Corollary~\ref{cor:fast-objective-function}.

\spara{Running time analysis of the optimization algorithms.}
We start by comparing the running times of \ouralgo, \blone, and \bltwo, which
use Algorithm~\ref{alg:opinions} as a subroutine to compute approximate
opinions~$\tzXC$, with an implementation that computes exact opinions $\zXC$. We
report our results in Table~\ref{tab:running-time}. While on \TwitterFive, the
exact methods are still relatively fast, on \TwitterFifty we observe that the
algorithms with approximate opinions are faster by a factor of~300. In other
words, running \emph{all} 150~iterations of \ouralgo with approximate opinions
is faster than running \emph{a single} iteration with exact opinions.

\begin{table}[t]
\small
  \centering
\caption{Running times of the algorithms. We report the time (in seconds)
	required for a single iteration with \emph{approximate} expressed
	opinions~$\tzXC$ and with \emph{exact} expressed opinions~$\zXC$. We also
	report the \emph{total} running time over 150~iterations when using the
	approximate solver. Here, we set $L=10$, $C=0.1$, $\theta=0.1$ and $T=150$.}
\label{tab:running-time}
\begin{tabular}{lrrcrrr}
\toprule
\multirow{2}{*}{\textbf{Algorithm}} & \multicolumn{3}{c}{\TwitterFive} & \multicolumn{3}{c}{\TwitterFifty} \\
\cmidrule(lr){2-4} \cmidrule(lr){5-7}
 & \textbf{Approx} & \textbf{Exact} & \textbf{Total} & \textbf{Approx} & \textbf{Exact} & \textbf{Total} \\
 & (1 iter.) & (1 iter.) & (150 iter.) & (1 iter.) & (1 iter.) & (150 iter.) \\ \hline
\textbf{GDPM} & 0.21 & 0.15 & 31.95&7.67 & 1530.02 & 1150.06 \\
\textbf{BL 1} & 0.09 & 0.12 & 13.51 & \textbf{4.76} & 1501.21 & \textbf{713.97} \\
\textbf{BL 2} & \textbf{0.08} & 0.13 & \textbf{12.67} & 4.98 & 1485.77 &747.07 \\ 
\bottomrule
\end{tabular}
\end{table}

Furthermore, in Fig.~\ref{fig:running-time}(b) we visualize the running time of
a single iteration of \ouralgo, \blone, and \bltwo and in
Fig.~\ref{fig:running-time}(c) we plot the total time for 100~iterations of
\ouralgo and 10~iterations of \blone and \bltwo.  The figures show that for all
algorithms their running time grows linearly in the number of nodes. However,
note that a single iteration of \ouralgo is faster than the baselines,
particularly on large graphs. The reason is that for each row, \blone and \bltwo
need to compute the topic indices~$j$ and~$j'$ which shall be favored and
penalized (see Algorithm~\ref{alg:baseline}), which is costly; on the other
hand, \ouralgo computes the gradient only once and the projection operation in
\ouralgo for updating each row is highly efficient. This has the effect that as
the size of the graphs increases, \ouralgo becomes more efficient than the
baselines in terms of total running time, even though it performs 10~times more
iterations.

\spara{Comparison with off-the-shelf convex solver.}
Figure~\ref{fig:convex_GDPM_running_time} reports the comparison of the running
time between \ouralgo and \cjl to solve Problem~\ref{eq:problem}. \cjl is a
popular off-the-shelf convex optimization tool written in Julia and we use it
with the SCS solver. It shows that
\ouralgo outperforms the off-the-shelf solver in every experiment.

Our experiments are conducted on random graphs and the probability of creating
an edge in the random graph is set to 0.5. We set the available memory resource
for the experiments to be 102 GB. Running \cjl on a graph with more than 500~nodes
exceeds the memory constraint. Therefore, we only report the results of graphs
with 50 to 500 nodes.  We generate synthetic opinions and user-to-topic matrix
$\m+X$ and influence--topic matrix $\m+Y$ following the same method described in
section \ref{sec:exp-setting}. The number of
topics is set to 10 for all matrices. We set the learning rate $L=100$ and run
$2000$ iterations for \ouralgo. In the \cjl experiment, the algorithm runs
until the problem is solved or reaches the maximum iteration limit. In all
experiments, \cjl and \ouralgo return the same objective value on the same
setting after optimization with a precision of $10^{-4}$. 

The reason for the stark contrast between the two algorithms is as follows: Any
non-tailored algorithm that computes our objective function or the gradients of
our problem, must compute the matrix $(\m+I + \m+L + \LXC)^{-1}$, which is a
dense matrix with $\Omega(n^2)$ non-zero entries. In fact, even writing down
$\LXC$ might take space $\Omega(n^2)$ based on the structure of the matrices
$\m+X$ and $\m+Y$. This is the reason why the baseline \cjl uses so much memory
in our experiments. Our algorithms bypass all of these issues, because our
efficient routine for estimating the opinions from
Proposition~\ref{prop:fast-opinions} allows us to efficiently estimate the
objective function (Corollary~\ref{cor:fast-objective-function}) and the
gradient (Proposition~\ref{prop:gradient}). Because of this, all of our
algorithms use near-linear space in the size of the input graph. This is the
reason why any non-tailored optimization approach is going to fail, unless it
provides us an explicit way for efficiently computing the gradient and the
objective function.

\begin{figure}
    \centering
    \includegraphics[width=0.4\columnwidth]{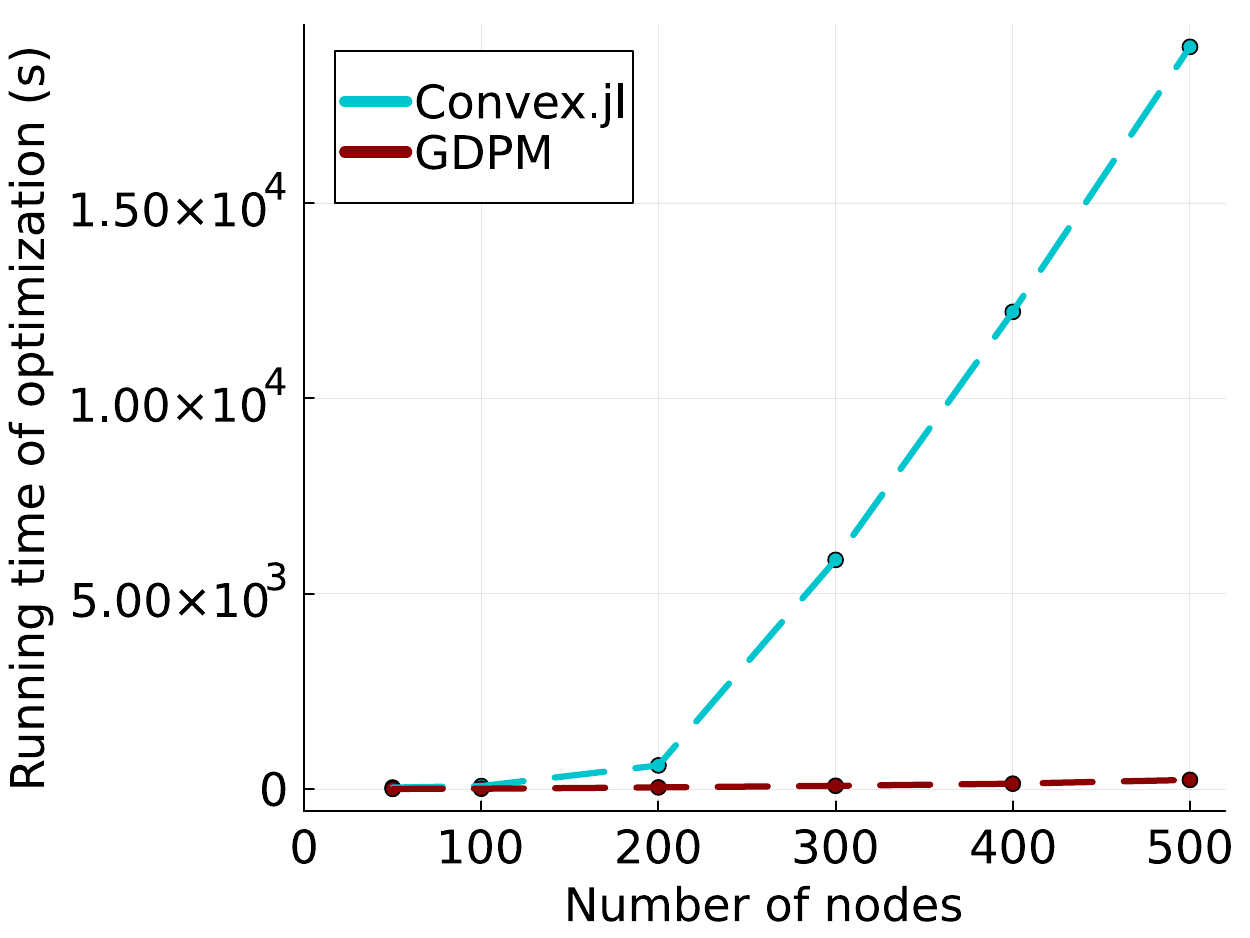}
    \caption{Running time comparison of \cjl and \ouralgo on random graphs with
		synthetic opinions. We set \ouralgo to run 2,000 iterations and \cjl to run until the problem is solved. The objective value returned by both methods coincides with a precision of $10^{-4}$. We set $C=0.1$, $\theta=0.1$ for all experiments.}
    \label{fig:convex_GDPM_running_time}
\end{figure}

\spara{Dependency of the convergence speed on $C$.}
Next, in Figure~\ref{fig:C_convergence} we study how the covergence speed of
\ouralgo depends on the parameter~$C$.  As we can see from the figure, the
number of iterations until convergence increases very slowly (if at all). In
particular, consider as a convergence condition the criterion that in two
consecutive iterations $i-1$ and $i$, the disagreement--polarization index changes
satisfy $I^{(i)}(G)/I^{(i-1)}(G) > 0.99999$. Then on \TwitterFifty, the number
of iterations needed to converge for the corresponding
$C=0.1,0.2,0.3,0.4$-values is given by 91, 94, 95, 99, 102.

\begin{figure}[t]
  \centering
    \begin{tabular}{cc}
         \includegraphics[width=0.45\columnwidth] {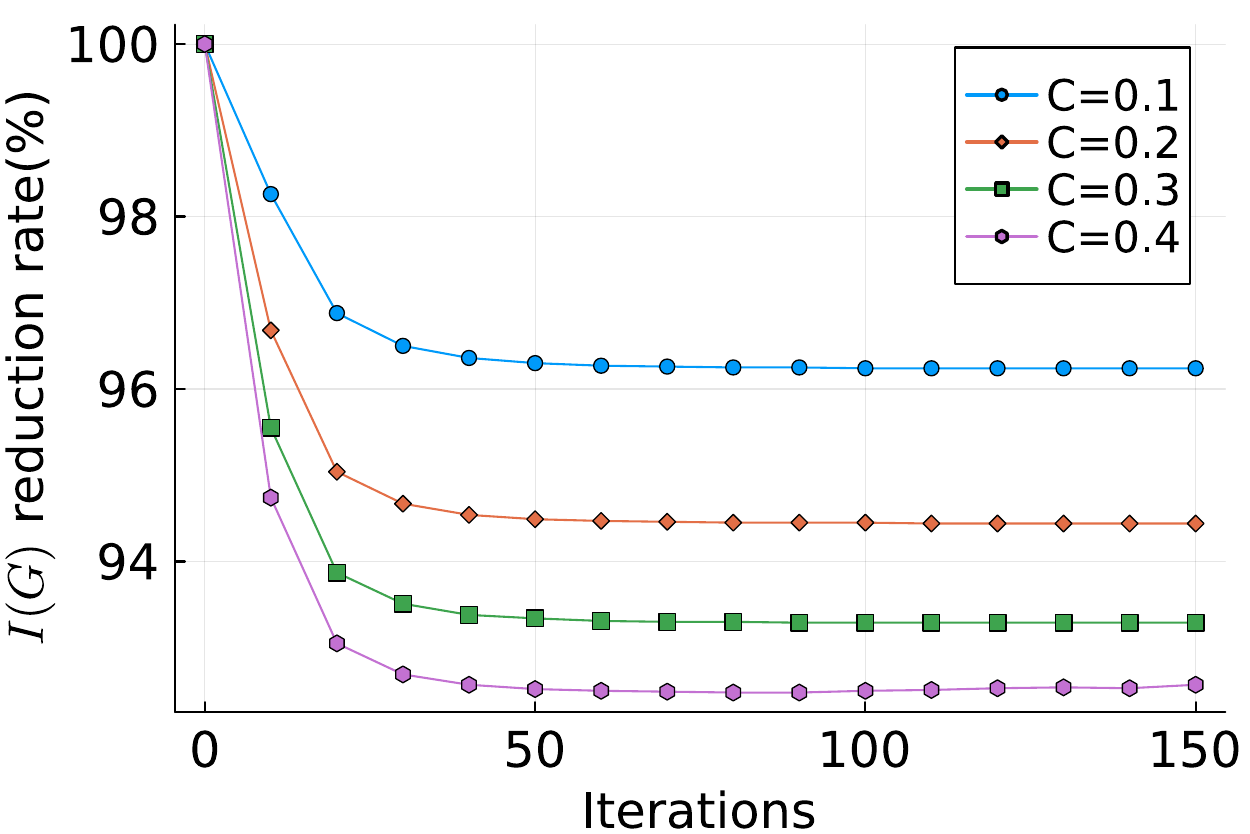} & \includegraphics[width=0.45\columnwidth]{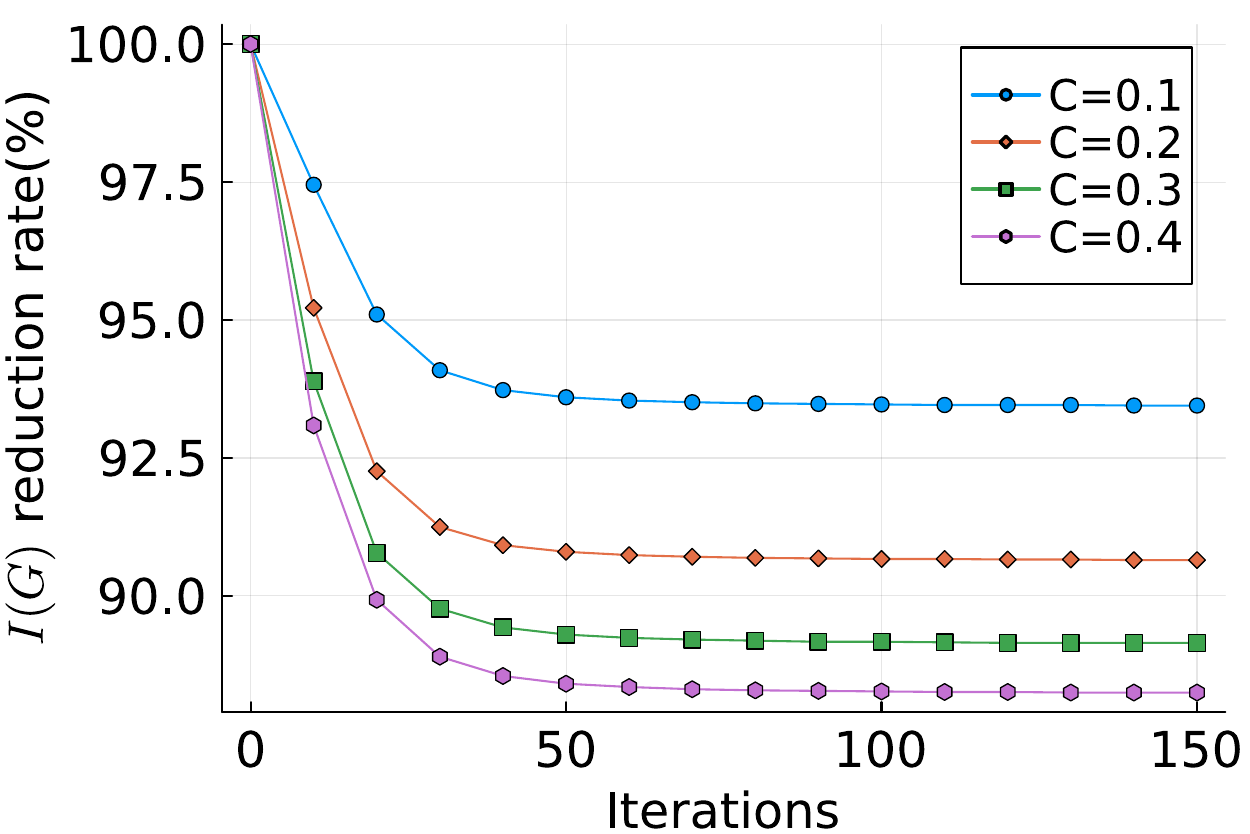} \\ 
         (a) \TwitterFive  & (b) \TwitterFifty \\
    \end{tabular}
    \caption{
	Convergence of \ouralgo for different values of $C$ on the two \twitter datasets
	($\vartheta=0.1$). The $y$-axis shows the reduction ratio
	${f(\m+X_{\ALG})}/{f(\m+X)}$. We vary $C \in \{0.1, 0.2, 0.3, 0.4 \}$.}
    \label{fig:C_convergence}
\end{figure}

\spara{Simultaneously changing $C$ and $\theta$.}
Next, we run experiments on \TwitterFive and \TwitterFifty, in which we
simultaneously change the parameters $C$ and $\theta$.  We report the results in
Figure~\ref{fig:vary_C_theta}.  We observe that the disagreement--polarization
index decreases both as a function of $C$, as well as for~$\theta$. Furthermore,
across all $C$-values that we consider, the decrease for larger values of
$\theta$ is comparable.

\begin{figure}[t]
  \centering
    \begin{tabular}{cc}
         \includegraphics[width=0.45\columnwidth]{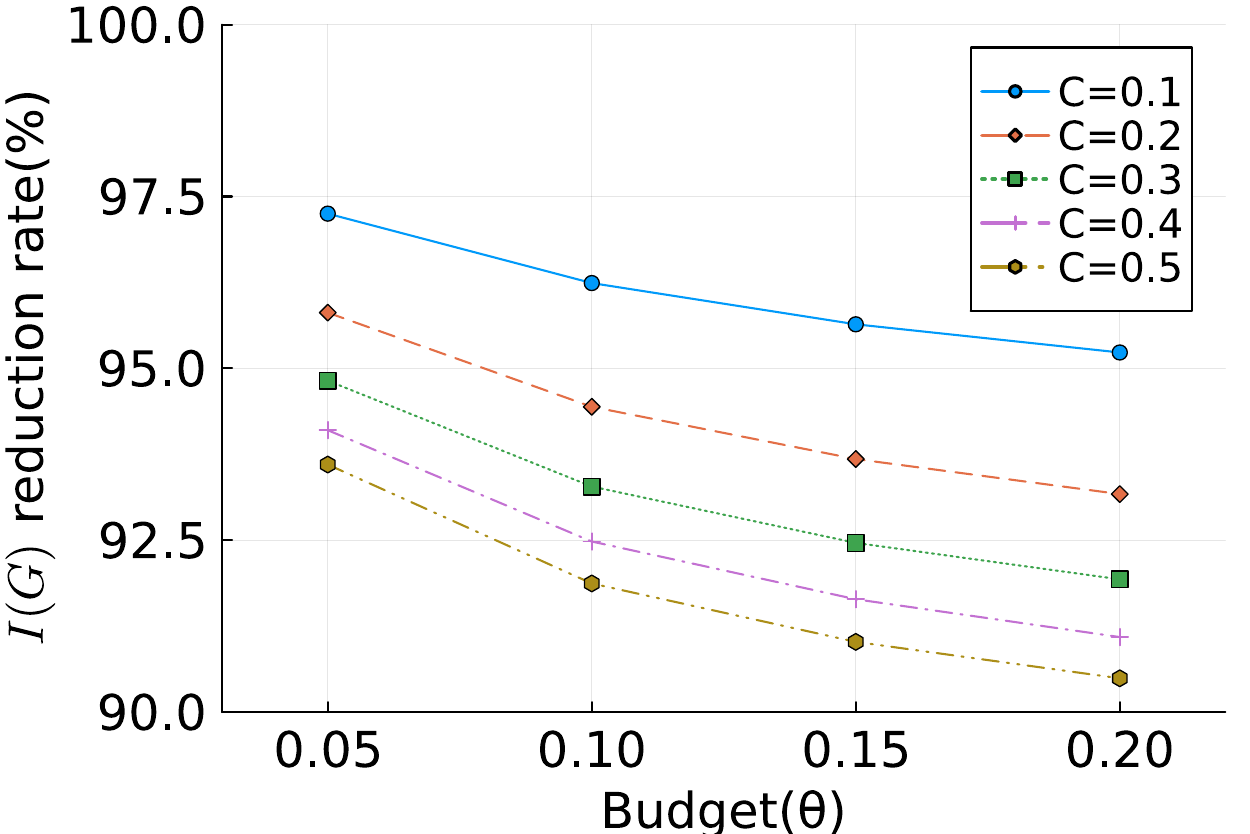} & \includegraphics[width=0.45\columnwidth]{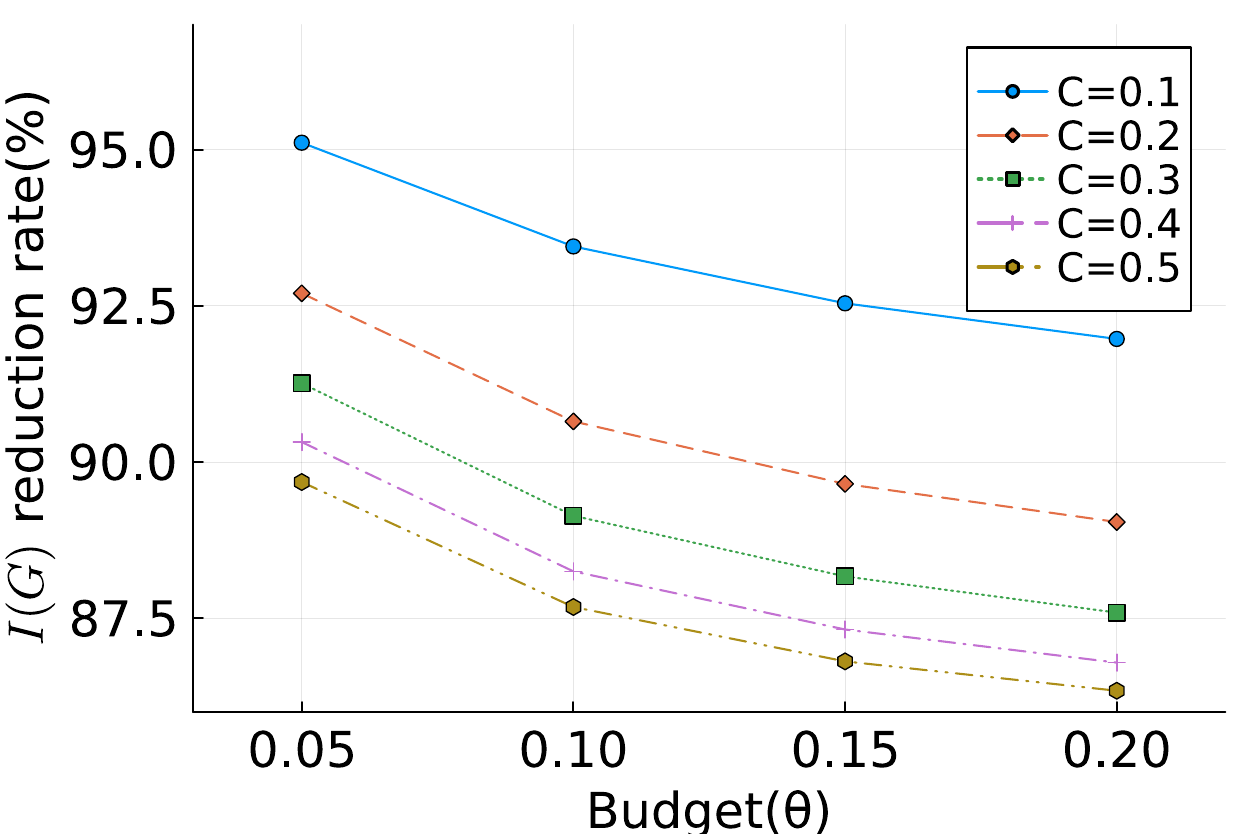} \\ 
         (a) \TwitterFive  & (b) \TwitterFifty \\
    \end{tabular}
    \caption{
    Reduction of the disagreement--polarization index on the two \twitter datasets
	for \ouralgo. The $y$-axis shows the reduction ratio ${f(\m+X_{\ALG})}/{f(\m+X)}$.
	We vary $C \in\{0.1,0.2,0.3,0.4\}$ and $\vartheta \in \{0.05,0.1,0.15,0.2\} $. }
    \label{fig:vary_C_theta}
\end{figure}

\spara{Node degree changes after optimization.}
To understand how the node degrees are affected by our model, we consider
the node degree increase rate after optimization in \TwitterFive and
\TwitterFifty. We report the results in Figure~\ref{fig:node-degree-increase}
for $C=10\%$.
In Figures~\ref{fig:node-degree-increase}~(a)--(b), users are ranked in
descending order by their corresponding influence score among all topics.
Formally, the influence score of a node~$u$ is given by
$\sum_{j=1}^k \m+Y_{ju}$.
The figures show that the user groups with the highest influence scores have the
largest standard deviations and higher means. In
Figures~\ref{fig:node-degree-increase}~(c)--(d), users are ranked in descending order by
their node degree in the original graph. The figures show that the mean of the increase rate
in groups with large node degrees is less than 10\% (until group 12 in
\TwitterFive and group 9 in \TwitterFifty).

\begin{figure}[t]
	\centering
	\begin{tabular}{cc}
	\includegraphics[width=0.4\columnwidth]{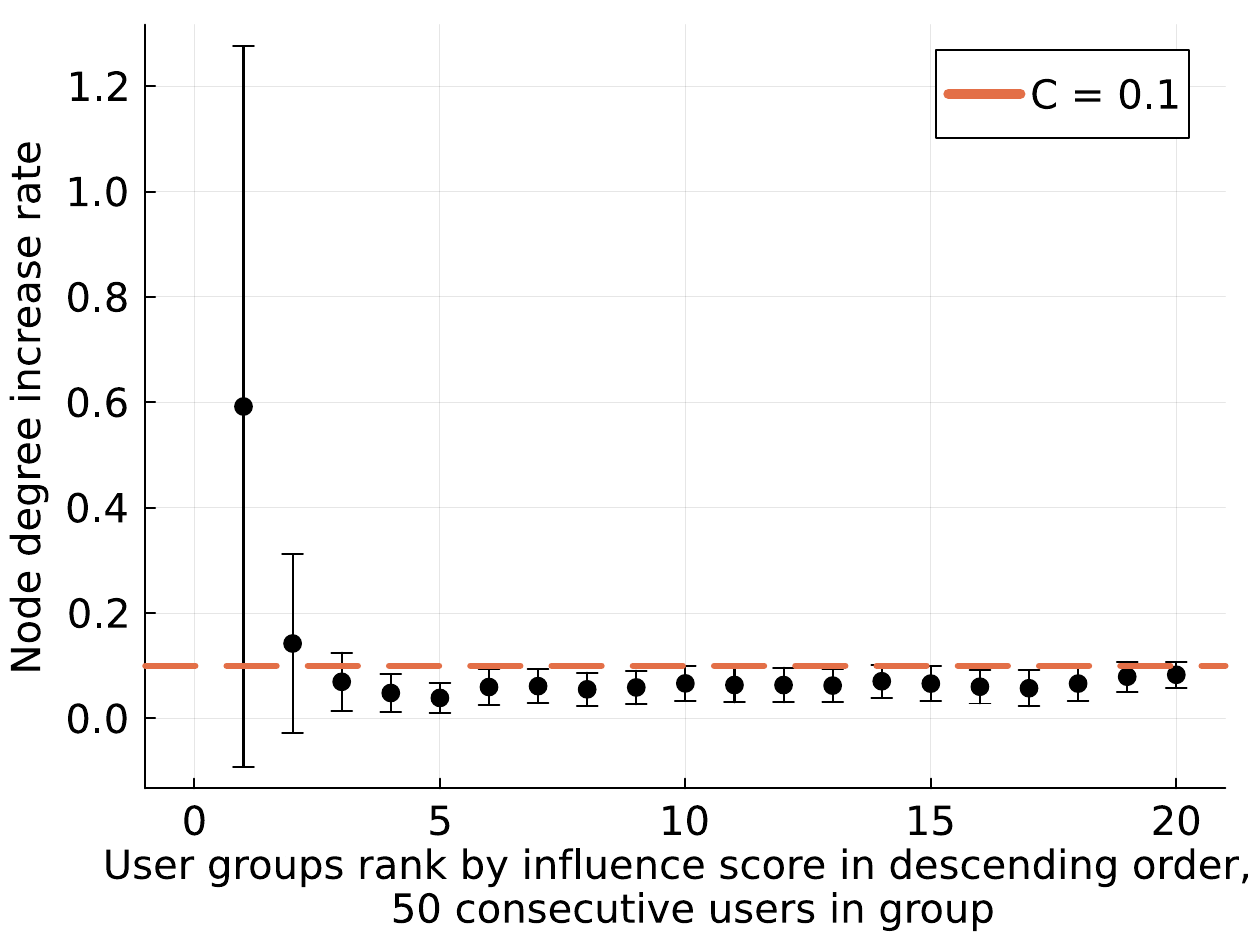} &
	\includegraphics[width=0.4\columnwidth]{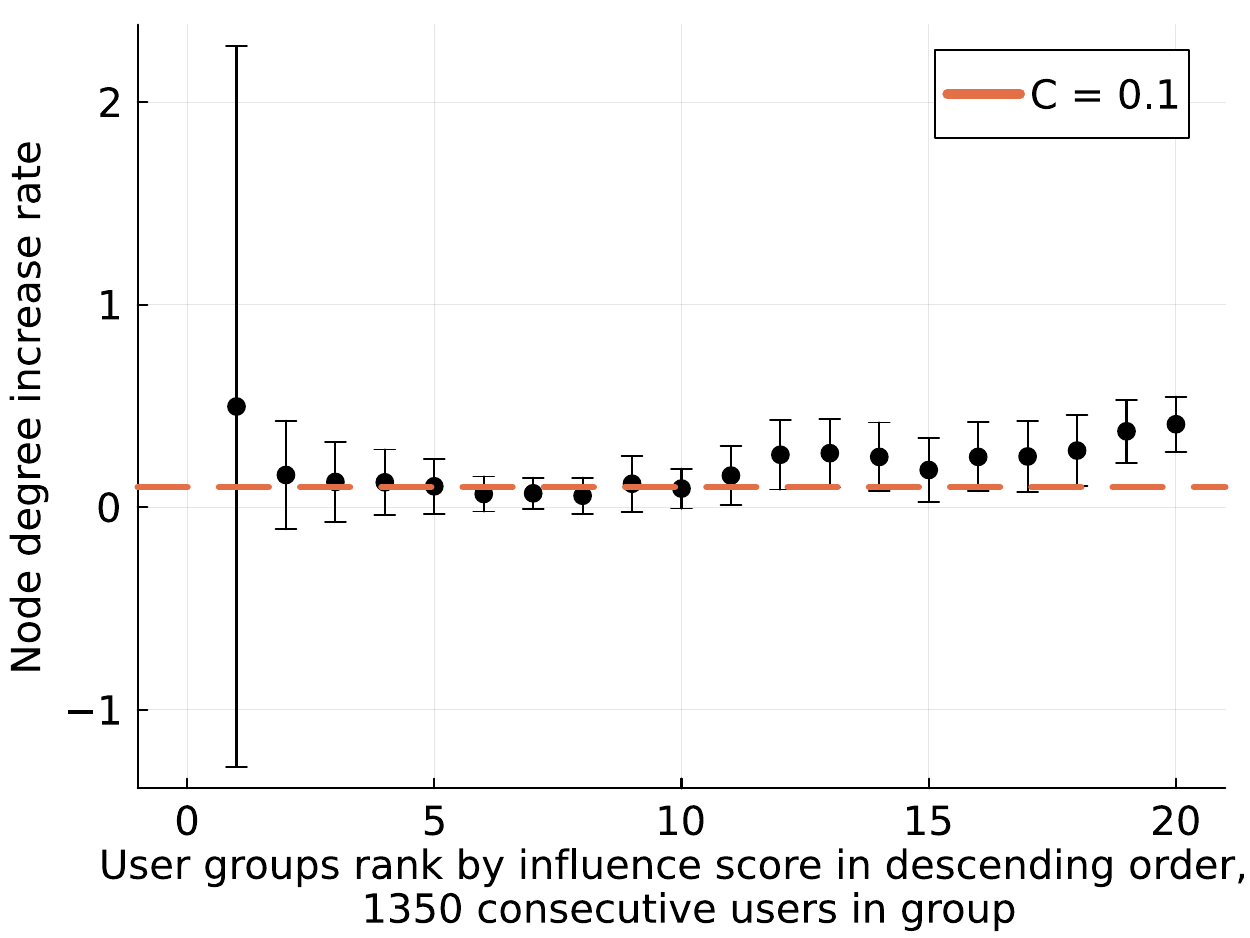} \\
	(a) \TwitterFive & (b) \TwitterFifty \\
 	\includegraphics[width=0.4\columnwidth]{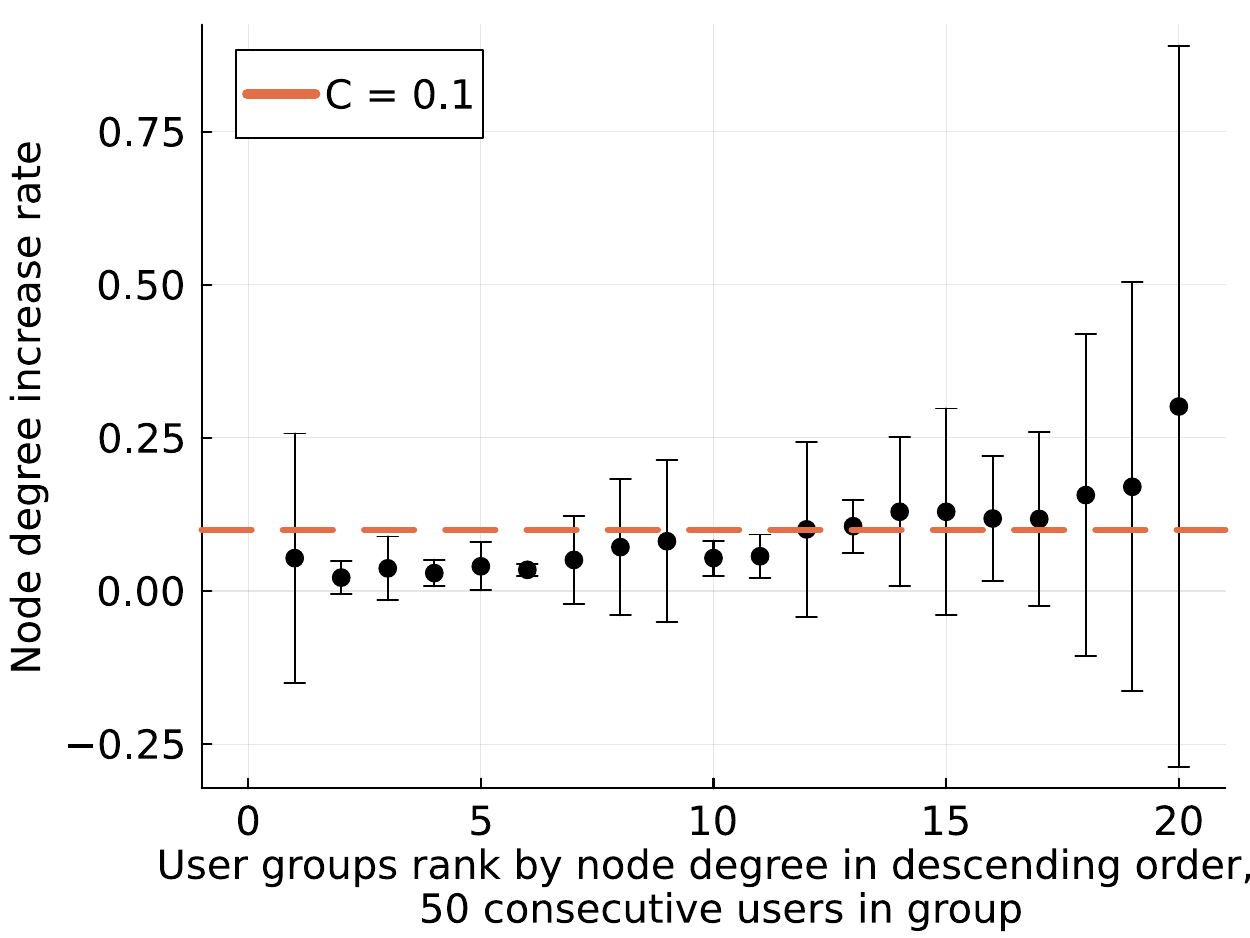} &
	\includegraphics[width=0.4\columnwidth]{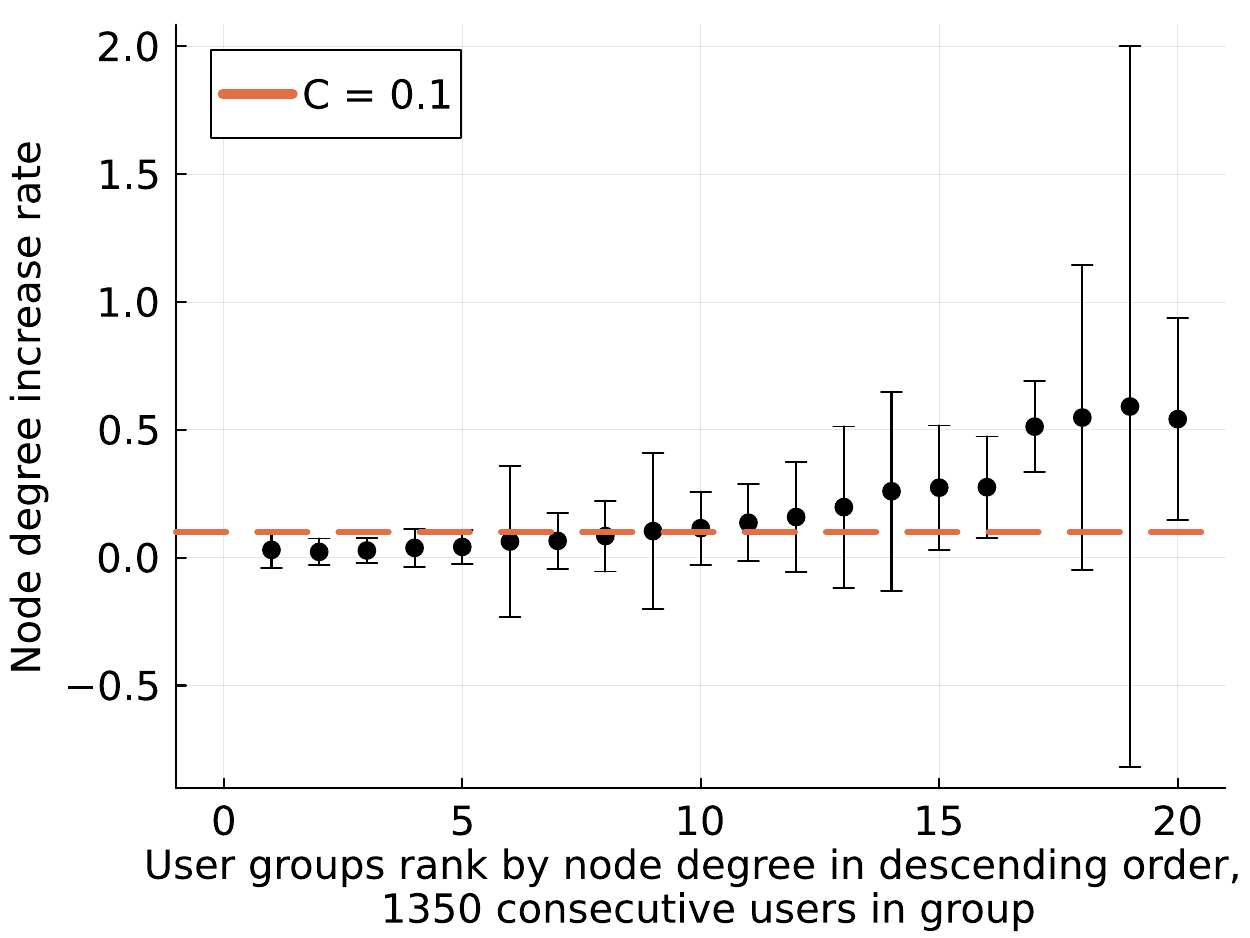} \\
	(c) \TwitterFive & (d) \TwitterFifty \\
	\end{tabular}
	\caption{Mean and standard deviation of node degree increase rate after
		optimization. The x-axis is user groups consisting of consecutive users
			in ranked lists in descending order, the y-axis is the node degree
			increase rate after adding the recommendation $\AXC$ with optimal
			$\m+X$. We set $C=0.1$ in all experiments (dashed line). In (a)--(b)
			users are ranked by influence scores among all topics. In (c)--(d)
			users are ranked by node degree in the original graph. 
	}
	\label{fig:node-degree-increase}
\end{figure}

\section{Omitted proofs}

\subsection{Preliminaries on linear algebra and optimization}
We start by defining additional notation and recalling some basic facts from
linear algebra and optimization.

We write $\lambda_i(\m+X)$ to denote the $i$-th eigenvalue of $\m+A$. Similarly,
$\sigma_i(\m+X)$ denotes the $i$-th singular value of $\m+X$. We will sometimes
also write $\lambda_{\min}(\m+X)$, $\lambda_{\max}(\m+X)$,
$\sigma_{\min}(\m+X)$ and $\sigma_{\max}(\m+X)$ to denote the smallest and
largest eigenvalues and singular values of $\m+X$, respectively.

Next, let us recall basic facts about matrix norms, where we let
$\m+X\in\mathbb{R}^{m\times k}$, $\m+Y\in\mathbb{R}^{k\times n}$ and
$\v+v\in\mathbb{R}^k$. Then we have that $\norm{\m+X}_2 \leq \norm{\m+X}_F$.
Furthermore, it holds that
$\norm{\m+X \m+Y}_2 \leq \norm{\m+X}_2 \cdot \norm{\m+Y}_2$, as well as 
$\norm{\m+X \m+Y}_F \leq \norm{\m+X}_2 \cdot \norm{\m+Y}_F$.
We also have $\norm{\m+X \v+v}_2 \leq \norm{\m+X}_2 \cdot \norm{\v+v}_2$.
Furthermore, we denote the Frobenius scalar product by
$\langle\m+X, \m+Y\rangle_F = \sum_{ij} \m+X_{ij} \m+Y_{ij}$.

If $\m+X,\m+Y\in\mathbb{R}^{n\times n}$ are invertible then observe that
$\m+X^{-1} - \m+Y^{-1} = 
\m+X^{-1}(\m+Y - \m+X)\m+Y^{-1}$, which based on the previous matrix inequalities
implies that 
$\norm{\m+X^{-1} - \m+Y^{-1}}_F
				\leq \norm{\m+X^{-1}}_2 \cdot
				\norm{\m+Y^{-1}}_2 \cdot
				\norm{\m+X - \m+Y}_F$.
Next, the Neumann series states that if $\norm{\m+X}_2 < 1$ then
$(\m+I - \m+X)^{-1} = \sum_{i=0}^\infty \m+X^i$.

The \emph{prox}-operator is given by
\begin{align}
\label{eq:prox}
	\prox_f(x)
	= \argmin_u \left\{ f(u) + \frac{1}{2}\norm{u - x}^2 \right\}.
\end{align}

Given a convex set~$Q$, we write $\delta_Q(x)\in\{0,\infty\}$ to denote its
indicator function, i.e.,
\begin{align*}
	\delta_Q(x) =
	\begin{cases}
		0, & x \in Q,\\
		\infty, & x\not\in Q.
	\end{cases}
\end{align*}
	
\subsection{Useful facts}

\begin{lemma}
\label{lem:eigenvalues-M}
	Let $\m+M = \m+D + \m+L$, where $\m+D\in\mathbb{R}^{n\times n}$ is a
	diagonal matrix with diagonal entries $\m+D_{ii}\geq 1$ and $\m+L$ is the
	Laplacian of a connected undirected weighted graph.
	Then all eigenvalues of $\m+M$ are at least $1$. Furthermore, for all
	$k\in\mathbb{N}$, the eigenvalues of $\m+M^{-k}$ are most 1 and
	$\sum_{i=1}^n \lambda_i(\m+M^{-k})\leq n$.
\end{lemma}
\begin{proof}
	Observe that all eigenvalues of $\m+D$ are at least $1$ since $\m+D$ is
	diagonal and $\m+D_{ii} \geq 1$ for all $i$. Furthermore, $\m+L$ is positive
	semidefinite and thus all eigenvalues of $\m+L$ are non-negative. Thus,
	Weyl's inequality implies that all eigenvalues of $\m+D$ are at least $1$.

	The claim about the eigenvalues of $\m+M^{-k}$ follows from the fact that
	the eigenvalues of $\m+M^{-k}$ are given by
	$\lambda_1(\m+M)^{-k},\dots,\lambda_n(\m+M)^{-k}$ and the above argument
	implies that all of these numbers are at most $1$. Summing over these
	eigenvalues gives their third claim of the lemma.
\end{proof}

\begin{lemma}
\label{lem:difference-hadamard}
	Suppose $\v+y, \v+z \in [-1,1]^n$ then
	$\norm{(\v+y \odot \v+y) - (\v+z \odot \v+z)}_2
		\leq 2 \norm{\v+y - \v+z}_2$.
\end{lemma}
\begin{proof}
	First, we note that for $a,b\in[-1,1]$ it holds that
	$(a^2-b^2)^2\leq 4(a-b)^2$ since if $a=b$ the inequality clearly holds and
	if $a\neq b$ then
	\begin{align*}
		\frac{(a^2-b^2)^2}{(a-b)^2}
   		= \left(\frac{a^2-b^2}{a-b}\right)^2
		= \left(\frac{(a+b)(a-b)}{a-b}\right)^2
		= (a+b)^2 \leq 4.
	\end{align*}
		
	Using this inequality we get that
	\begin{align*}
		\norm{(\v+y \odot \v+y) - (\v+z \odot \v+z)}_2^2
		&= \sum_{i=1}^n (\v+y_i^2 - \v+z_i^2)^2 \\
		&\leq 4 \sum_{i=1}^n (\v+y_i - \v+z_i)^2 \\
		&= 4 \norm{\v+y - \v+z}_2^2.
	\end{align*}

	We obtain the result by taking square roots.
\end{proof}

\begin{lemma}
\label{lem:difference-rank-one}
	Suppose $\v+y, \v+z \in [-1,1]^n$ then
	$\norm{\v+y \cdot \v+y^\top - \v+z \cdot \v+z^\top}_F
		\leq 2\sqrt{n} \norm{\v+y - \v+z}_2$.
\end{lemma}
\begin{proof}
	We have that
	\begin{align*}
		\norm{\v+y \cdot \v+y^\top - \v+z \cdot \v+z^\top}_F
		&= \norm{(\v+y - \v+z) \cdot \v+y^\top + \v+z \cdot (\v+y^\top - \v+z^\top)}_F \\
		&\leq \norm{(\v+y - \v+z) \cdot \v+y^\top}_F
				+ \norm{\v+z \cdot (\v+y^\top - \v+z^\top)}_F \\
		&\leq (\norm{\v+y}_2 + \norm{\v+z}_2) \cdot
				\norm{\v+y - \v+z}_2.
	\end{align*}
	Since the entries of $\v+y$ and $\v+z$ are in $[-1,1]$, we get
	that $\norm{\v+y}_2 + \norm{\v+z}_2 \leq 2\sqrt{n}$, which implies the lemma.
\end{proof}

The following lemma is a corollary of the Laplacian-solver technique by Koutis, Miller and
Peng~\cite{koutis2014approaching} and allows us to efficiently solve linear
systems approximately.
\begin{lemma}
\label{lem:laplacian-solver}
	Let $\m+D\in\mathbb{R}^{n\times n}$ be a diagonal matrix with entries
	$\m+D_{ii}\geq 1$, let $\m+L$ be the Laplacian of an undirected, connected,
	weighted graph, let $\v+b\in\mathbb{R}^n$ with $\norm{\v+b}_2\leq\poly(n)$ and let
	$\varepsilon > 0$.
	Then there exists a function $\Solve(\m+D + \m+L,\v+b,\varepsilon)$ that
	returns a vector $\widetilde{\v+x}\in\mathbb{R}^n$ such that
	$\norm{\widetilde{\v+x} - (\m+D + \m+L)^{-1} \v+b}_2 \leq \varepsilon$ in time
	$\tO(m\log(1/\varepsilon))$.
\end{lemma}
\begin{proof}
For a symmetric matrix $\m+A\in\mathbb{R}^{n\times n}$, we write $\m+A^+$ to
denote the Moore--Penrose pseudoinverse.  We say that $\m+A$ is
\emph{diagonally dominant} if for all $i$,
$\m+A_{ii}\geq\sum_{j\neq i} \abs{\m+A_{ij}}$.
For a vector $\v+x\in\mathbb{R}^n$ we set
$\norm{\v+x}_{\m+A} = \sqrt{ \v+x^\top \m+A \v+x }$.
We will use the following result by 
Koutis, Miller and Peng~\cite{koutis2014approaching}.

\begin{lemma}[Koutis, Miller, Peng~\cite{koutis2014approaching}]
\label{lem:laplacian-solver-original}
	Let $\m+A\in\mathbb{R}^{n\times n}$ be a symmetric, diagonally dominant
	matrix with $m$ non-zero entries, let $\v+b\in\mathbb{R}^n$ and let
	$\varepsilon > 0$.
	Then there exists a function $\Solve(\m+A,\v+b,\varepsilon)$, which
	returns a vector $\widetilde{\v+x}\in\mathbb{R}^{n\times n}$ such that
	$\norm{\widetilde{\v+x} - \m+A^+ \v+b}_{\m+A} \leq \varepsilon \norm{\m+A^+ \v+b}_{\m+A}$
	in expected time $\tO(m \log(1/\varepsilon))$.
\end{lemma}

We start with an observation about the norm $\norm{\v+v}_{\m+A}$, where we use
that diagonally dominant matrices are positive semidefinite:
\begin{align*}
	\norm{\v+v}_{\m+A}
	&= \sqrt{\v+v^\top \m+A \v+v} \\
	&= \sqrt{\v+v^\top \m+A^{1/2} \m+A^{1/2} \v+v} \\
	&= \norm{\m+A^{1/2} \v+v}_2 \\
	&\geq \sigma_{\min}(\m+A^{1/2}) \norm{\v+v}_2.
\end{align*}
By rearranging terms we get that
$\norm{\v+v}_2 \leq \frac{1}{\sigma_{\min}(\m+A^{1/2})} \norm{\v+v}_{\m+A}$
if $\sigma_{\min}(\m+A^{1/2})>0$.

We use the algorithm from Lemma~\ref{lem:laplacian-solver-original}
with $\m+A = \m+D + \m+L$, $\v+b=\v+b$ and
$\varepsilon'=\frac{\varepsilon}{\norm{\v+b}_2}$ to obtain a vector
$\widetilde{\v+x}$. Note that since we assume that $\norm{\v+b}_2 \leq \poly(n)$, we
get that this algorithm runs in time
$\tO(m \log(1/\varepsilon')) = \tO(m \log(1/\varepsilon))$, since this only adds
additional $\bigO(\lg n)$-term, which is hidden in the $\tO(\cdot)$-notation.

Now we observe that by Lemma~\ref{lem:eigenvalues-M}, $\m+D + \m+L$ is positive
definite with all eigenvalues at least~$1$. Thus
$(\m+D + \m+L)^+ = (\m+D + \m+L)^{-1}$. Furthermore, the lemma implies that
$\sigma_{\min}((\m+D + \m+L)^{1/2}) = \lambda_{\min}(\m+D + \m+L) \geq 1$
and that all eigenvalues of $(\m+D + \m+L)^{-1}$ are in the interval $(0,1]$.

Now we use our previous result about $\norm{\v+v}_{\m+A}$ to get that
\begin{align*}
	&\norm{\widetilde{\v+x} - (\m+D + \m+L)^{-1} \v+b}_2 \\
	&\leq \frac{1}{\sigma_{\min}( (\m+D + \m+L)^{1/2})}
			\cdot \norm{\widetilde{\v+x} - (\m+D + \m+L)^{-1} \v+b}_{\m+D + \m+L} \\
	&\leq \frac{1}{\sigma_{\min}( (\m+D + \m+L)^{1/2})}
			\cdot \varepsilon' \norm{(\m+D + \m+L)^{-1} \v+b}_{\m+D + \m+L} \\
	&= \frac{1}{\sigma_{\min}( (\m+D + \m+L)^{1/2})}
			\cdot \varepsilon' \sqrt{\v+b^\top (\m+D + \m+L)^{-1} \v+b} \\
	&\leq \frac{1}{\sigma_{\min}( (\m+D + \m+L)^{1/2})}
			\cdot \varepsilon' \lambda_{\max}( (\m+D + \m+L)^{-1} ) \norm{\v+b}_2 \\
	&\leq \varepsilon' \cdot \norm{\v+b}_2 \\
	&= \varepsilon.
\end{align*}
\end{proof}

\subsection{Proof of Lemma~\ref{lem:weight-edges}}
First, recall that since $\m+X$ and $\m+Y$ are row-stochastic matrices with
non-negative entries. Thus, we get that
$\norm{\m+X\m+Y + \m+Y^T \m+X^T}_{1,1} =  \norm{\m+X\m+Y}_{1,1} + \norm{\m+Y^T \m+X^T}_{1,1}$.

Next, we have that
\begin{align*}
	\norm{\m+X \m+Y \v+1}_{1,1}
	&= \v+1^\top \m+X \m+Y \v+1 \\
	&= \v+1^\top \m+X \v+1 \\
	&= \v+1^\top \v+1 \\
	&= n.
\end{align*}

Similarly, since $\m+Y^\top$ and $\m+X^\top$ are column-stochastic, we get that
\begin{align*}
	\norm{\m+Y^\top \m+X^\top}_{1,1}
	&= \v+1^\top \m+Y^\top \m+X^\top \v+1 \\
	&= \v+1^\top \m+X^\top \v+1 \\
	&= \v+1^\top \v+1 \\
	&= n.
\end{align*}

Summing over these two quantities proves the lemma.

\subsection{Proof of Proposition~\ref{prop:fast-opinions}}
\label{sec:proof-prop-fast-opinions}

Since this proof is rather involved, we first present a short proof sketch before giving the full details.

\subsubsection{Proof sketch}
	Algorithm~\ref{alg:opinions} is based on the observation that
	using the Woodbury matrix identity with
	$\m+M = \m+I + \m+L + \diag(\AXC\v+1)$,
	and $\m+U$ and $\m+V$ as before, we get that
	\begin{align*}
		\zXC &= \m+M^{-1} \v+s
			 	+ \frac{CW}{2n}
					\m+M^{-1}
					\m+U
			 		\left(
						\m+I
						- \frac{CW}{2n} \m+V \m+M^{-1} \m+U
					\right)^{-1}
					\m+V
					\m+M^{-1} \v+s.
	\end{align*}
	Now Algorithm~\ref{alg:opinions} basically computes this quantity from
	right to left.
    Our main insight here is that we can compute the quantities $\m+M^{-1}\v+s$ and $\m+M^{-1}\m+U$ using the Laplacian solver from Lemma~\ref{lem:laplacian-solver}. Here, we approximate $\m+M^{-1}\m+U$ column-by-column using the call $\Solve(\m+M,\v+w_j,\varepsilon_{\m+R})$, where $\v+w_j$ is the $j$-th column of $\m+U$ and $\varepsilon_{\m+R}$ is a suitable error parameter. The remaining matrix multiplications are efficient since $\m+U$ has only $2k$~columns and since $\m+V$ has only $2k$~rows.
 
    To obtain our guarantees for the approximation error, we have to perform an intricate error analysis to ensure that errors do not compound too much. This is a challenge since we solve $\m+I - \frac{CW}{2n} \m+V\m+M^{-1}\m+U$ only approximately but then we have to compute an inverse of this approximate quantity. 
In the proposition we use the assumptions that $\m+V \m+M^{-1} \m+U$ exists and
that $\norm{\m+V \m+M^{-1} \m+U}_2 \leq 0.99 \frac{2n}{CW}$, to ensure that this can be done without obtaining too much error.
\sninline{While in general
this may not be true, we show experimentally that this condition is satisfied in
practice.} \sninline{We have to run these experiments.} %
In the proof we will also show that these assumptions imply that the inverse
$\m+S^{-1}$ used in the algorithm exists.

\subsubsection{Formal proof}
    We state the pseudocode of the algorithm in Algorithm~\ref{alg:opinions}.

	Set $\m+U = \begin{pmatrix} \m+X & \m+Y^\top \end{pmatrix}$ and
	$\m+V = \begin{pmatrix} \m+Y \\ \m+X^\top \end{pmatrix}$.
	Now observe that $\frac{CW}{2n} \m+U \m+V = \AXC$.

	Recall that the Woodbury matrix identity states that
	\begin{align*}
		(\m+M + \m+U \m+C \m+V)^{-1}
		= \m+M^{-1}
			- \m+M^{-1} \m+U (\m+C^{-1} + \m+V \m+M^{-1} \m+U)^{-1} \m+V \m+M^{-1}.
	\end{align*}

	Using the Woodbury matrix identity with
	$\m+M = \m+I + \m+L + \diag(\AXC\v+1)$,
	$\m+C = -\frac{CW}{2n} \m+I$,
	and $\m+U$ and $\m+V$ as before,
	we obtain that
	\begin{align*}
		\zXC &= (\m+I + \m+L + \LXC)^{-1} \v+s \\
			 &= (\m+I + \m+L + \diag(\AXC\v+1) - \AXC)^{-1} \v+s \\
			 &= \left(\m+M -
					\m+U \cdot
					 \frac{CW}{2n} \m+I
					\cdot \m+V
				\right)^{-1} \v+s \\
			 &= \m+M^{-1} \v+s
			 	- \m+M^{-1}
					\m+U
			 		\left(
						- \frac{2n}{CW} \m+I
						+ \m+V \m+M^{-1} \m+U
					\right)^{-1}
					\m+V
					\m+M^{-1} \v+s \\
			 &= \m+M^{-1} \v+s
			 	+ \frac{CW}{2n}
					\m+M^{-1}
					\m+U
			 		\left(
						\m+I
						- \frac{CW}{2n} \m+V \m+M^{-1} \m+U
					\right)^{-1}
					\m+V
					\m+M^{-1} \v+s,
	\end{align*}
	where we use that $\m+C^{-1} = - \frac{2n}{CW} \m+I$.

	We start by giving some intuition why Algorithm~\ref{alg:opinions} computes
	an approximate version of $\zXC$. We present the running time analysis and
	the formal error analysis below.
	
	First, observe that $\m+M$ is a symmetric, diagonally dominant matrix with
	$\bigO(m)$ entries and diagonal entries $\m+M_{ii} \geq 1$. Thus, whenever we
	wish to solve $\m+M^{-1} \v+v$, we can use the Laplacian solver from
	Lemma~\ref{lem:laplacian-solver}. 

	Next, note that $\v+z_1$ is an approximation of $\m+M^{-1} \v+s$.  Then
	$\v+y_1$ becomes an approximation of $\m+V \m+M^{-1} \v+s$.  Now $\m+R$ is
	an efficient approximation of $\m+M^{-1} \m+U$ by using
	Lemma~\ref{lem:laplacian-solver}.  Hence, $\m+S$ is an approximation of
	$\m+I - \frac{CW}{2n} \m+V \m+M^{-1} \m+U$ and $\m+T$ approximates
	$(\m+I - \frac{CW}{2n} \m+V \m+M^{-1} \m+U)^{-1}$; we note that in the proof
	of Claim~\ref{claim:norm_E_T} we also point out why this inverse exists
	under the assumptions of the lemma.
	Continuing this approach, we obtain that $\v+y_2$ approximates
	$\left( \m+I
			- \frac{CW}{2n}\m+V \m+M^{-1} \m+U \right)^{-1}
					\m+V \m+M^{-1} \v+s$,
	$\v+y_3$ approximates
	$\m+U\left( \m+I
			- \frac{CW}{2n} \m+V \m+M^{-1} \m+U \right)^{-1}
					\m+V \m+M^{-1} \v+s$,
	$\v+z_2$ approximates
	$\m+M^{-1} \m+U^\top \left( \m+I
			- \frac{CW}{2n} \m+V \m+M^{-1} \m+U \right)^{-1}
					\m+V \m+M^{-1} \v+s$,
	and hence $\tzXC = \v+z_1 + \frac{CW}{2n} \v+z_2$ approximates $\zXC$.
	
	\textbf{Running time analysis.} Now let us consider the running time.
	Let us start by observing that to obtain our running time bounds we can only
	apply the Laplacian solvers on vectors $\v+b$ with
	$\norm{\v+b}_2\leq\poly(n)$. Note that for $\Solve(\m+M,\v+s,\varepsilon_{\v+z_1})$
	and $\Solve(\m+M,\v+w_i,\varepsilon_{\m+R})$, $i=1,\dots,2k$, this is clear since
	the entries in $\v+s$ and $\v+w_i$ are bounded by $[-1,1]$ and $[0,1]$,
	respectively. For $\Solve(\m+M,\v+y_3,\varepsilon_{\v+z_2})$, we note that one
	can show that $\norm{\v+y_3} \leq \poly(n)$ similar to the proofs of
	Claims~\ref{claim:error-third-term} and~\ref{claim:error-second-term}
	below.

	First, note that $\diag(\AXC \v+1)$ can be computed in time $\bigO(nk)$.
	Given $\diag(\AXC \v+1)$, we can compute $\m+M$ in time $\bigO(m)$ because it is
	the sum of $\m+L$ and two diagonal matrices.
	Next, $\v+z_1$ can be computed in time
	$\tO(m \lg(1/\varepsilon_{\v+z_1})) = \tO(m \lg(W/\varepsilon))$ using the guarantee
	from Lemma~\ref{lem:laplacian-solver}.
	In the next step, $\v+y_1 = \m+V \v+z_1 \in \mathbb{R}^{2k}$ can be computed
	in time $\bigO(nk)$ since $\m+V$ is a $(2k)\times n$ matrix.
	Then by Lemma~\ref{lem:laplacian-solver}$, \m+R$ can be computed in time
	$\tO(mk \lg(1/\varepsilon_{\m+R})) = \tO(mk \lg(W/\varepsilon))$ since we
	need $2k$~calls $\Solve(\m+M,\v+w_j,\varepsilon_{\m+R})$.
	As $\m+I$ has only $\bigO(k)$ non-zero entries and
	$\m+V$ and $\m+R$ are matrices of sizes $(2k)\times n$ and $n \times (2k)$,
	respectively, we can compute $\m+S\in\mathbb{R}^{(2k)\times(2k)}$ in time
	$\bigO(nk^2)$.  Due to size of $\m+S$, we can compute $\m+T = \m+S^{-1}$ in time
	$\bigO(k^3)$ using Gaussian elimination and thus we obtain
	$\v+y_2\in\mathbb{R}^{2k}$ in time $\bigO(k^2)$. Then we can compute $\v+y_3$ in
	time $\bigO(nk)$ since $\m+U\in\mathbb{R}^{n\times(2k)}$ and we need time
	$\tO(m \lg(1/\varepsilon_{\v+z_2})) = \tO(m \lg(W/\varepsilon))$ for the call to
	$\Solve(\m+M,\v+y_3,\varepsilon_{\v+z_2})$ by
	Lemma~\ref{lem:laplacian-solver}.
	Summing over all terms above, we obtain our desired running time bound.

	\textbf{Analysis of approximation error.}
	Let $\v+z^*_2$ denote the error-free version of $\v+z_2$, i.e.,
	\begin{align*}
		\v+z_2^* = \m+M^{-1} \m+U
			 		\left(
						\m+I - \frac{CW}{2n} \m+V \m+M^{-1} \m+U
					\right)^{-1}
					\m+V \m+M^{-1} \v+s.
	\end{align*}
	Observe that the difference between $\v+z_2$ and $\v+z_2^*$ will be
	our main source of error when approximating $\zXC$ with $\tzXC$. Hence, next
	we write down $\v+z_2$ to understand where we get inaccuracies compared to
	$\v+z_2^*$.

	First, we set $\v+z_1^*=\m+M^{-1} s$, i.e., it is the exact solution of
	$\m+M \v+z_1^* = \v+s$. Hence, we get that
	$\v+z_1 = \v+z_1^* + \v+e_{\v+z_1}$ where
	$\v+e_{\v+z_1}$ is the error vector introduced by the Laplacian solver.
	Thus, we get that
	$\v+y_1 = \m+V \v+z_1 = \m+V \m+M^{-1} \v+s + \m+V \v+e_{\v+z_1}$.
	Next, let $\m+E_{\m+T}$ be the error matrix such that
	$\m+T = \left(\m+I - \frac{CW}{2n} \m+V \m+M^{-1} \m+U \right)^{-1}
			+ \m+E_{\m+T}$.
	Then we have that
	\begin{align*}
		\v+y_2 &= \m+T \v+y_1
		= \left(\m+I - \frac{CW}{2n} \m+V \m+M^{-1} \m+U \right)^{-1}
				\m+V \m+M^{-1} \v+s
			+ \left(\left(\m+I - \frac{CW}{2n} \m+V \m+M^{-1} \m+U \right)^{-1} +
					\m+E_{\m+T}\right) \m+V \v+e_{\v+z_1}
			+ \m+E_{\m+T} \m+V \m+M^{-1} \v+s.
	\end{align*}
	Next,
	\begin{align*}
		\v+y_3 &= \m+U \v+y_2
		= \m+U \left(\m+I - \frac{CW}{2n} \m+V \m+M^{-1} \m+U \right)^{-1}
				\m+V \m+M^{-1} \v+s
			+ \m+U \left(\left(\m+I - \frac{CW}{2n} \m+V \m+M^{-1} \m+U \right)^{-1} +
					\m+E_{\m+T}\right) \m+V \v+e_{\v+z_1}
			+ \m+U \m+E_{\m+T} \m+V \m+M^{-1} \v+s.
	\end{align*}
	Finally, let $\v+e_{\v+z_2}$ be such that
	$\v+z_2 = \m+M^{-1} \v+y_3 + \v+e_{\v+z_2}$.
	Then we get that
	\begin{align*}
		\v+z_2
		&= \m+M^{-1} \v+y_3 + \v+e_{\v+z_2} \\
		&= \m+M^{-1} \m+U \left(\m+I - \frac{CW}{2n} \m+V \m+M^{-1} \m+U \right)^{-1}
				\m+V \m+M^{-1} \v+s
			+ \m+M^{-1} \m+U \left(\left(\m+I - \frac{CW}{2n} \m+V \m+M^{-1} \m+U \right)^{-1} +
					\m+E_{\m+T}\right) \m+V \v+e_{\v+z_1} \\
		&\quad + \m+M^{-1} \m+U \m+E_{\m+T} \m+V \m+M^{-1} \v+s
			+ \v+e_{\v+z_2} \\
		&= \v+z_2^*
			+ \m+M^{-1} \m+U \left(\left(\m+I - \frac{CW}{2n} \m+V \m+M^{-1} \m+U \right)^{-1} + \m+E_{\m+T}\right) \m+V \v+e_{\v+z_1}
			+ \m+M^{-1} \m+U \m+E_{\m+T} \m+V \m+M^{-1} \v+s
			+ \v+e_{\v+z_2}.
	\end{align*}

	The above implies that we return an approximation $\tzXC$ such that
	\begin{equation}
	\label{eq:error-terms}
	\begin{aligned}
		&\norm{\tzXC - \zXC}_2 \\
		&= \norm{\v+z_1 + \frac{CW}{2n} \v+z_2 - \v+z_1^* - \frac{CW}{2n} \v+z_2^*}_2 \\
		&= \norm{\v+e_{\v+z_1}
			+ \frac{CW}{2n} \left(
					\m+M^{-1} \m+U \left(\left(\m+I - \frac{CW}{2n} \m+V \m+M^{-1} \m+U \right)^{-1} + \m+E_{\m+T}\right) \m+V \v+e_{\v+z_1}
					+ \m+M^{-1} \m+U \m+E_{\m+T} \m+V \m+M^{-1} \v+s
					+ \v+e_{\v+z_2} }_2
				\right) \\
		&\leq \norm{\v+e_{\v+z_1}}_2
			+ \frac{CW}{2n}\left(
					\norm{ \m+M^{-1} \m+U \left(\left(\m+I - \frac{CW}{2n} \m+V \m+M^{-1} \m+U \right)^{-1} + \m+E_{\m+T}\right) \m+V \v+e_{\v+z_1}}_2
					+ \norm{\m+M^{-1} \m+U \m+E_{\m+T} \m+V \m+M^{-1} \v+s}_2
					+ \norm{\v+e_{\v+z_2} }_2
				\right).
	\end{aligned}
	\end{equation}
	Thus, for the remainder of the proof, we bound these four error terms.

	\begin{claim}
	\label{claim:norm-inverse}
		$\norm{\left( \m+I - \frac{CW}{2n} \m+V \m+M^{-1} \m+U \right)^{-1}}_2
			\leq 100$.
	\end{claim}
	\begin{proof}
		Using our assumption that
		$\norm{\m+V \m+M^{-1} \m+U}_2 \leq 0.99 \frac{2n}{CW}$,
	   	the Neumann series, triangle inequality and the geometric
		series, we obtain that
		\begin{align*}
			\norm{\left( \m+I - \frac{CW}{2n} \m+V \m+M^{-1} \m+U \right)^{-1}}_2
			&= \norm{
					\sum_{i=0}^\infty \left(\frac{CW}{2n} \m+V \m+M^{-1} \m+U\right)^i
				}_2 \\
			&\leq
				\sum_{i=0}^\infty 
					\left(\frac{CW}{2n}\right)^{i}
					\norm{ \m+V \m+M^{-1} \m+U }_2^i \\
			&\leq \sum_{i=0}^\infty 0.99^i \\
			&= \frac{1}{1-0.99} \\
			&= 100.
			\qedhere
		\end{align*}
	\end{proof}

	\begin{claim}
	\label{claim:norm_E_T}
		$\norm{\m+E_{\m+T}}_2 \leq
			\min\left\{ 100,
				\frac{2n}{CW} \cdot
				\frac{\varepsilon/4}{\norm{\m+U}_2 \cdot \norm{\m+V}_2
					\cdot \norm{\v+s}_2} \right\}$.
	\end{claim}
	\begin{proof}
		Let $\m+E_{\m+R}$ be the error matrix such that
		$\m+R = \m+M^{-1} \m+U + \m+E_{\m+R}$
		and recall that
		$\m+T = \left(\m+I - \frac{CW}{2n} \m+V \m+M^{-1} \m+U \right)^{-1}
			+ \m+E_{\m+T}$.
		Observe that $\m+S = \m+I - \frac{CW}{2n}\m+V \m+R
			= \m+I - \frac{CW}{2n}\m+V (\m+M^{-1} \m+U + \m+E_{\m+R})$.
		Then we have that
		\begin{align*}
			\norm{\m+E_{\m+T}}_2
			&= \norm{
					\m+T - \left(
						\m+I - \frac{CW}{2n} \m+V \m+M^{-1} \m+U
				\right)^{-1}
				}_2 \\
			&= \norm{
					\m+S^{-1} - \left(
						\m+I + \frac{CW}{2n} \m+V \m+M^{-1} \m+U
				\right)^{-1}
				}_2 \\
			&\leq
				\norm{\m+S^{-1}}_2
				\cdot \norm{
					\left(
						\m+I - \frac{CW}{2n} \m+V \m+M^{-1} \m+U
					\right)^{-1}
				}_2
				\cdot \norm{
					\m+S - \left(
						\m+I
						- \frac{CW}{2n} \m+V \m+M^{-1} \m+U
					\right)
				}_2 \\
			&=
				\norm{\m+S^{-1}}_2
				\cdot \norm{
					\left(
						\m+I - \frac{CW}{2n} \m+V \m+M^{-1} \m+U
					\right)^{-1}
				}_2
				\cdot \norm{ \frac{CW}{2n}\m+V \m+E_{\m+R}}_2 \\
			&\leq \frac{CW}{2n}
				\norm{\m+S^{-1}}_2
				\cdot \norm{
					\left(
						\m+I - \frac{CW}{2n} \m+V \m+M^{-1} \m+U
					\right)^{-1}
				}_2
				\cdot \norm{\m+V}_2
			   	\cdot \norm{\m+E_{\m+R}}_2.
		\end{align*}

		Next, let $\v+e_{\v+w_j}$ denote the error in the $j$-th column of
		$\m+R$, i.e., $\v+e_{\v+w_j}$ is the $j$-th column of $\m+E_{\m+R}$.
		Observe that by Lemma~\ref{lem:laplacian-solver} and our choice of
		$\varepsilon_{\m+R} =\frac{1}{2k}
			\min\left\{0.009 \frac{2n}{CW \cdot \norm{\m+V}_F},
			\frac{2n}{10^{5}\cdot CW \cdot \norm{\m+V}_2}
			\cdot \min\left\{ 100,
					\frac{2n}{CW} \cdot \frac{\varepsilon/4}
					{\norm{\m+U}_2 \cdot \norm{\m+V}_2
						\cdot \norm{\v+s}_2} \right\}
			\right\}$,
		we get that $\norm{\v+e_{\v+w_j}}_2 \leq \varepsilon_{\m+R}$ for all $j$.
		Then we get that
		\begin{align*}
			\norm{\m+E_{\m+R}}_2
			&\leq \norm{\m+E_{\m+R}}_F \\
			&= \sqrt{ \sum_j \sum_i (\m+E_{\m+R})_{ij}^2 } \\
			&= \sqrt{ \sum_j \norm{\v+e_{\v+w_j}}_2^2 } \\
			&\leq \sum_j \norm{\v+e_{\v+w_j}}_2 \\
			&\leq 2 k \cdot \varepsilon_{\m+R} \\
			&\leq 
				\min\left\{0.009 \frac{2n}{CW \cdot \norm{\m+V}_F},
				\frac{2n}{10^{5}\cdot CW \cdot \norm{\m+V}_2}
				\cdot \min\left\{ 100,
						\frac{2n}{CW} \cdot \frac{\varepsilon/4}
						{\norm{\m+U}_2 \cdot \norm{\m+V}_2
							\cdot \norm{\v+s}_2} \right\}
				\right\},
		\end{align*}
		where the fourth step holds since $\norm{\v+v}_2\leq\norm{\v+v}_1$
		for any vector $\v+v\in\mathbb{R}^n$.

		Thus, similar to the proof of Claim~\ref{claim:norm-inverse} we obtain
		that
		\begin{align*}
			\norm{\m+S^{-1}}_2
			&= \norm{\left( \m+I - \frac{CW}{2n} \m+V (\m+M^{-1} \m+U + \m+E_{\m+R})
					\right)^{-1}}_2 \\
			&\leq
				\sum_{i=0}^\infty 
					\left(\frac{CW}{2n}\right)^i
					\norm{ \m+V \m+M^{-1} \m+U
							+ \m+V \m+E_{\m+R}
					}_2^i \\
			&\leq 
				\sum_{i=0}^\infty 
					\left(\frac{CW}{2n}\right)^i
					\left(
						\norm{ \m+V \m+M^{-1} \m+U }_2
						+ \norm{\m+V}_F \norm{\m+E_{\m+R}}_F
					\right)^i \\
			&\leq \sum_{i=0}^\infty 
					\left(\frac{CW}{2n}\right)^i
					\left(
						0.999 \frac{2n}{CW}
					\right)^i \\
			&= \sum_{i=0}^\infty 0.999^i \\
			&= 10^3.
		\end{align*}
		Note that this implies that $\m+S^{-1}$ exists.

		Now combining the above results with Claim~\ref{claim:norm-inverse}, we
		get that
		\begin{align*}
			\norm{\m+E_{\m+T}}_2
			&\leq \frac{CW}{2n}
				\norm{\m+S^{-1}}_2
				\cdot \norm{
					\left(
						\m+I - \frac{CW}{2n} \m+V \m+M^{-1} \m+U
					\right)^{-1}
				}_2
				\cdot \norm{\m+V}_2
			   	\cdot \norm{\m+E_{\m+R}}_2 \\
			&\leq \frac{CW}{2n}
				\cdot 10^3
				\cdot 10^2
				\cdot \norm{\m+V}_2
				\cdot \frac{2n}{10^{5}\cdot CW \cdot \norm{\m+V}_2}
					\cdot \min\left\{ 100,
					\frac{2n}{CW} \cdot \frac{\varepsilon/4}
						{\norm{\m+U}_2 \cdot \norm{\m+V}_2
							\cdot \norm{\v+s}_2} \right\} \\
			&\leq \min\left\{ 100,
					\frac{2n}{CW} \cdot \frac{\varepsilon/4}
						{\norm{\m+U}_2 \cdot \norm{\m+V}_2
							\cdot \norm{\v+s}_2} \right\}.
			\qedhere
		\end{align*}
	\end{proof}

	\begin{claim}
	\label{claim:error-third-term}
		$\frac{CW}{2n} \norm{ \m+M^{-1} \m+U \m+E_{\m+T} \m+V \m+M^{-1} \v+s}_2
			\leq \varepsilon/4$.
	\end{claim}
	\begin{proof}
		First, observe that $\norm{\m+M^{-1}}_2 \leq 1$ by
		Lemma~\ref{lem:eigenvalues-M}.
		Now using Claim~\ref{claim:norm_E_T} we get that
		\begin{align*}
			\frac{CW}{2n} \norm{ \m+M^{-1} \m+U \m+E_{\m+T} \m+V \m+M^{-1} \v+s}_2
			&\leq \frac{CW}{2n} \cdot \norm{ \m+M^{-1} }_2^2 
				\cdot \norm{\m+U}_2 \cdot \norm{\m+V}_2 \cdot \norm{\v+s}_2
				\cdot \norm{\m+E_{\m+T}}_2 \\
			&\leq \frac{CW}{2n}
				\cdot \norm{\m+U}_2 \cdot \norm{\m+V}_2 \cdot \norm{\v+s}_2
				\cdot \norm{\m+E_{\m+T}}_2 \\
			&\leq \varepsilon/4.
			\qedhere
		\end{align*}
	\end{proof}

	\begin{claim}
	\label{claim:error-second-term}
		$\frac{CW}{2n} \norm{ \m+M^{-1} \m+U \left(\left(\m+I - \frac{CW}{2n} \m+V \m+M^{-1} \m+U \right)^{-1} + \m+E_{\m+T}\right) \m+V \v+e_{\v+z_1}}_2
			\leq \varepsilon/4$.
	\end{claim}
	\begin{proof}
		First, observe that by Claims~\ref{claim:norm-inverse}
		and~\ref{claim:norm_E_T} and the triangle
		inequality,
		\begin{align*}
			\norm{ \left(\m+I - \frac{CW}{2n} \m+V \m+M^{-1} \m+U \right)^{-1} + \m+E_{\m+T} }_2
			\leq \norm{ \left(\m+I - \frac{CW}{2n} \m+V \m+M^{-1} \m+U
					\right)^{-1}}_2
		   		+ \norm{\m+E_{\m+T}}_2
			\leq 200.
		\end{align*}

		Using the inequality from above and Lemma~\ref{lem:eigenvalues-M} and
		our choice of
		$\varepsilon_{\v+z_1} = \frac{\varepsilon}{4}
				\cdot \min\left\{
					1,
					\frac{2n}{200 \cdot CW \cdot \norm{\m+U}_2 \cdot \norm{\m+V}_2}
				\right\}$,
		we obtain that
		\begin{align*}
			&\frac{CW}{2n} \norm{ \m+M^{-1} \m+U \left(\left(\m+I - \frac{CW}{2n} \m+V \m+M^{-1} \m+U \right)^{-1} + \m+E_{\m+T}\right) \m+V \v+e_{\v+z_1}}_2 \\
			&\leq \frac{CW}{2n}
				\cdot \norm{\m+M^{-1}}_2
				\cdot \norm{\m+U}_2
				\cdot \norm{ \left(\m+I - \frac{CW}{2n} \m+V \m+M^{-1} \m+U \right)^{-1} + \m+E_{\m+T}}_2
				\cdot \norm{\m+V}_2
				\cdot \norm{\v+e_{\v+z_1}}_2 \\
			&\leq \frac{CW}{2n}
				\cdot \norm{\m+U}_2
				\cdot 200
				\cdot \norm{\m+V}_2
				\cdot \norm{\v+e_{\v+z_1}}_2 \\
			&\leq \varepsilon/4.
			\qedhere
		\end{align*}
	\end{proof}

	Now continuing Eq.~\eqref{eq:error-terms}, using
	Claims~\ref{claim:error-second-term} and~\ref{claim:error-third-term}, and
	our choices $\varepsilon_{\v+z_1} = \frac{\varepsilon}{4}
				\cdot \min\left\{
					1,
					\frac{2n}{200 \cdot CW \cdot \norm{\m+U}_2 \cdot \norm{\m+V}_2}
				\right\}$ and
	$\varepsilon_{\v+z_2} = \frac{2n}{CW} \cdot \frac{\varepsilon}{4}$, we get that
	\begin{align*}
		&\norm{\tzXC - \zXC}_2 \\
		&\leq \norm{\v+e_{\v+z_1}}_2
			+ \frac{CW}{2n}\left(
					\norm{ \m+M^{-1} \m+U \left(\left(\m+I - \frac{CW}{2n} \m+V \m+M^{-1} \m+U \right)^{-1} + \m+E_{\m+T}\right) \m+V \v+e_{\v+z_1}}_2
					+ \norm{\m+M^{-1} \m+U \m+E_{\m+T} \m+V \m+M^{-1} \v+s}_2
					+ \norm{\v+e_{\v+z_2} }_2
				\right) \\
		&\leq \varepsilon/4 + \varepsilon/4 + \varepsilon/4 + \varepsilon/4 \\
		&= \varepsilon.
	\end{align*}

\subsection{Proof of Corollary~\ref{cor:fast-objective-function}}
	We compute $\tzXC$ using Proposition~\ref{prop:fast-opinions} with
	$\varepsilon' = \frac{\varepsilon}{\sqrt{n}}$. Then we set
	$\widetilde{f}=\v+s^\top \tzXC$. Using the Cauchy--Schwarz inequality, we
	obtain that
	\begin{align*}
		\abs{\widetilde{f} - f(\m+X)}
		&= \abs{ \v+s^\top \tzXC - \v+s^\top \zXC} \\
		&= \abs{ \v+s^\top (\tzXC - \zXC)} \\
		&\leq \norm{\v+s}_2 \cdot \norm{\tzXC - \zXC}_2 \\
		&\leq \sqrt{n} \cdot \frac{\varepsilon}{\sqrt{n}} \\
		&= \varepsilon,
	\end{align*}
	where we use that all entries in $\v+s$ are in the interval $[-1,1]$.

\subsection{Proof of Proposition~\ref{prop:gradient}}

\subsubsection{Derivation of the gradient}

\paragraph{Matrixcalculus.}
We use \url{matrixcalculus.org}~\cite{laue2018computing,laue2020simple} to
obtain the gradient. We set $F = \m+I + \m+L$, $c = \frac{CW}{2n}$, $v=\v+1$ and
use the input \emph{s'*inv(F + diag (c*(X*Y + Y'*X')*v) - c*(X*Y + Y'*X')) * s}.

We obtain:
\begin{align*}
	\frac{\partial}{\partial X} \left( s^\top \cdot \mathrm{inv}(F+\mathrm{diag}(c\cdot (X\cdot Y+Y^\top \cdot X^\top )\cdot v)-c\cdot (X\cdot Y+Y^\top \cdot X^\top ))\cdot s \right) = \\\quad\quad -(c\cdot t_{5}\odot t_{7}\cdot (Y\cdot v)^\top +c\cdot v\cdot ((t_{8}\odot t_{6})\cdot Y^\top )-(c\cdot t_{5}\cdot (t_{6}\cdot Y^\top )+c\cdot t_{7}\cdot (t_{8}\cdot Y^\top )))
\end{align*}

where
\begin{itemize}
	\item $T_0 = X\cdot Y$
	\item $T_1 = T_{0}^\top +T_{0}$
	\item $T_2 = T_{0}^\top +T_{0}$
	\item $T_3 = \mathrm{inv}(F+c\cdot \mathrm{diag}(T_{2}\cdot v)-c\cdot T_{2})$
	\item $T_4 = \mathrm{inv}(F^\top +c\cdot \mathrm{diag}(v^\top \cdot T_{1})-c\cdot T_{1})$
	\item $t_5 = T_{4}\cdot s$
	\item $t_6 = s^\top \cdot T_{4}$
	\item $t_7 = T_{3}\cdot s$
	\item $t_8 = s^\top \cdot T_{3}$
\end{itemize}
and
\begin{itemize}
  \item $F$ is a symmetric matrix
  \item $X$ is a matrix
  \item $Y$ is a matrix
  \item $c$ is a scalar
  \item $s$ is a vector
  \item $v$ is a vector
\end{itemize}

\paragraph{Simplification 1.}
We note above that $T_1=T_2$, hence we replace ever occurence of $T_2$ with
$T_1$. Furthermore, in our setting we have that $F$ and $T_1$ are symmetric
$n\times n$ matrices, which implies that $T_3 = T_4$; hence, we replace every
occurence of $T_4$ with $T_3$.

Then we get:
\begin{align*}
	\frac{\partial}{\partial X} \left( \v+s^\top \cdot \mathrm{inv}(F+\mathrm{diag}(c\cdot (X\cdot Y+Y^\top \cdot X^\top)\cdot v)-c\cdot (X\cdot Y+Y^\top \cdot X^\top ))\cdot \v+s \right) = \\\quad\quad -(c\cdot t_{5}\odot t_{7}\cdot (Y\cdot v)^\top +c\cdot v\cdot ((t_{8}\odot t_{6})\cdot Y^\top )-(c\cdot t_{5}\cdot (t_{6}\cdot Y^\top )+c\cdot t_{7}\cdot (t_{8}\cdot Y^\top )))
\end{align*}

where
\begin{itemize}
	\item $T_0 = X\cdot Y$
	\item $T_1 = T_{0}^\top +T_{0}$
	\item $T_3 = \mathrm{inv}(F+c\cdot \mathrm{diag}(T_{1}\cdot v)-c\cdot T_{1})$
	\item $t_5 = T_{3}\cdot s$
	\item $t_6 = s^\top \cdot T_{3}$
	\item $t_7 = T_{3}\cdot s$
	\item $t_8 = s^\top \cdot T_{3}$
\end{itemize}

\paragraph{Simplification 2.}
Next, observe that $t_5 = t_7$ and $t_6 = t_8$. Hence, we only use $t_5$ and
$t_6$. 

Then we get:
\begin{align*}
	\frac{\partial}{\partial X} \left( s^\top \cdot \mathrm{inv}(F+\mathrm{diag}(c\cdot (X\cdot Y+Y^\top \cdot X^\top )\cdot v)-c\cdot (X\cdot Y+Y^\top \cdot X^\top ))\cdot s \right) = \\\quad\quad -(c\cdot t_{5}\odot t_{5}\cdot (Y\cdot v)^\top +c\cdot v\cdot ((t_{6}\odot t_{6})\cdot Y^\top )-(c\cdot t_{5}\cdot (t_{6}\cdot Y^\top )+c\cdot t_{5}\cdot (t_{6}\cdot Y^\top )))
\end{align*}

where
\begin{itemize}
	\item $T_0 = X\cdot Y$
	\item $T_1 = T_{0}^\top +T_{0}$
	\item $T_3 = \mathrm{inv}(F+c\cdot \mathrm{diag}(T_{1}\cdot v)-c\cdot T_{1})$
	\item $t_5 = T_{3}\cdot s$
	\item $t_6 = s^\top \cdot T_{3}$
\end{itemize}

\paragraph{Simplification 3.}
Next, we remove the leading minus sign by multiplying it inside. We also observe
that $t_5 = t_6^\top$ since $T_3$ is symmetric. Hence, we only use $t_5$ and
$t_5^\top$. 

Then we get:
\begin{align*}
	\frac{\partial}{\partial X} \left( s^\top \cdot \mathrm{inv}(F+\mathrm{diag}(c\cdot (X\cdot Y+Y^\top \cdot X^\top )\cdot v)-c\cdot (X\cdot Y+Y^\top \cdot X^\top ))\cdot s \right) =
	\\\quad\quad
	-c\cdot t_{5}\odot t_{5}\cdot (Y\cdot v)^\top
	-c\cdot v\cdot ((t_{5}^\top \odot t_{5}^\top)\cdot Y^\top )
	+c\cdot t_{5}\cdot (t_{5}^\top\cdot Y^\top )
	+c\cdot t_{5}\cdot (t_{5}^\top \cdot Y^\top )
\end{align*}

where
\begin{itemize}
	\item $T_0 = X\cdot Y$
	\item $T_1 = T_{0}^\top +T_{0}$
	\item $T_3 = \mathrm{inv}(F+c\cdot \mathrm{diag}(T_{1}\cdot v)-c\cdot T_{1})$
	\item $t_5 = T_{3}\cdot s$
\end{itemize}

\paragraph{Simplification 4.}
Next, we first observe that the final two terms above are the same. Also, recall
that we set $v = \v+1$ and that $Y$ is a row-stochastic $k\times n$ matrix.
Hence, we get that $Yv = Y \v+1 = \v+1_k$, where $\v+1_k\in\mathbb{R}^k$ is a
row-vector in which all entries are set to $1$.
We also substitute our notation from above and observe that
$T_3=\mathrm{inv}(F+c\cdot \mathrm{diag}(T_{1}\cdot v)-c\cdot T_{1})
	= (\m+I + \m+L + \LXC)^{-1}$.

Then we get:
\begin{align*}
	\frac{\partial}{\partial X} \left( s^\top \cdot \mathrm{inv}(F+\mathrm{diag}(c\cdot (X\cdot Y+Y^\top \cdot X^\top )\cdot v)-c\cdot (X\cdot Y+Y^\top \cdot X^\top ))\cdot s \right) =
	\\\quad\quad
	-c\cdot t_{5}\odot t_{5}\cdot \v+1_k^\top
	-c\cdot \v+1\cdot ((t_{5}^\top \odot t_{5}^\top)\cdot Y^\top )
	+2c\cdot t_{5}\cdot (t_{5}^\top\cdot Y^\top )
\end{align*}
where
\begin{itemize}
	\item $t_5 = (\m+I + \m+L + \LXC)^{-1} \cdot s$
\end{itemize}

\paragraph{Simplification 5.}
Observing that $t_5 = \zXC$, we obtain at our final gradient.

Then we get:
\begin{align*}
	\nabla_{\m+X} \left( s^\top (\m+I + \m+L + \LXC)^{-1} s \right) =
	\frac{CW}{2n} (2\cdot \zXC \cdot (\zXC^\top\cdot \m+Y^\top )
	- \zXC\odot \zXC\cdot \v+1_k^\top
	- \v+1 ((\zXC^\top \odot \zXC^\top)\cdot \m+Y^\top ))
\end{align*}

\subsubsection{The gradient is Lipschitz}

We need to show that
$\norm{\nabla_{\m+X} f(\m+X_1) - \nabla_{\m+X} f(\v+X_2)}_F \leq
	L \norm{\m+X_1 - \m+X_2}_F$
	for all $\m+X_1,\m+X_2\in Q$.

Using the previously derived gradient and the triangle inequality, we get that
\begin{align*}
	\norm{\nabla_{\m+X} f(\m+X_1) - \nabla_{\m+X} f(\v+X_2)}_F
	&\leq \frac{CW}{2n} \Huge[
			\norm{2\cdot \zXCO \cdot (\zXCO^\top\cdot \m+Y^\top )
				- 2\cdot \zXCT \cdot (\zXCT^\top\cdot \m+Y^\top )}_F \\
		&\quad
			+ \norm{\zXCO\odot \zXCO\cdot \v+1_k^\top
				- \zXCT \odot \zXCT \cdot \v+1_k^\top}_F \\
		&\quad
			+ \norm{\v+1_n ((\zXCO^\top \odot \zXCO^\top)\cdot \m+Y^\top )
				- \v+1_n ((\zXCT^\top \odot \zXCT^\top)\cdot \m+Y^\top ))}_F
		\Huge] \\
	&= \frac{CW}{2n} \huge[
			\norm{2\cdot (\zXCO \cdot \zXCO^\top
					- \zXCT \cdot \zXCT^\top) \m+Y^\top )}_F \\
		&\quad
			+ \norm{(\zXCO\odot \zXCO
				- \zXCT \odot \zXCT) \v+1_k^\top)}_F \\
		&\quad
			+ \norm{\v+1_n (\zXCO^\top \odot \zXCO^\top
				- \zXCT^\top \odot \zXCT^\top)\cdot \m+Y^\top}_F
		\huge]
\end{align*}
Next, we will now bound each of these terms invidually.

We start by making a crucial observation about the difference of $\zXCO$ and
$\zXCT$, where we use Lemma~\ref{lem:eigenvalues-M} in the final step:
\begin{align*}
	&\norm{\zXCO - \zXCT}_F \\
	&= \norm{(\m+I + \m+L + \LXCO)^{-1} \v+s
					- (\m+I + \m+L + \LXCT)^{-1} \v+s}_F \\
	&\leq \norm{\v+s}_2
		\cdot \norm{(\m+I + \m+L + \LXCO)^{-1}
					- (\m+I + \m+L + \LXCT)^{-1}}_F \\
	&\leq \norm{\v+s}_2
		\cdot
   		\norm{(\m+I + \m+L + \LXCO)^{-1}}_2
		\cdot
   		\norm{(\m+I + \m+L + \LXCT)^{-1}}_2
		\cdot
   		\norm{(\m+I + \m+L + \LXCO)
			- (\m+I + \m+L + \LXCT)}_F \\
	&= 2 \norm{\v+s}_2
		\cdot
   		\norm{(\m+I + \m+L + \LXCO)^{-1}}_2
		\cdot
   		\norm{(\m+I + \m+L + \LXCT)^{-1}}_2
		\cdot
   		\norm{ (\m+X_1 - \m+X_2) \m+Y}_F \\
	&\leq 2 \norm{\v+s}_2
		\cdot
		\norm{\m+Y}_2
		\cdot
   		\norm{ \m+X_1 - \m+X_2}_F.
\end{align*}

Next, observe that for all $\m+X$ we have that $\norm{\zXC}_2 \leq \sqrt{n}$
since $\norm{\zXC}_2 \leq \norm{(\m+I + \m+L + \LXCO)^{-1}}_2 \cdot
\norm{\v+s}_2 \leq \norm{\v+s}_2 \leq \sqrt{n}$, where we use
Lemma~\ref{lem:eigenvalues-M} and the fact that the entries in $\v+s$ are in
$[-1,1]$. In particular, this implies that $\norm{\zXCO}_2 + \norm{\zXCT}_2 \leq
2\sqrt{n}$. 

Now using Lemma~\ref{lem:difference-rank-one} together with
$\norm{\zXCO}_2 + \norm{\zXCT}_2 \leq 2\sqrt{n}$ and our inequalty from above,
we obtain that
\begin{align*}
	&\norm{2\cdot (\zXCO \cdot \zXCO^\top - \zXCT \cdot
					\zXCT^\top) \m+Y^\top}_F \\
	&\leq 2 \norm{\m+Y}_2
		\cdot \norm{ \zXCO \cdot \zXCO^\top - \zXCT \cdot \zXCT^\top }_F \\
	&\leq 4\sqrt{n} \cdot \norm{\m+Y}_2
		\cdot \norm{\zXCO - \zXCT}_F \\
	&\leq 8 \sqrt{n}
		\cdot \norm{\v+s}_2
		\cdot
		\norm{\m+Y}_2^2
		\cdot
   		\norm{ \m+X_1 - \m+X_2}_F.
\end{align*}

Next, we show a fact about $\zXCO\odot \zXCO - \zXCT \odot \zXCT$. Using
Lemma~\ref{lem:difference-hadamard} and our inequality from above, we get that
\begin{align*}
	&\norm{\zXCO\odot \zXCO - \zXCT \odot \zXCT}_2 \\
	&\leq 2 \norm{\zXCO - \zXCT}_2 \\
	&\leq 4 \norm{\v+s}_2
		\cdot
		\norm{\m+Y}_2
		\cdot
   		\norm{ \m+X_1 - \m+X_2}_F.
\end{align*}
Using the inequality from above, we get that
\begin{align*}
	&\norm{(\zXCO\odot \zXCO - \zXCT \odot \zXCT) \v+1_k^\top}_F \\
	&= \sqrt{k} \cdot \norm{\zXCO\odot \zXCO - \zXCT \odot \zXCT}_2 \\
	&\leq 4 \sqrt{k} \cdot \norm{\v+s}_2
		\cdot \norm{\m+Y}_2
		\cdot \norm{\m+X_1 - \m+X_2}_F.
\end{align*}
Furthermore, we again use the inequality from above to obtain that
\begin{align*}
	&\norm{\v+1_n (\zXCO^\top \odot \zXCO^\top
				- \zXCT^\top \odot \zXCT^\top)\cdot \m+Y^\top}_F \\
	&\leq \norm{\m+Y}_2 \cdot
		\norm{\v+1_n (\zXCO^\top \odot \zXCO^\top
				- \zXCT^\top \odot \zXCT^\top)}_F \\
	&= \norm{\m+Y}_2 \cdot \sqrt{n} \cdot
		\norm{\zXCO^\top \odot \zXCO^\top
				- \zXCT^\top \odot \zXCT^\top}_2 \\
	&\leq 4 \sqrt{n} \cdot \norm{\v+s}_2
		\cdot \norm{\m+Y}_2^2
		\cdot \norm{\m+X_1 - \m+X_2}_F.
\end{align*}

By combining the results from above and using our assumption that $k\leq n$, we
get that the gradient is Lipschitz with
$L=\frac{8CW}{\sqrt{n}} \cdot \norm{\v+s}_2 \cdot \norm{\m+Y}_2^2$.

\subsubsection{Approximate gradient}
We use Proposition~\ref{prop:fast-opinions} with
$\varepsilon'=\frac{1}{8}\cdot
	\frac{\min\{ \varepsilon, \sqrt{\varepsilon} \} \cdot \sqrt{n}}
		{(1+CW)(1+\norm{\m+Y}_F)}$
to obtain $\tzXC$. Note that in the numerator we use
$\min\{\varepsilon,\sqrt{\varepsilon}\}$ since it might be that
$\sqrt{\varepsilon}>\varepsilon$ for
$\varepsilon<1$; in denominator we use the terms $1+CW$ and $1+\norm{\m+Y}_F$
because it is possible that $CW<1$ and also $\norm{\m+Y}_F<1$.
Observe that with this choice of $\varepsilon'$, Proposition~\ref{prop:fast-opinions}
guarantees a running time of $\tO((mk + nk^2 + k^3)\lg(W/\varepsilon))$.

Now we consider the approximate gradient
\begin{align*}
	\widetilde{\nabla}_{\m+X} f(\m+X, C) =
		\frac{CW}{2n} (2\cdot \tzXCT \cdot \tzXCT^\top\cdot \m+Y^\top
		- (\tzXCT\odot \tzXCT) \cdot \v+1_k^\top
		- \v+1_n (\tzXCT^\top \odot \tzXCT^\top)\cdot \m+Y^\top.
\end{align*}

\sninline{Somehow the spacing in the following equations is weird. Dunno why.
Also, we have \\huge does not seem to work for parentheses, so we have to fix
that.}

Let $\nabla_{\m+X} f(\m+X)$ denote the exact gradient. Then we get that
\begin{align*}
	\norm{\widetilde{\nabla}_{\m+X} f(\m+X) - \nabla_{\m+X} f(\m+X)}_F
	&\leq \frac{CW}{2n} \Huge[
			\norm{ 2\cdot (\tzXCT \cdot \tzXCT^\top - \zXC \cdot \zXC^\top)
					\cdot\m+Y^\top }_F \\
		&\quad + \norm{(\tzXCT\odot \tzXCT - \zXC\odot\zXC) \cdot \v+1_k^\top}_F \\
		&\quad + \norm{\v+1_n (\tzXCT^\top\odot\tzXCT^\top-\zXC^\top\odot\zXC^\top)\cdot \m+Y^\top}_F
		\Huge] \\
	&\leq \frac{CW}{2n} \Huge[
			2 \norm{\m+Y}_F \cdot \norm{ \tzXCT \cdot \tzXCT^\top - \zXC \cdot \zXC^\top}_2 \\
		&\quad + \sqrt{k} \cdot \norm{\tzXCT\odot \tzXCT - \zXC\odot\zXC}_2 \\
		&\quad + \sqrt{n} \norm{\m+Y}_F \cdot
		\norm{\tzXCT^\top\odot\tzXCT^\top-\zXC^\top\odot\zXC^\top}_2
		\Huge].
\end{align*}

Now let $\v+e$ be the error vector such that $\tzXC=\zXC + \v+e$ and recall that
by Proposition~\ref{prop:fast-opinions} we have that $\norm{\tzXC-\zXC}_2 =
\norm{\v+e}_2 \leq \varepsilon'$. Then
\begin{align*}
	\norm{ \tzXCT \cdot \tzXCT^\top - \zXC \cdot \zXC^\top}_2
	&= \norm{ (\zXC + \v+e)\cdot (\zXC + \v+e)^\top - \zXC \cdot \zXC^\top }_2 \\
	&\leq \norm{ \v+e \cdot \v+e^\top}_2 + 2\norm{\zXC \cdot \v+e^\top}_2 \\
	&\leq \norm{\v+e}_2^2 + 2\norm{\zXC}_2 \cdot \norm{\v+e}_2 \\
	&\leq \norm{\v+e}_2^2 + 2\sqrt{n} \cdot \norm{\v+e}_2 \\
	&\leq \frac{3}{8}\cdot \frac{\varepsilon \cdot n}
								{(1+CW) \cdot (1+\norm{\m+Y}_F)}.
\end{align*}

Next, using Lemma~\ref{lem:difference-hadamard} we get that
\begin{align*}
	\norm{\tzXCT\odot \tzXCT - \zXC\odot\zXC}_2
	&\leq 2 \norm{\tzXCT - \zXC}_2 \\
	&\leq 2 \norm{\v+e} \\
	&\leq \frac{1}{4} \cdot \frac{\min\{\varepsilon, \sqrt{\varepsilon}\} \cdot \sqrt{n}}
								 {(1+CW)\cdot(1+\norm{\m+Y}_F)}.
\end{align*}

\sninline{Spacing is weird again.}

Now continuing our inequalities from above and using that $k\leq n$, we obtain
that
\begin{align*}
	\norm{\widetilde{\nabla}_{\m+X} f(\m+X) - \nabla_{\m+X} f(\m+X)}_F
	&\leq \frac{CW}{2n} \Huge[
			2 \norm{\m+Y}_F \cdot
				\norm{ \tzXCT \cdot \tzXCT^\top - \zXC \cdot \zXC^\top}_2
					\\
		&\quad + \sqrt{k} \cdot \norm{\tzXCT\odot \tzXCT - \zXC\odot\zXC}_2 \\
		&\quad + \sqrt{n} \norm{\m+Y}_F \cdot
		\norm{\tzXCT^\top\odot\tzXCT^\top-\zXC^\top\odot\zXC^\top}_2 \Huge] \\
	&\leq \frac{CW}{2n} \left(
			\frac{6}{8}\cdot \frac{\varepsilon \cdot n}{1+CW}
			+ \frac{1}{4}\cdot \frac{\min\{\varepsilon, \sqrt{\varepsilon}\} \cdot n}{(1+CW)(1+\norm{\m+Y}_F)}
			+ \frac{1}{4}\cdot \frac{\min\{\varepsilon, \sqrt{\varepsilon}\} \cdot n}{1+CW}
		\right) \\
	&\leq \frac{5}{8} \varepsilon.
\end{align*}

\subsection{Proof of Theorem~\ref{thm:approximation-guarantee}}

\subsubsection{Recap of d'Aspremont's Algorithm}
We start by formally stating the result by D'Aspremont~\cite{d2008smooth} for
gradient descent with approximate gradient.

Suppose we wish to solve the convex optimization problem $\min_{x\in Q} f(x)$,
where $f$ is a convex function mapping to $\mathbb{R}$ and $Q$ is a convex
set of feasible solutions.  D'Aspremont~\cite{d2008smooth} gave an algorithm
which approximately solves this problem using gradient descent, where the
gradient contains some amount of noise.  This algorithm is based on a method by
Nesterov~\cite{nesterov1983method}.

We state the pseudo\-code in Algorithm~\ref{alg:dAspremont}, where we let
$\widetilde{T}_Q = \argmin_{y\in Q}
	\left\{ \langle \widetilde{\nabla}_{\m+X} f(x), y - x\rangle
		+ \frac{L}{2} \norm{y-x}^2 \right\}$
and $d(x)$ is a prox-function for the set $Q$, i.e., $d$ is continuous and
strongly convex with parameter $\kappa$.  We state the guarantees for the
algorithm in the lemma below, where $\norm{\cdot}^*$ is the dual norm of
$\norm{\cdot}$.

{\SetAlgoNoLine
 \LinesNumbered
 \DontPrintSemicolon
 \SetAlgoNoEnd
\begin{algorithm2e}[t]
$x_0 \gets \argmin_{x\in Q} d(x)$\;
\For{$k=1,\dots,T$}{
	Compute the approximate gradient $\widetilde{\nabla}_{\m+X} f(x_k)$\;
	$y_k \gets \widetilde{T}(x_k)$\;
	$z_k \gets \argmin_{x\in Q}
			\{\frac{L}{\kappa} d(x) 
			+ \sum_{i=0}^k \alpha_i[f(x_i) + \langle \widetilde{\nabla}_{\m+X} f(x_i), x-x_i\rangle\}$\;
	$A_k \gets \sum_{i=0}^k \alpha_i$\;
	$\tau_k \gets \alpha_{k+1}/A_{k+1}$\;
	$x_{k+1} \gets \tau_k z_k + (1-\tau_k)y_k$\;
}
\caption{Gradient-descent algorithm with noisy gradient}
\label{alg:dAspremont}
\end{algorithm2e}
}

\begin{lemma}[d'Aspremont~\cite{d2008smooth}]
\label{lem:dAspremont}
	Let $L,\kappa,\varepsilon > 0$.
	Suppose the following conditions hold:
	\begin{enumerate}
		\item $\norm{\nabla_{\m+X} f(x) - \nabla_{\m+X} f(y)}^*
				\leq L \cdot \norm{x - y}$ for all $x,y\in Q$,
		\item $\abs{\langle \widetilde{\nabla}_{\m+X} f(x) - \nabla_{\m+X} f(x),
		   				y - z \rangle} \leq \varepsilon$ for all $x,y,z\in Q$,
		\item $x_0 = \argmin_{x\in Q} d(x)$ and $d(x_0)=0$,
		\item $d(x) \geq \frac{\kappa}{2} \norm{x - x_0}^2$,
		\item $(\alpha_k)_k$ is a sequence such that $0<\alpha_0\leq1$ and
			$\alpha^2_k \leq A_k$ for all $k\geq 0$.
	\end{enumerate}
	Then Algorithm~\ref{alg:dAspremont} satisfies 
	$f(y_k) - f(x^*) \leq \frac{Ld(x^*)}{A_k\kappa} + 3\varepsilon$
	for all $k\leq T$, where $x^* = \argmin_{x\in Q} f(x)$ is the minimizer of
	the optimization problem.
\end{lemma}

\subsubsection{Proof of Theorem~\ref{thm:approximation-guarantee}}

Recall that we let $Q\subseteq\mathbb{R}^{n\times k}$ denote the convex subset
of feasible solutions for Problem~\ref{problem:min-dpi}. We prove the following
useful lemma, which shows how certain functions can be optimized over our set of
constraints $Q$.

We start by proving a lemma that will be useful later when implementing
d'Aspremont's algorithm.
\begin{lemma}
\label{lem:opt-constraints}
	Let $\beta>0$ and let $\m+B,\m+X\in\mathbb{R}^{n\times k}$. Then the
	minimizer
	\begin{align*}
		\m+X^* =
			\argmin_{\m+U\in Q} \left\{
				 \frac{\beta}{2} \norm{\m+U - \m+X}_F^2
				 + \langle \m+B, \m+U \rangle_F
			\right\}
			\in\mathbb{R}^{n\times k}
	\end{align*}
	satisfies
	\begin{align*}
		\m+X^*_i
		= \ProjQ{\m+X_i - \frac{1}{\beta}\m+B_i - \mu_i^* \v+1},
	\end{align*}
	for all $i$. Furthermore, $\m+X^*$ can be computed in time $\bigO(nk)$.
\end{lemma}
\begin{proof}
	We have that
	\begin{align*}
		\m+X^* 
		&=\argmin_{\m+U\in Q} \left\{
				 \frac{\beta}{2} \norm{\m+U - \m+X}_F^2
				 + \langle \m+B, \m+U \rangle_F
			\right\} \\
		&=\argmin_{\m+U\in \mathbb{R}^{n\times k}} \left\{
				 \delta_Q(\m+U)
				 + \frac{\beta}{2} \norm{\m+U - \m+X}_F^2
				 + \langle \m+B, \m+U \rangle_F
			\right\} \\
		&= \argmin_{\m+U\in \mathbb{R}^{n\times k}} \left\{
				 \delta_Q(\m+U)
				 + \frac{\beta}{2} \left(
						 \langle \m+U, \m+U \rangle_F
						 - 2 \langle \m+U, \m+X \rangle_F
						 + \langle \m+X, \m+X \rangle_F
						 \right)
				 + \langle \m+B, \m+U \rangle_F
			\right\} \\
		&= \argmin_{\m+U\in \mathbb{R}^{n\times k}} \left\{
				 \delta_Q(\m+U)
				 + \frac{\beta}{2} \left(
						 \langle \m+U, \m+U \rangle_F
						 - 2 \langle \m+U, \m+X - \frac{1}{\beta} \m+B \rangle_F
						 + \langle \m+X, \m+X \rangle_F
						 \right)
			\right\} \\
		&= \argmin_{\m+U\in \mathbb{R}^{n\times k}} \left\{
				 \delta_Q(\m+U)
				 + \frac{\beta}{2} \left(
						 \langle \m+U, \m+U \rangle_F
						 - 2 \langle \m+U, \m+X - \frac{1}{\beta} \m+B \rangle_F
						 + \langle \m+X, \m+X \rangle_F
						 - 2\langle \m+X, \frac{1}{\beta} \m+B \rangle_F
						 + \langle \frac{1}{\beta} \m+B, \frac{1}{\beta} \m+B \rangle_F
						 \right)
			\right\} \\
		&= \argmin_{\m+U\in \mathbb{R}^{n\times k}} \left\{
				 \delta_Q(\m+U)
				 + \frac{\beta}{2} \left(
						 \langle \m+U, \m+U \rangle_F
						 - 2 \langle \m+U, \m+X - \frac{1}{\beta} \m+B \rangle_F
						 + \langle \m+X-\frac{1}{\beta} \m+B, \m+X-\frac{1}{\beta} \m+B \rangle_F
						 \right)
			\right\} \\
		&= \argmin_{\m+U\in \mathbb{R}^{n\times k}} \left\{
				 \delta_Q(\m+U)
				 + \frac{\beta}{2} 
					\norm{\m+U - \left(\m+X - \frac{1}{\beta}\m+B\right)}_F
			\right\} \\
		&= \argmin_{\m+U\in \mathbb{R}^{n\times k}} \left\{
				 \delta_Q(\m+U)
				 + \frac{1}{2}
					\norm{\m+U - \left(\m+X - \frac{1}{\beta}\m+B\right)}_F
			\right\} \\
		&= \prox_{\delta_Q}\left(\m+X - \frac{1}{\beta}\m+B\right),
	\end{align*}
	where in the penultimate step we use that $\delta_Q(\m+U)$ only takes the
	values $0$ and~$\infty$ and therefore we did not have to rescale $\delta_Q$.

	Next, let us consider the projection on our set of $Q$, which is given by all
	$\m+X'$ such that $\norm{\m+X_i'}_{1,1} = 1$ for all $i$ and
	$\m+X^{(L)}\! \leq \m+X' \leq \m+X^{(U)}\!$.
	As discussed before, the first constraint is equivalent to the hyperplane
	constraint $\langle \m+X_i', \v+1 \rangle = 1$ for all $i$, since
	$\m+X_i'\in[0,1]^n$ and the second constraint can be rewritten as a sequence
	of box constraints $\m+X^{(L)}_i \leq \m+X_i' \leq \m+X^{(U)}_i$ for all $i$.

	From the previous paragraph we get that it suffices if
	we project on the feasible set for each row of $\m+X^*$ individually.
	Since in the definition of
	$\prox_{\delta_Q}\left(\m+X - \frac{1}{\beta}\m+B\right)$ we considered
	the Frobenius norm, we can find the minimizer for each row individually.
	More concretely, the rows $\m+X_i$ of $\m+X^*$ are given by 
	$\prox_{\delta_{Q_i}}(\m+X_i - \frac{1}{\beta}\m+B_i)$, where by $Q_i$ we
	denote the set of constraints $\norm{\m+X_i'}_{1,1} = 1$
	and $\m+X^{(L)}_i \leq \m+X_i' \leq \m+X^{(U)}_i$.

	We note that this is the same \sninline{double-check} as computing the
	orthogonal projection of $\m+X_i - \frac{1}{\beta}\m+B_i$ onto $Q_i$. This
	orthogonal projection can be computed by the algorithm of
	Kiwiel~\cite{kiwiel2008breakpoint}  in time $\bigO(k)$.  Since we have to
	run this procedure for each of the $n$~rows of $\m+X^*$, we can compute
	$\m+X^*$ in time $\bigO(nk)$.
\end{proof}

To prove the theorem, we need to argue that we can apply
Lemma~\ref{lem:dAspremont} to {\ouralgo} and we also have
to analyze the running time.

First, we note that {\ouralgo} is an implementation of
Algorithm~\ref{alg:dAspremont} with the following parameters.

We set $\kappa=1$ and $d(\m+X') = \frac{1}{2} \norm{\m+X' - \m+X^{(0)}}_F^2$;
note that this trivially implies
$d(\m+X') \geq \frac{1}{2} \kappa \norm{\m+X' - \m+X^{(0)}}_F^2$
and also $d(\m+X^{(0)}) = 0$.
Furthermore, we set 
$\alpha_T = \frac{T+1}{2}$, which satisfies the conditions of
Lemma~\ref{lem:dAspremont} as pointed out in~\cite{d2008smooth}.
It remains to show that $\m+Z^{(T)}$ and $\widetilde{T}(\m+X^{(T)})$ are implemented
in accordance with Lemma~\ref{lem:dAspremont}. To this end, observe that
\begin{align*}
	\m+Z^{(T)}
	&= \argmin_{\m+U \in Q} \left\{
	  		 \frac{L}{\kappa} d(\m+U)
			 + \sum_{t=0}^T \alpha_t [f(\m+X^{(t)}) + \langle \widetilde\nabla_{\m+X} f(\m+X^{(t)}), \m+U - \m+X^{(t)} \rangle_F]
   		\right\} \\
	&= \argmin_{\m+U\in Q} \left\{
	  		 L \norm{\m+U - \m+X^{(0)}}_F^2
			 + \langle \sum_{t=0}^T \alpha_t \widetilde\nabla_{\m+X} f(\m+X^{(t)}), \m+U \rangle_F
   		\right\}.
\end{align*}
Hence, we can compute $\m+Z^{(T)}$ using Lemma~\ref{lem:opt-constraints} with
parameters $\m+X=\m+X^{(0)}$,
$\m+B = \sum_{t=0}^T \alpha_t \widetilde\nabla_{\m+X} f(\m+X^{(t)})$, and
$\beta=2L$.
Additionally, for the projection on our set of constraints we obtain that
\begin{align*}
	\widetilde{T}_Q(\m+X^{(T)})
	&= \argmin_{\m+U\in Q}\left\{
			\langle \widetilde{\nabla}_{\m+X} f(\m+X^{(T)}), \m+U - \m+X^{(T)}\rangle_F
			+ \frac{L}{2} \norm{\m+U-\m+X^{(T)}}^2_F
		\right\} \\
	&= \argmin_{\m+U\in Q}\left\{
			\langle \widetilde{\nabla}_{\m+X} f(\m+X^{(T)}), \m+U\rangle_F
			+ \frac{L}{2} \norm{\m+U-\m+X^{(T)}}^2_F
		\right\}.
\end{align*}
Hence, we can compute $\widetilde{T}_Q(\m+X^{(T)})$ using
Lemma~\ref{lem:opt-constraints} with parameters $\m+X=\m+X^{(T)}$,
$\m+B = \widetilde\nabla_{\m+X} f(\m+X^{(T)})$, and $\beta=L$.

Observe that in {\ouralgo}, $\m+V^{(T)}$ corresponds to
$\widetilde{T}_Q(\m+X^{(T)})$ and $\m+W^{(T)}$ corresponds to $\m+Z^{(T)}$.

This implies that we can we can use Lemma~\ref{lem:dAspremont} if we can also
show that the gradient of our objective function is Lipschitz and that the error
for our gradient is small.

First, using that we are working with the Frobenius norm
which is self-dual, we can apply Proposition~\ref{prop:gradient} to obtain the
gradient is $L$-smooth for 
$L=\frac{8CW}{\sqrt{n}} \cdot \norm{\v+s}_2 \cdot \norm{\m+Y}_2^2$.

Second, we obtain our approximate gradient result as follows. Observe that our
feasible space only contains row-stochastic matrices with entries in $[0,1]$.
Thus, we get that for all $\m+X_2,\m+X_3\in Q$ we have that
$\norm{\m+X_2-\m+X_3}_F^2 \leq 2n$, since in each row the difference can be at
most~$2$. Thus, if we compute $\widetilde{\nabla}_{\m+X} f(\m+X_1)$ using
Lemma~\ref{prop:gradient} with $\varepsilon'=\frac{\varepsilon}{\sqrt{2n}}$ and
using the Cauchy--Schwarz inequality we get that
\begin{align*}
	&\abs{\langle \widetilde{\nabla}_{\m+X} f(\m+X_1) - \nabla_{\m+X} f(\m+X_2),
		   				\m+X_2 - \m+X_2 \rangle_F} \\
	&\leq \norm{\widetilde{\nabla}_{\m+X} f(\m+X_1) - \nabla_{\m+X} f(\m+X_2)}_F
		\cdot \norm{\m+X_2 - \m+X_2}_F \\
	&\leq \varepsilon' \cdot \sqrt{2n} \\
	&\leq \varepsilon
\end{align*}
for all $\m+X_1,\m+X_2,\m+X_3\in Q$.

As pointed out in~\cite{d2008smooth}, if we compute the gradient with precision
$\varepsilon/6$ then we obtain a solution with additive error at most
$\varepsilon$ after $O\left(\frac{Ld(\m+X^*)}{\varepsilon}\right)$ iterations of the
algorithm, where $\m+X^*$ is the optimal solution. Since $\m+X^*$ and
$\m+X^{(0)}$ are row-stochastic matrices with entries in $[0,1]$, we get that
$d(\m+X^*) \leq n$. Now using our previous bound on $L$ and
$\norm{\v+s}_2\leq \sqrt{n}$, as well as
$\norm{\m+Y}_2^2 \leq \norm{\m+Y}_F^2 \leq k$ since $\m+Y$ is row stochastic
with $k$~rows, we get that the number of iterations is bounded by
$O\left(\sqrt{\frac{CWkn}{\varepsilon}}\right)$.

Next, observe that in each iteration we spend expected time
$\tO((mk + nk^2 + k^3)\lg(W/\varepsilon))$
to compute the approximate gradient by Proposition~\ref{prop:gradient}. 
Computing the matrices $\m+V^{(T)}$ and $\m+W^{(T)}$ takes time $\bigO(nk\lg k)$ by
Lemma~\ref{lem:opt-constraints}. Thus, the total time of each iteration is
bounded by $\tO((mk + nk^2 + k^3)\lg(W/\varepsilon))$.

Combing these results we obtain a total expected running time of
$\tO\left(\sqrt{\frac{CWkn}{\varepsilon}}(mk + nk^2 + k^3)\lg(W/\varepsilon)\right)$.

\end{document}